\documentclass[a4paper,UKenglish,10pt]{article}
\usepackage[DIV9]{typearea}
\usepackage[UKenglish]{babel}

\usepackage{lmodern}

\usepackage{textcomp}           
\usepackage[T1]{fontenc}
\usepackage[utf8]{inputenc}
\usepackage{tikz}

\usepackage{stmaryrd}
\usepackage{amssymb}
\usepackage{paralist}
\usepackage{ifthen}
\usepackage{amsmath}

\usepackage{color}

\usetikzlibrary{positioning,backgrounds,calc,decorations,arrows}
\usetikzlibrary{%
  decorations.pathmorphing,%
  decorations.pathreplacing,%
  decorations.markings,%
  decorations.shapes,%
  decorations.footprints,%
  decorations.fractals,%
  decorations.text%
}%
\tikzset{%
  >=latex',%
  single step/.style={>=to,->},%
  etc/.style={edge from parent/.style={-,dotted,thick,draw}},%
  shorten/.style={shorten >=#1, shorten <=#1},%
  level distance=1cm,%
  sibling distance = 10mm,%
  edge from parent/.style={->,draw},%
  fun/.style={>=to,->},%
  node name/.style args={#1:#2}{label={[black!70]#1:\small$#2$}}%
  }%

\newcommand\wsuc\oh

\newcommand\prefix[2]{#1|_{#2}}

\newcommand\nodeAtPos[2]{\symb{node}_{#1}(#2)}

\newcommand\predAcy[2]{\symb{Pre}^a_{#1}(#2)}

\newcommand\nn[2]{#2^{[#1]}}

\newcommand\lebot{\le_\bot}

\newcommand\lebotg{\lebot^\calG}

\newcommand\concat{\cdot}
\newcommand\Concat{\prod}

\newcommand\glb{\sqcap}
\newcommand\Lub{\bigsqcup}
\newcommand\Glb{\bigsqcap}

\newcommand\similar[2]{\symb{sim}(#1,#2)}
\newcommand\similara[2]{\abssim{\trunca{}{}}(#1,#2)}

\newcommand\abssim[1]{\symb{sim}_{#1}}
\newcommand\abssimp[1]{\symb{sim}'_{#1}}

\newcommand\dd{\mathbf{d}}
\newcommand\dda{\dd_{\trunca{}{}}}

\newcommand\trunc[2]{#1\mathclose{\wr}#2}
\newcommand\trunca[2]{#1{\upharpoonright}#2}
\newcommand\truncl[2]{#1{\upharpoonright}#2}

\newcommand\fNodes[2]{N^{#1}_{=#2}}
\newcommand\tNodes[2]{N^{#1}_{<#2}}

\newcommand\nats{{\mathbb N}}
\newcommand\natsp{\nats^+}
\newcommand\reals{{\mathbb R}}
\newcommand\realsp{\reals^+}
\newcommand\realsnn{\reals^+_0}
\renewcommand\epsilon{\varepsilon}

\newcommand\id{\symb{id}}
\newcommand\abs[1]{\left\lvert#1\right\rvert}

\newcommand\oh\widehat
\newcommand\ul\underline
\newcommand\ol\overline
\newcommand\ot\widetilde

\newcommand\canon[1]{\calC(#1)}

\newcommand\unrav[1]{\calU\left(#1\right)}
\newcommand\quotient[2]{#1/\!\raisebox{-.65ex}{\ensuremath{#2}}}
\newcommand\eqc[2]{[#1]_{#2}}

\newcommand\isom{\cong}

\newcommand\calC{\mathcal{C}}

\newcommand\calG{\mathcal{G}}

\newcommand\calN{\mathcal{N}}

\newcommand\calP{\mathcal{P}}

\newcommand\calR{\mathcal{R}}

\newcommand\calT{\mathcal{T}}
\newcommand\calU{\mathcal{U}}
\newcommand\calV{\mathcal{V}}

\newcommand\calPos{\calP}

\mathchardef\mhyphen="2D

\newcommand\fcolon{\colon\,}
\newcommand\funto{\rightarrow}
\newcommand\limto{\rightarrow}
\newcommand\len[1]{\left\lvert #1 \right\rvert}

\newcommand\seq[1]{\langle #1 \rangle}
\newcommand\emptyseq{\seq{}}
\newcommand\srank[1]{\symb{ar}(#1)}

\newcommand\rank[2]{\symb{ar}_{#1}(#2)}
\newcommand\nin{\not\in}

\newcommand\prs{p}
\newcommand\mrs{m}

\newcommand\isoto{\stackrel{}{\ot\rightarrow}}

\newcommand\homto{\rightarrow}

\newcommand{\setcom}[2]{\set{#1\left\vert\vphantom{#1}\,#2\right.}}
\newcommand{\set}[1]{\left\{#1\right\}}

\newcommand\pos[1]{\calPos(#1)}
\newcommand\posBound[2]{\calPos_{\le #1}(#2)}

\newcommand\nodePos[2]{\calPos_{#1}(#2)}
\newcommand\nodePosAcy[2]{\calPos^a_{#1}(#2)}
\newcommand\nodePosMin[2]{\calPos^m_{#1}(#2)}

\newcommand\depth[2]{\symb{depth}_{#1}(#2)}

\newcommand\iptgraphs[1][\Sigma]{\calG^\infty(#1_\bot)}

\newcommand\itgraphs[1][\Sigma]{\calG^\infty(#1)}

\newcommand\ictgraphs[1][\Sigma]{\calG^\infty_\calC(#1)}

\newcommand\ipctgraphs[1][\Sigma]{\calG^\infty_\calC(#1_\bot)}

\newcommand\iterms[1][\Sigma]{%
\calT^\infty(#1)
}
\newcommand\ipterms[1][\Sigma]{%
\calT^\infty(#1_\bot)
}

\newcommand\fterms[1][\Sigma]{%
\calT(#1)
}

\newcommand\vterms[1][\Sigma]{%
\def\termsSig{#1}%
\vtermsI%
}
\newcommand\vtermsI[1][\calV]{%
\def\termsVar{#1}%
\calT(\termsSig,\termsVar)%
}

\newcommand\viterms[1][\Sigma]{%
\def\itermsSig{#1}%
\vitermsI%
}
\newcommand\vitermsI[1][\calV]{%
\def\itermsVar{#1}%
\calT^\infty(\itermsSig,\itermsVar)%
}

\newcommand\symb[1]{\mathsf{#1}}
\newcommand\glab{\symb{lab}}
\newcommand\gsuc{\symb{suc}}

\newcommand\atPos[2]{#1|_{#2}}
\newcommand\substAtPos[3]{#1[#3]_{#2}}

\newcommand\lhs[1]{#1_{l}}
\newcommand\rhs[1]{#1_{r}}

\newcommand\subgraph[2]{#1|_{#2}}

\def\nothing{}
\let\oldTo\to

\newcommand\finleftright{\leftrightarrow}
\newcommand\finright{\oldTo}
\newcommand\finleft{\leftarrow}
\newcommand\weakright{\hookrightarrow}
\newcommand\weakleft{\hookleftarrow}
\newcommand\strongright{\twoheadrightarrow}
\newcommand\strongleft{\twoheadleftarrow}
\newcommand\wmrsright{\mathrel{\hookrightarrow^{\hspace{-10pt}m}\hspace{2pt}}}
\newcommand\wmrsleft{\mathrel{\hookleftarrow^{\hspace{-7pt}m}\hspace{0pt}}}
\newcommand\wprsright{\mathrel{\hookrightarrow^{\hspace{-8pt}p}\hspace{4pt}}}
\newcommand\wprsleft{\mathrel{\hookleftarrow^{\hspace{-5pt}p}\hspace{0pt}}}

\newcommand\mrsright{\mathrel{\twoheadrightarrow^{\hspace{-10pt}m}\hspace{2pt}}}
\newcommand\mrsleft{\mathrel{\twoheadleftarrow^{\hspace{-7pt}m}\hspace{0pt}}}
\newcommand\prsright{\mathrel{\twoheadrightarrow^{\hspace{-8pt}p}\hspace{4pt}}}
\newcommand\prsleft{\mathrel{\twoheadleftarrow^{\hspace{-5pt}p}\hspace{0pt}}}

\newcommand{\RewArr}[2] {
  \RewStmt{#1}{\nothing}{#2}
}

\newcommand{\RewStmt}[3] {
  \def\RewArrArr{#1}
  \def\RewArrRhs{#2}
  \def\RewArrIter{#3}
  \RewArrI
}

\makeatletter \newcommand{\RewArrI}[1][\nothing] { \def\RewArrCxt{#1}
  \ifthenelse{\equal{\RewArrArr}{\oldTo} \OR
    \equal{\RewArrArr}{\finright} \OR
    \equal{\RewArrArr}{\finleftright} \OR
    \equal{\RewArrArr}{\prsright} \OR
    \equal{\RewArrArr}{\mrsright} \OR
    \equal{\RewArrArr}{\wprsright} \OR
    \equal{\RewArrArr}{\wmrsright} \OR
    \equal{\RewArrArr}{\strongright} \OR
    \equal{\RewArrArr}{\weakright}} {
    \RewArrArr\ifthenelse{\equal{\RewArrIter}{\nothing}}{}{^{\RewArrIter}}\ifthenelse{\equal{\RewArrCxt}{\nothing}}{}{_{\RewArrCxt}}
  } { \ifthenelse{\equal{\RewArrArr}{\finleft} \OR
      \equal{\RewArrArr}{\prsleft} \OR
      \equal{\RewArrArr}{\mrsleft} \OR
      \equal{\RewArrArr}{\wprsleft} \OR
      \equal{\RewArrArr}{\wmrsleft} \OR
      \equal{\RewArrArr}{\strongleft} \OR
      \equal{\RewArrArr}{\weakleft}}{
      \RewArrArr\ifthenelse{\equal{\RewArrIter}{\nothing}}{}{^{\RewArrIter}}\ifthenelse{\equal{\RewArrCxt}{\nothing}}{}{_{\RewArrCxt}}
    } { \latex@error{Rewrite arrow not defined}\@ehc } } \RewArrRhs }
\makeatother

\newcommand{\twoheadleftrightarrow}{\twoheadleftarrow\hspace{-7pt}\twoheadrightarrow}
\renewcommand{\to}{\RewArr{\finright}{\nothing}}

\newcommand{\fto}{\RewArr{\finright}}

\newcommand{\pacont}{\RewStmt{\prsright}{\dots}{\nothing}}

\newcommand{\pato}{\RewArr{\prsright}{\nothing}}

\newcommand{\macont}{\RewStmt{\mrsright}{\dots}{\nothing}}

\newcommand{\mato}{\RewArr{\mrsright}{\nothing}}

\newcommand{\wpacont}{\RewStmt{\wprsright}{\dots}{\nothing}}

\newcommand{\wpato}{\RewArr{\wprsright}{\nothing}}

\newcommand{\wmacont}{\RewStmt{\wmrsright}{\dots}{\nothing}}

\newcommand{\wmato}{\RewArr{\wmrsright}{\nothing}}

\newcommand{\join}[1][\nothing]
{
  \ifthenelse{\equal{#1}{\nothing}}{\downarrow}{\downarrow_{#1}}
}

\newcommand{\wconv}[1][\nothing]
{
  \ifthenelse{\equal{#1}{\nothing}}{\leftrightarrow^{w}}{\leftrightarrow^{w}_{#1}}
}

\newcommand{\sconv}[1][\nothing]
{
  \ifthenelse{\equal{#1}{\nothing}}{\twoheadleftrightarrow}{\twoheadleftrightarrow_{#1}}
}

\newcommand\cons{{\,:\,}}
\usepackage{appendix}
\usepackage{enumerate}
\usepackage{subfig}
\usepackage{amsthm}

\theoremstyle{plain}\newtheorem{theorem}{Theorem}[section]
\theoremstyle{plain}\newtheorem{proposition}[theorem]{Proposition}
\theoremstyle{plain}\newtheorem{fact}[theorem]{Fact}
\theoremstyle{plain}\newtheorem{lemma}[theorem]{Lemma}
\theoremstyle{plain}\newtheorem{corollary}[theorem]{Corollary}
\theoremstyle{definition}\newtheorem{definition}[theorem]{Definition}
\theoremstyle{definition}\newtheorem{rem}[theorem]{Remark}
\theoremstyle{definition}\newtheorem{example}[theorem]{Example}

\title{Infinitary Term Graph Rewriting}

\author{Patrick Bahr
  \\Department of Computer Science,\\University of Copenhagen
  \\
  Universitetsparken 1, 2100 Copenhagen, Denmark\\
\texttt{paba@diku.dk}}

\begin{document}
\bibliographystyle{plain}%
\maketitle

\begin{abstract}
  \noindent 
  Term graph rewriting provides a formalism for implementing term
  rewriting in an efficient manner by avoiding duplication. Infinitary
  term rewriting has been introduced to study infinite term reduction
  sequences. Such infinite reductions can be used to reason about lazy
  evaluation. In this paper, we combine term graph rewriting and
  infinitary term rewriting thereby addressing both components of lazy
  evaluation: non-strictness and sharing. Moreover, we show how our
  theoretical underpinnings, based on a metric space and a complete
  semilattice, provides a unified framework for both term rewriting
  and term graph rewriting. This makes it possible to study the
  correspondences between these two worlds. As an example, we show how
  the soundness of term graph rewriting w.r.t.\ term rewriting can be
  extended to the infinitary setting.
\end{abstract}

\section*{Introduction}
\label{S:one}

\emph{Infinitary term rewriting} \cite{kennaway03book} extends the
theory of term rewriting by giving a meaning to transfinite reductions
instead of dismissing them as undesired and meaningless
artifacts. \emph{Term graphs}, on the other hand, allow to explicitly
represent and reason about sharing and recursion \cite{ariola97ic} by
dropping the restriction to a tree structure that we have for
terms. Apart from that, term graphs also provide a finite
representation of certain infinite terms, viz.\ \emph{rational
  terms}. As Kennaway et
al.~\cite{kennaway95segragra,kennaway94toplas} have shown, this can be
leveraged in order to finitely represent restricted forms of
infinitary term rewriting using \emph{term graph rewriting}.

However, in order to benefit from this, we need to know for which
class of term rewriting systems the set of rational terms is closed
under (normalising) reductions. One such class of systems -- a rather
restrictive one -- is the class of \emph{regular equation systems}
\cite{courcelle83tcs} which consist of rules having only constants on
their left-hand side. Having an understanding of infinite reductions
over term graphs could help to investigate closure properties of
rational terms in the setting of infinitary term rewriting.

By studying infinitary calculi of term graph rewriting, we can also
expect to better understand calculi with explicit sharing and/or
recursion. Due to the lack of finitary confluence of these systems,
Ariola and Blom \cite{ariola02apal} resort to a notion of skew
confluence in order to be able to define infinite normal forms. An
appropriate infinitary calculus could provide a direct approach to
define infinite normal forms.

Historically, the theory of infinitary term rewriting is mostly based
on the metric space of terms \cite{kennaway03book}. Its notion of
convergence captures ``well-behaved'' transfinite reductions. A more
structured approach, based on the complete semilattice structure of
terms, yields a conservative extension of the metric calculus of
infinitary term rewriting~\cite{bahr10rta2} that allows local
divergence.

In previous work \cite{bahr11rta}, we have carefully devised a
complete metric space and a complete semilattice of term graphs in
order to investigate different modes of convergence for term
graphs. The resulting theory allows to treat infinitary term rewriting
as well as graph rewriting in the same theoretical framework. While
the devised metric and partial order on term graphs manifests the same
compatibility that is known for terms \cite{bahr10rta2}, it is too
restrictive as we will illustrate.

In this paper, we follow a different approach by taking the arguably
simplest generalisation of the metric space and the complete
semilattice of terms to term graphs. While the notion of convergence
in these structures has some oddities which makes them somewhat
incompatible, we will show that these incompatibilities vanish once we
move from the weak notion of convergence that was considered in
\cite{bahr11rta} to the much more well-behaved strong notion of
convergence \cite{kennaway95ic}. More concretely, we will show that,
w.r.t.\ strong convergence, the metric calculus of infinitary term
graph rewriting is the \emph{total fragment} of the partial order
calculus of infinitary term graph rewriting.

We show that our simple approach to infinitary term graph rewriting
yields simple limit constructions that makes them easy to relate to
the limit constructions on terms. As a result of that we are able to
generalise the soundness result as well as a limited completeness
result for term graph rewriting \cite{kennaway94toplas} to the
infinitary setting.

\section{Preliminaries}
\label{sec:preliminaries}

We assume the reader to be familiar with the basic theory of ordinal
numbers, orders and topological spaces \cite{kelley55book}, as well as
term rewriting \cite{terese03book}. In the following, we
briefly recall the most important notions.

\subsection{Sequences}
We use $\alpha, \beta, \gamma, \lambda, \iota$ to denote ordinal
numbers. A \emph{sequence} $S$ of length $\alpha$ in a set $A$,
written $(a_\iota)_{\iota < \alpha}$, is a function from $\alpha$ to
$A$ with $\iota \mapsto a_\iota$ for all $\iota \in \alpha$. We use
$\len{S}$ to denote the length $\alpha$ of $S$. If $\alpha$ is a limit
ordinal, then $S$ is called \emph{open}. Otherwise, it is called
\emph{closed}. If $\alpha$ is a finite ordinal, then $S$ is called
\emph{finite}. Otherwise, it is called \emph{infinite}. For a finite
sequence $(a_i)_{i < n}$ we also use the notation
$\seq{a_0,a_1,\dots,a_{n-1}}$. In particular, $\emptyseq$ denotes an
empty sequence.

The \emph{concatenation} $(a_\iota)_{\iota<\alpha}\concat
(b_\iota)_{\iota<\beta}$ of two sequences is the sequence
$(c_\iota)_{\iota<\alpha+\beta}$ with $c_\iota = a_\iota$ for $\iota <
\alpha$ and $c_{\alpha+\iota} = b_\iota$ for $\iota < \beta$. A
sequence $S$ is a (proper) \emph{prefix} of a sequence $T$, denoted $S
\le T$ (resp.\ $S < T$), if there is a (non-empty) sequence $S'$ with
$S\concat S' = T$. The prefix of $T$ of length $\beta$ is denoted
$\prefix{T}{\beta}$. The binary relation $\le$ forms a complete
semilattice. Similarly, a sequence $S$ is a (proper) \emph{suffix} of
a sequence $T$ if there is a (non-empty) sequence $S'$ with $S'\concat
S = T$.

Let $S = (a_\iota)_{\iota < \alpha}$ be a sequence. A sequence $T =
(b_{\iota})_{\iota < \beta}$ is called a \emph{subsequence} of $S$ if
there is a monotone function $f\fcolon \beta \to \alpha$ such that
$b_\iota = a_{f(\iota)}$ for all $\iota < \beta$. The subsequence $S$
is called \emph{finial} if $f$ is cofinal, i.e.\ if for each $\iota <
\beta$ there is some $\gamma < \alpha$ with $f(\gamma) \ge \iota$.

\subsection{Metric Spaces}
\label{sec:metric-spaces}

A pair $(M,\dd)$ is called a \emph{metric space} if $\dd \fcolon M
\times M \to \realsnn$ is a function satisfying $\dd(x,y) = 0$ iff
$x=y$ (identity), $\dd(x, y) = \dd(y, x)$ (symmetry), and $\dd(x, z)
\le \dd(x, y) + \dd(y, z)$ (triangle inequality), for all $x,y,z\in
M$.  If $\dd$ instead of the triangle inequality, satisfies the
stronger property $\dd(x, z) \le \max \set{ \dd(x, y),\dd(y, z)}$
(strong triangle), then $(M,\dd)$ is called an \emph{ultrametric
  space}. Let $(a_\iota)_{\iota<\alpha}$ be a sequence in a metric
space $(M,\dd)$. The sequence $(a_\iota)_{\iota<\alpha}$
\emph{converges} to an element $a\in M$, written
$\lim_{\iota\limto\alpha} a_\iota$, if, for each $\varepsilon \in
\realsp$, there is a $\beta < \alpha$ such that $\dd(a,a_\iota) <
\varepsilon$ for every $\beta < \iota < \alpha$;
$(a_\iota)_{\iota<\alpha}$ is \emph{continuous} if
$\lim_{\iota\limto\lambda} a_\iota = a_\lambda$ for each limit ordinal
$\lambda < \alpha$. The sequence $(a_\iota)_{\iota<\alpha}$ is called
\emph{Cauchy} if, for any $\varepsilon \in \realsp$, there is a
$\beta<\alpha$ such that, for all $\beta < \iota < \iota' < \alpha$,
we have that $\dd(m_\iota,m_{\iota'}) < \varepsilon$.  A metric space
is called \emph{complete} if each of its non-empty Cauchy sequences
converges.

Note that the limit of a converging sequence is preserved by taking
cofinal subsequences:
\begin{proposition}[invariance of the limit]
  \label{prop:convSubseq}
  Let $(a_i)_{i<\alpha}$ be a sequence in a metric space $(A,\dd)$. If
  $\lim_{\iota\limto\alpha} a_\iota = a$ then $\lim_{\iota \limto
    \beta} b_\iota = a$ for any cofinal subsequence $(b_i)_{i<\beta}$ of
  $(a_i)_{i<\alpha}$.
\end{proposition}

\subsection{Partial Orders}
\label{sec:partial-orders}
A \emph{partial order} $\le$ on a set $A$ is a binary relation on $A$
that is \emph{transitive}, \emph{reflexive}, and
\emph{antisymmetric}. The pair $(A,\le)$ is then called a
\emph{partially ordered set}. A subset $D$ of the underlying set $A$
is called \emph{directed} if it is non-empty and each pair of elements
in $D$ has an upper bound in $D$. A partially ordered set $(A, \le)$
is called a \emph{complete partial order} (\emph{cpo}) if it has a
least element and each directed set $D$ has a \emph{least upper bound}
(\emph{lub}) $\Lub D$. A cpo $(A, \le)$ is called a \emph{complete
  semilattice} if every \emph{non-empty} set $B$ has \emph{greatest
  lower bound} (\emph{glb}) $\Glb B$. In particular, this means that
for any non-empty sequence $(a_\iota)_{\iota<\alpha}$ in a complete
semilattice, its \emph{limit inferior}, defined by $\liminf_{\iota
  \limto \alpha}a_\iota = \Lub_{\beta<\alpha} \left(\Glb_{\beta \le
    \iota < \alpha} a_\iota\right)$, always exists.

It is easy to see that the limit inferior of closed sequences is
simply the last element of the sequence. This is, however, only a
special case of the following more general proposition:
\begin{proposition}[invariance of the limit inferior]
  \label{prop:liminfSuffix}
  Let $(a_\iota)_{\iota < \alpha}$ be a sequence in a partially
  ordered set and $(b_\iota)_{\iota< \beta}$ a non-empty suffix of
  $(a_\iota)_{\iota < \alpha}$. Then $\liminf_{\iota \limto \alpha}
  a_\iota = \liminf_{\iota \limto \beta} b_\iota$.
\end{proposition}
\begin{proof}
  We have to show that $\Lub_{\gamma<\alpha} \Glb_{\gamma \le \iota <
    \alpha} a_\iota = \Lub_{\beta \le \gamma<\alpha} \Glb_{\gamma \le
    \iota < \alpha} a_\iota = \ol a'$ holds for each $\beta < \alpha$.
  Let $b_\gamma = \Glb_{\gamma \le \iota < \alpha} a_\iota$ for each
  $\gamma < \alpha$, $A = \setcom{b_\gamma}{\gamma < \alpha}$ and $A'
  = \setcom{b_\gamma}{\beta \le \gamma < \alpha}$. Note that $\ol a =
  \Lub A$ and $\ol a' = \Lub A'$. Because $A' \subseteq A$, we have
  that $\ol a' \le \ol a$. On the other hand, since $b_\gamma \le
  b_{\gamma'}$ for $\gamma \le \gamma'$, we find, for each $b_\gamma
  \in A$, some $b_{\gamma'} \in A'$ with $b_\gamma \le
  b_{\gamma'}$. Hence, $\ol a \le \ol a'$. Therefore, due to the
  antisymmetry of $\le$, we can conclude that $\ol a = \ol a'$.
\end{proof} 
Note that the limit in a metric space has the same behaviour as the
one for the limit inferior described by the proposition
above. However, one has to keep in mind that -- unlike the limit --
the limit inferior is not invariant under taking cofinal subsequences!

With the prefix order $\le$ on sequences we can generalise
concatenation to arbitrary sequences of sequences: Let
$(S_\iota)_{\iota < \alpha}$ be a sequence of sequences in a common
set. The concatenation of $(S_\iota)_{\iota < \alpha}$, written
$\Concat_{\iota < \alpha} S_\iota$, is recursively defined as the
empty sequence $\emptyseq$ if $\alpha = 0$, $\left(\Concat_{\iota <
    \alpha'} S_\iota\right) \concat S_{\alpha'}$ if $\alpha = \alpha'
+ 1$, and $\Lub_{\gamma < \alpha} \Concat_{\iota < \gamma} S_\iota$ if
$\alpha$ is a limit ordinal.

\subsection{Terms}
\label{sec:terms}

Since we are interested in the infinitary calculus of term rewriting,
we consider the set $\iterms$ of \emph{infinitary terms} (or simply
\emph{terms}) over some \emph{signature} $\Sigma$. A \emph{signature}
$\Sigma$ is a countable set of symbols. Each symbol $f$ is associated
with its arity $\srank{f}\in \nats$, and we write $\Sigma^{(n)}$ for
the set of symbols in $\Sigma$ which have arity $n$. The set $\iterms$
is defined as the \emph{greatest} set $T$ such that $t \in T$ implies
$t = f(t_1,\dots, t_k)$, where $f \in \Sigma^{(k)}$, and
$t_1,\dots,t_k\in T$. For each constant symbol $c\in \Sigma^{(0)}$, we
write $c$ for the term $c()$. We consider $\iterms$ as a superset of
the set $\fterms$ of \emph{finite terms}. For a term $t \in \iterms$
we use the notation $\pos{t}$ to denote the \emph{set of positions} in
$t$. $\pos{t}$ is the least subset of $\nats^{*}$ such that $\emptyseq
\in \pos{t}$ and $\pi\concat\seq{i} \in \pos{t}$ if $t =
f(t_1,\dots,t_k)$ with $0 \le i < k$.  For terms $s,t \in \iterms$ and
a position $\pi \in \pos{t}$, we write $\atPos{t}{\pi}$ for the
\emph{subterm} of $t$ at $\pi$, $t(\pi)$ for the function symbol in
$t$ at $\pi$, and $\substAtPos{t}{\pi}{s}$ for the term $t$ with the
subterm at $\pi$ replaced by $s$. A position is also called an
\emph{occurrence} if the focus lies on the subterm at that position
rather than the position itself.

On $\iterms$ a similarity measure $\similar{\cdot}{\cdot} \in \nats
\cup \set{\infty}$ can be
defined by setting
\[
\similar{s}{t} = \min \setcom{\len{\pi}}{\pi \in \pos{s}\cap\pos{t}, s(\pi) \neq
  t(\pi)} \cup \set{\infty} \qquad \text {for } s,t\in \iterms
\]
That is, $\similar{s}{t}$ is the minimal depth at which $s$ and $t$
differ, resp.\ $\infty$ if $s = t$. Based on this, a distance function
$\dd$ can be defined by $\dd(s,t) = 2^{-\similar{s}{t}}$, where we
interpret $2^{-\infty}$ as $0$. The pair $(\iterms, \dd)$ is known to
form a complete ultrametric space \cite{arnold80fi}. \emph{Partial
  terms}, i.e.\ terms over signature $\Sigma_\bot = \Sigma \uplus
\set{\bot}$ with $\bot$ a fresh constant symbol, can be endowed with a
binary relation $\lebot$ by defining $s \lebot t$ iff $s$ can be
obtained from $t$ by replacing some subterm occurrences in $t$ by
$\bot$. Interpreting the term $\bot$ as denoting ``undefined'',
$\lebot$ can be read as ``is less defined than''. The pair
$(\ipterms,\lebot)$ is known to form a complete semilattice
\cite{goguen77jacm}. To explicitly distinguish them from partial
terms, we call terms in $\iterms$ \emph{total}.

\subsection{Term Rewriting Systems}
\label{sec:abstr-reduct-syst}

For term rewriting systems, we have to consider terms with
variables. To this end, we assume a countably infinite set $\calV$ of
variables and extend a signature $\Sigma$ to a signature $\Sigma_\calV
= \Sigma \uplus \calV$ with variables in $\calV$ as nullary
symbols. Instead of $\iterms[\Sigma_\calV]$ we also write $\viterms$.
A \emph{term rewriting system} (TRS) $\calR$ is a pair $(\Sigma, R)$
consisting of a signature $\Sigma$ and a set $R$ of \emph{term rewrite
  rules} of the form $l \to r$ with $l \in \vterms \setminus \calV$
and $r \in \viterms$ such that all variables in $r$ are contained in
$l$. Note that the left-hand side must be a finite term
\cite{kennaway03book}! We usually use $x,y,z$ and primed resp.\
indexed variants thereof to denote variables in $\calV$.

As in the finitary setting, every TRS $\calR$ defines a rewrite
relation $\to[\calR]$:
\[
s \to[\calR] t \iff \exists \pi \in \pos{s}, l\to r \in
R, \sigma\colon\; \atPos{s}{\pi} = l\sigma, t = \substAtPos{s}{\pi}{r\sigma}
\]
Instead of $s \to[\calR] t$, we sometimes write $s \to[\pi,\rho] t$ in
order to indicate the applied rule $\rho$ and the position $\pi$, or
simply $s \to t$. The subterm $\atPos{s}{\pi}$ is called a
\emph{$\rho$-redex} or simply \emph{redex}, $r\sigma$ its
\emph{contractum}, and $\atPos{s}{\pi}$ is said to be
\emph{contracted} to $r\sigma$.

Let $\rho_1\fcolon l_1 \to r_1$, $\rho_2\fcolon l_2 \to r_2$ be rules
in a TRS $\calR$ with variables renamed apart. The rules
$\rho_1,\rho_2$ are said to \emph{overlap} if there is a non-variable
position $\pi$ in $l_1$ such that $\atPos{l_1}{\pi}$ and $l_2$ are
unifiable and $\pi$ is not the root position $\emptyseq$ in case
$\rho_1,\rho_2$ are renamed copies of the same rule. A TRS is called
\emph{non-overlapping} if none of its rules overlap. A term $t \in
\vterms$ is called \emph{linear} if each variable occurs at most once
in $t$. The TRS $\calR$ is called \emph{left-linear} if the left-hand
side of every rule in $\calR$ is linear. It is called
\emph{orthogonal} if it is left-linear and non-overlapping.

\section{Infinitary Term Rewriting}
\label{sec:preliminaries}

Before pondering over the right approach to an infinitary calculus of
term graph rewriting, we want to provide a brief overview of
infinitary term graph
rewriting~\cite{kennaway03book,bahr10rta2,blom04rta}. This should give
a insight into the different approaches to deal with infinite
reductions.

A \emph{(transfinite) reduction} in a term rewriting system $\calR$,
is a sequence $S = (t_\iota \to_{\pi_\iota} t_{\iota +1})_{\iota <
  \alpha}$ of rewriting steps in $\calR$. The reduction $S$ is called
\emph{weakly $\mrs$-continuous}, written $S\fcolon t_0 \wmacont$, if
the sequence of terms $(t_\iota)_{\iota < \wsuc\alpha}$ is continuous,
i.e.\ $\lim_{\iota\limto\lambda} t_\iota = t_\lambda$ for each limit
ordinal $\lambda < \alpha$. The reduction $S$ is said to \emph{weakly
  $\mrs$-converge} to a term $t$, written $S\fcolon t_0 \wmato t$, if
it is weakly $\mrs$-continuous and $\lim_{\iota\limto\alpha} t_\iota =
t$.

For strong convergence, also the positions $\pi_\iota$ at which
reductions take place are taken into consideration: A reduction $S$ is
called strongly \emph{$\mrs$-continuous}, written $S\fcolon t_0
\macont $, if it is weakly $\mrs$-continuous and the depths of redexes
$(\len{\pi_\iota})_{\iota < \lambda}$ tend to infinity for each limit
ordinal $\lambda < \alpha$, i.e.\ $liminf_{\iota \limto \lambda}
\len{\pi_\iota} = \omega$. A reduction $S$ is said to strongly
$\mrs$-converge to $t$, written $S\fcolon t_0 \mato t$, if it weakly
$\mrs$-converges to $t$ and the depths of redexes
$(\len{\pi_\iota})_{\iota < \lambda}$ tend to infinity for each limit
ordinal $\lambda \le \alpha$.

\begin{example}
  \label{ex:termRewr}
  Consider the term rewriting system $\calR$ containing the rule
  $\rho_1\fcolon a \cons x \to b \cons a \cons x$, where $\cons$ is a
  binary symbol that we write infix and assume to associate to the
  right. That is, the right-hand side of the rule is parenthesised as
  $b \cons (a \cons x)$. Think of the $\cons$ symbol as the list
  constructor \emph{cons}. Using $\rho_1$. we have the infinite
  reduction
  \[
  S\fcolon a \cons c \to b \cons a \cons c \to b \cons b \cons a \cons c \to \dots
  \]
  The position at which two consecutive terms differ moves deeper and
  deeper during the reduction $S$. Hence, $S$ weakly $\mrs$-converges
  to the infinite term $s$ satisfying the equation $s = b \cons s$,
  i.e.\ $s = b \cons b \cons b \cons \dots$. Since also the position
  at which the reductions take place moves deeper and deeper, $S$ also
  strongly $\mrs$-converges to $s$.

  Now consider a TRS with the slightly different rule $\rho_2\fcolon a
  \cons x \to a \cons b \cons x$. This yields a reduction
  \[
  S'\fcolon a \cons c \to a \cons b \cons c \to a \cons b \cons b \cons c \to \dots
  \]
  The reduction $S'$ weakly $\mrs$-converges to the term $s' = a \cons
  b \cons b \cons \dots$. However, since in each step in $S'$ takes
  place at the root, it is not strongly $\mrs$-converging.
\end{example}

Strong $\mrs$-convergence is determined by the depth of the redexes
only. The metric space is only used to determine the limit term.
\begin{proposition}[{\cite[Prop.~5.5]{bahr10rta}}]
  \label{prop:strConvDepth}
  Let $S = (t_\iota \to[\pi_\iota] t_{\iota +1})_{\iota < \lambda}$ be
  a strongly $\mrs$-continuous open reduction in a TRS. Then $S$ is
  strongly $\mrs$-convergent iff the sequence
  $(\len{\pi}_\iota)_{\iota<\lambda}$ of redex depths tends to
  infinity.
\end{proposition}

In the partial order model of infinitary rewriting, convergence is
modelled by the limit inferior: A reduction $S = (t_\iota
\to_{\pi_\iota} t_{\iota +1})_{\iota < \alpha}$ of \emph{partial
  terms} is called \emph{weakly $\prs$-continuous}, written $S\fcolon
t_0 \wpacont$, if $\liminf_{\iota<\lambda} t_\iota = t_\lambda$ for
each limit ordinal $\lambda < \alpha$. The reduction $S$ is said to
\emph{weakly $\prs$-converge} to a term $t$, written $S\fcolon t_0
\wpato t$, if it is weakly $\prs$-continuous and
$\liminf_{\iota<\wsuc\alpha} t_\iota = t$.

Again, for strong convergence, the positions $\pi_\iota$ at which
reductions take place are taken into consideration. In particular, we
consider for a reduction step $t_\iota \to[\pi_\iota] t_{\iota+1}$ the
\emph{reduction context} $c_\iota =
\substAtPos{t_\iota}{\pi}{\bot}$. To indicate the reduction context of
a reduction step, we also write $t_\iota \to[c_\iota] t_{\iota+1}$. A
reduction $S = (t_\iota \to_{c_\iota} t_{\iota +1})_{\iota < \alpha}$
is called \emph{strongly $\prs$-continuous}, written $S\fcolon t_0
\pacont$, if $\liminf_{\iota<\lambda} c_\iota = t_\lambda$ for each
limit ordinal $\lambda < \alpha$. The reduction $S$ is said to
\emph{weakly $\prs$-converge} to a term $t$, written $S\fcolon t_0
\pato t$, if it is weakly $\prs$-continuous and either $T$ is closed
with $t = t_\alpha$, or $\liminf_{\iota<\wsuc\alpha} c_\iota = t$.

The distinguishing feature of the partial order approach is that,
given a complete semilattice, each continuous reduction also
converges.  This provides a conservative extension to
$\mrs$-convergence that allows rewriting modulo \emph{meaningless
  terms}~\cite{bahr10rta2} by essentially mapping those parts of the
reduction to $\bot$ that are divergent according to the metric model.

Intuitively, weak $\prs$-convergence on terms describes an
approximation process. To this end, the partial order $\lebot$
captures a notion of \emph{information preservation}, i.e.\ $s \lebot
t$ iff $t$ contains at least the same information as $s$ does but
potentially more. A monotonic sequence of terms $t_0 \lebot t_1 \lebot
\dots$ thus approximates the information contained in $\Lub_{i<\omega}
t_i$. Given this reading of $\lebot$, the glb $\Glb T$ of a set of
terms $T$ captures the common (non-contradicting) information of the
terms in $T$. Leveraging this, a sequence that is not necessarily
monotonic can be turned into a monotonic sequence $t_j = \Glb_{j \le i
  < \omega} s_j$ such that each $t_j$ contains exactly the information
that remains stable in $(s_i)_{i<\omega}$ from $j$ onwards. Hence, the
limit inferior $\liminf_{i \limto \omega} s_i = \Lub_{j <
  \omega}\Glb_{j \le i < \omega} s_i$ is the term that contains the
accumulated information that eventually remains stable in
$(s_i)_{i<\omega}$. This is expressed as an approximation of the
monotonically increasing information that remains stable from some
point on. For the strong variant, instead of the terms $s_\iota$, the
reduction contexts $c_\iota$ are considered. Each reduction context
$c_\iota$ is an underapproximation of the shared structure $s_\iota
\glb s_{\iota+1}$ between two consecutive terms $s_\iota,s_{\iota+1}$.

\begin{example}
  \label{ex:termRewr2}
  Reconsider the system from Example~\ref{ex:termRewr}. The reduction
  $S$ also weakly and strongly $\prs$-converges to $s$. Its sequence
  of stable information $\bot \cons \bot \lebot b \cons \bot \cons
  \bot \lebot b \cons b \cons \bot \cons \bot \lebot \dots$
  approximates $s$. The same also applies to the stricter
  underapproximation $\bot \lebot b \cons \bot \lebot b \cons b \cons
  \bot \lebot \dots$ by reduction contexts. Now consider the rule
  $\rho_1$ together with the rule $\rho_3\fcolon b \cons x \to a \cons
  b \cons x$. Starting with the same term, but applying the two rules
  alternately at the root, we obtain the reduction sequence
  \[
  T\fcolon a \cons c \to b \cons a \cons c \to a \cons b \cons a \cons c \to b \cons a \cons b \cons a \cons c \to \dots
  \]
  Now the differences between two consecutive terms occur right below
  the root symbol ``$\cons$''. Hence, $T$ does not even weakly
  $\mrs$-converge. This, however, only affects the left argument of
  ``$\cons$''. Following the right argument position, the bare list
  structure becomes eventually stable. The sequence of stable
  information $\bot \cons \bot \lebot \bot \cons \bot \cons \bot
  \lebot \bot \cons \bot \cons \bot \cons \bot \lebot \dots$
  approximates the term $t = \bot \cons \bot \cons \bot \dots$. Hence,
  $T$ weakly $\prs$-converge to $t$. Since each reduction takes place
  at the root, each reduction context is $\bot$. Therefore, $T$
  strongly $\prs$-converges to the term $\bot$.
\end{example}

Note that in both the metric and the partial order setting continuity
is simply the convergence of every proper prefix:
\begin{proposition}[{\cite{bahr10rta}}]
  \label{prop:contConv}
  Let $S = (t_\iota \to t_{\iota +1})_{\iota < \alpha}$ be a reduction
  in a TRS. Then $S$ is strongly $\mrs$-continuous iff every proper
  prefix $\prefix{S}{\beta}$ strongly $\mrs$-converges to $t_\beta$
  The same holds for strong $\prs$-continuity/-convergence and weak
  counterparts.
\end{proposition}

Moreover, the relation between $\mrs$- and $\prs$-convergence
illustrated in the examples above is characteristic:
$\prs$-convergence is a conservative extension of $\mrs$-convergence.
\begin{theorem}[total $\prs$-convergence = $\mrs$-convergence]
  \label{thr:strongExt}
  For every reduction $S$ in a TRS the following equivalences hold:
  \begin{center}
    \begin{inparaenum}[(i)]
    \item $S\fcolon s \wpato t$ is total iff $S\fcolon s \wmato t$,
      and \label{item:strongExtI}%
      \quad%
    \item $S\fcolon s \pato t$ is total iff $S\fcolon s \mato
      t$. \label{item:strongExtII}
    \end{inparaenum}
  \end{center}
  The same also holds for continuity instead of convergence.
\end{theorem}
Kennaway \cite{kennaway92rep} and Bahr \cite{bahr10rta} investigated
abstract models of infinitary rewriting based on metric spaces resp.\
partially ordered sets. We will take these abstract models as a basis
to formulate a theory of infinitary term graph reductions. The key
question that we have to address is what an appropriate metric space
resp.\ partial order on term graphs looks like.

\section{Graphs and Term Graphs}
\label{sec:term-graphs}

This section provides the basic notions for term graphs and more
generally for graphs.  Terms over a signature, say $\Sigma$, can be
thought of as rooted trees whose nodes are labelled with symbols from
$\Sigma$. Moreover, in these trees a node labelled with a $k$-ary
symbol is restricted to have out-degree $k$ and the outgoing edges are
ordered. In this way the $i$-th successor of a node labelled with a
symbol $f$ is interpreted as the root node of the subtree that
represents the $i$-th argument of $f$. For example, consider the term
$f(a,h(a,b))$. The corresponding representation as a tree is shown in
Figure~\ref{fig:exTermTree}.

\begin{figure}
  \centering
  \subfloat[$f(a,h(a,b))$.]{
    \label{fig:exTermTree}
    \begin{tikzpicture}[->,baseline=(b.base)]
      \node (r1) at (0,0)  {$f$}
      child {
        node  {$a$}
      } child {
        node {$h$}
        child {
          node {$a$}
        } child {
          node (b) {$b$}
        }
      };
    \end{tikzpicture}
  }
  \qquad
  \subfloat[A graph.]{
    \label{fig:exGraph}
    \begin{tikzpicture}[->,baseline=(a.base)]
      \node (f) {$f$}
      child {
        node (g) {$h$}
        child [missing]
        child {
          node (a) {$a$}
        }
      } child {
        node (b) {$b$}
      };
      \node [node distance=1cm,right=of f] {$h$}
      edge (b)
      edge [bend left=50] (a);
      \path[use as bounding box] (-1.3,0);
      \draw (g) edge[out=-115,in=180, min distance=1.5cm] (f); 
    \end{tikzpicture}
  }
  \qquad
  \subfloat[A term graph.]{%
    \label{fig:exTermGraph}%
    \begin{tikzpicture}[baseline=(n4.base)]
      \node (n1) {$f$}
      child {
        node (n2) {$h$}
        child [missing]
        child {
          node (n4) {$a$}
        }
      }
      child {
        node (n3) {$f$}
      };
      \draw[->] (n3) edge[out=245, in=-25] (n2);
      \draw[->] (n3) edge[out=295, in=25] (n4);
      \path[use as bounding box] (-1.5,0);
      \draw[->] (n2) edge[out=245, in=180, min distance=2cm] (n1);
    \end{tikzpicture}
    }%
  \caption{Example for a tree representation of a term; generalisation
    to (term) graphs.}
\end{figure}
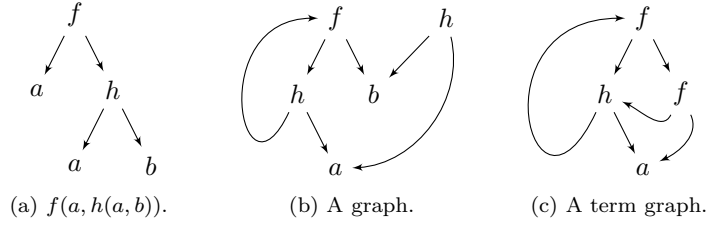

In term graphs, the restriction to a tree structure is abolished. The
notion of term graphs we are using is taken from Barendregt et al.\
\cite{barendregt87parle}.

\begin{definition}[graph]
  \label{def:graph}
  Let $\Sigma$ be a signature. A \emph{graph} over $\Sigma$ is a tuple
  $g = (N,\glab,\gsuc)$ consisting of a set $N$ (of \emph{nodes}), a
  \emph{labelling function} $\glab\fcolon N \funto \Sigma$, and a
  \emph{successor function} $\gsuc\fcolon N \funto N^*$ such that
  $\len{\gsuc(n)} = \srank{\glab(n)}$ for each node $n\in N$, i.e.\ a
  node labelled with a $k$-ary symbol has precisely $k$ successors. If
  $\gsuc(n) = \seq{n_0,\dots,n_{k-1}}$, then we write $\gsuc_{i}(n)$
  for $n_i$. Moreover, we use the abbreviation $\rank{g}{n}$ for the
  arity $\srank{\glab(n)}$ of $n$.
\end{definition}

\begin{example}
  Let $\Sigma = \set{f/2,h/2,a/0,b/0}$ be a signature. The graph over
  $\Sigma$, depicted in Figure~\ref{fig:exGraph}, is given by the
  triple $(N, \glab, \gsuc)$ with $N = \set{n_0,n_1,n_2,n_3,n_4}$,
  $\glab(n_0) = f, \glab(n_1) = \glab(n_4) = h, \glab(n_2) = b,
  \glab(n_3)=a$ and $\gsuc(n_0) = \seq{n_1,n_2},
  \gsuc(n_1)=\seq{n_0,n_3}, \gsuc(n_2) = \gsuc(n_3) =\seq{},
  \gsuc(n_4) = \seq{n_2,n_3}$.
\end{example}

\begin{definition}[path, reachability]
  \label{def:graphPath}
  Let $g = (N,\glab,\gsuc)$ be a graph and $n,n' \in N$.
  \begin{enumerate}[(i)]
  \item A \emph{path} in $g$ from $n$ to $n'$ is a finite sequence
    $(p_i)_{i<l}$ in $\nats$ such that either
    \begin{itemize}
    \item $n = n'$ and $(p_i)_{i<l}$ is empty, i.e.\ $l = 0$, or
    \item $0 \le p_0 < \rank{g}{n}$ and the suffix $(p_i)_{1\le i <
        l}$ is a path in $g$ from $\gsuc_{p_0}(n)$ to $n'$.
    \end{itemize}
  \item If there exists a path from $n$ to $n'$ in $g$, we say that
    $n'$ is \emph{reachable} from $n$ in $g$.
  \end{enumerate}
\end{definition}

\begin{definition}[term graph]
  \label{def:tgraph}
  Given a signature $\Sigma$, a \emph{term graph} $g$ over $\Sigma$ is
  a tuple $(N,\glab,\gsuc,r)$ consisting of an \emph{underlying} graph
  $(N,\glab,\gsuc)$ over $\Sigma$ whose nodes are all reachable from
  the \emph{root node} $r\in N$. The class of all term graphs over
  $\Sigma$ is denoted $\itgraphs$. We use the notation $N^{g}$,
  $\glab^{g}$, $\gsuc^{g}$ and $r^{g}$ to refer to the respective
  components $N$,$\glab$, $\gsuc$ and $r$ of $g$. Given a graph or a
  term graph $h$ and a node $n$ in $h$, we write $\subgraph{h}{n}$ to
  denote the sub-term graph of $h$ rooted in $g$.
\end{definition}

\begin{example}
  Let $\Sigma = \set{f/2,h/2,c/0}$ be a signature. The term graph over
  $\Sigma$, depicted in Figure~\ref{fig:exTermGraph}, is given by the
  quadruple $(N, \glab, \gsuc,r)$, where $N = \set{r, n_1, n_2,
    n_3}$, $\gsuc(r) = \seq{n_1, n_2}$, $\gsuc(n_1) = \seq{r, n_3}$,
  $\gsuc(n_2) = \seq{n_1, n_3}$, $\gsuc(n_3) = \emptyseq$ and
  $\glab(r) = \glab(n_2) = f$, $\glab(n_1) = h$, $\glab(n_3) = c$.
\end{example}

Paths in a graph are not absolute but relative to a starting node. In
term graphs, however, we have a distinguished root node from which
each node is reachable. Paths relative to the root node are central
for dealing with term graphs:
\begin{definition}[position, depth, cyclicity, tree]
  \label{def:tgraphOcc}
  Let $g \in \itgraphs$ and $n \in N$.
  \begin{enumerate}[(i)]
  \item A \emph{position} of $n$ is a path in the underlying graph
    of $g$ from $r^g$ to $n$. The set of all positions in $g$ is
    denoted $\pos{g}$; the set of all positions of $n$ in $g$ is denoted
    $\nodePos{g}{n}$.\footnote{The notion/notation of positions is
      borrowed from terms: Every position $\pi$ of a node $n$
      corresponds to the subterm represented by $n$ occurring at
      position $\pi$ in the unravelling of the term graph to a term.}
  \item The \emph{depth} of $n$ in $g$, denoted $\depth{g}{n}$, is the
    minimum of the lengths of the positions of $n$ in $g$, i.e.\
    $\depth{g}{n} = \min \setcom{\len{\pi}}{\pi \in \nodePos{g}{n}}$.
  \item For a position $\pi \in \pos{g}$, we write
    $\nodeAtPos{g}{\pi}$ for the unique node $n\in N^g$ with $\pi \in
    \nodePos{g}{n}$ and $g(\pi)$ for its symbol $\glab^g(n)$.
  \item A position $\pi\in\pos{g}$ is called \emph{cyclic} if there
    are paths $\pi_1 <\pi_2 \le \pi$ with $\nodeAtPos{g}{\pi_1} =
    \nodeAtPos{g}{\pi_2}$. The non-empty path $\pi'$ with
    $\pi_1\concat \pi' = \pi_2$ is then called a \emph{cycle} of
    $\nodeAtPos{g}{\pi_1}$. A position that is not cyclic is called
    \emph{acyclic}.
  \item The term graph $g$ is called a \emph{term tree} if each node
    in $g$ has exactly one position.
  \end{enumerate}
\end{definition}

Note that the labelling function of graphs -- and thus term graphs --
is \emph{total}. In contrast, Barendregt et al.\
\cite{barendregt87parle} considered \emph{open} (term) graphs with a
\emph{partial} labelling function such that unlabelled nodes denote
holes or variables. This is reflected in their notion of homomorphisms
in which the homomorphism condition is suspended for unlabelled nodes.

\subsection{Homomorphisms}
\label{sec:homomorphisms}

Instead of a partial node labelling function, we chose a
\emph{syntactic} approach that is closer to the representation in
terms: Variables, holes and ``bottoms'' are represented as
distinguished syntactic entities. We achieve this on term graphs by
making the notion of homomorphisms dependent on a distinguished set of
constant symbols $\Delta$ for which the homomorphism condition is
suspended:

\begin{definition}[$\Delta$-homomorphism]
  \label{def:D-hom}
  Let $\Sigma$ be a signature, $\Delta\subseteq \Sigma^{(0)}$, and
  $g,h \in \itgraphs$.
  \begin{enumerate}[(i)]
  \item A function $\phi\fcolon N^g \funto N^h$ is called
    \emph{homomorphic}\ in $n \in N^g$ if the following holds:
    \begin{align*}
      \glab^g(n) &= \glab^h(\phi(n))
      \tag{labelling}\\
      \phi(\gsuc^g_i(n)) &= \gsuc^h_i(\phi(n)) \quad \text{ for all } 0 \le i <
      \rank{g}{n} \tag{successor}
    \end{align*}
  \item A \emph{$\Delta$-homomorphism} $\phi$ from $g$ to $h$, denoted
    $\phi\fcolon g \homto_\Delta h$, is a function $\phi\fcolon N^g
    \funto N^h$ that is homomorphic in $n$ for all $n \in N^g$ with
    $\glab^g(n) \nin \Delta$ and satisfies $\phi(r^g) = r^h$.
  \end{enumerate}
\end{definition}

It should be obvious that we get the usual notion of homomorphisms on
term graphs if $\Delta = \emptyset$. The $\Delta$-nodes can be thought
of as holes in the term graphs which can be filled with other term
graphs. For example, if we have a distinguished set of variable
symbols $\calV \subseteq \Sigma^{(0)}$, we can use
$\calV$-homomorphisms to formalise the matching step of term graph
rewriting which requires the instantiation of variables.

\begin{proposition}[$\Delta$-homomorphism preorder]
  \label{prop:catgraph}
  The $\Delta$-homomorphisms on $\itgraphs$ form a category which is a
  preorder. That is, there is at most one $\Delta$-homomorphism from
  one term graph to another.
\end{proposition}
\begin{proof}
  The identity $\Delta$-homomorphism is obviously the identity mapping
  on the set of nodes. Moreover, an easy equational reasoning reveals
  that the composition of two $\Delta$-homomorphisms is again a
  $\Delta$-homomorphism. Associativity of this composition is obvious
  as $\Delta$-homomorphisms are functions.

  In order to show that the category is a preorder assume that there
  are two $\Delta$-homomorphisms $\phi_1,\phi_2\fcolon g \homto_\Delta
  h$. We prove that $\phi_1 = \phi_2$ by showing that $\phi_1(n) =
  \phi_2(n)$ for all $n \in N^g$ by induction on the depth of $n$.

  Let $\depth{g}{n} = 0$, i.e.\ $n = r^g$. By the root condition, we
  have that $\phi_1(r^g) = r^h = \phi_2(r^g)$. Let $\depth{g}{n} = d >
  0$. Then $n$ has a position $\pi \concat \seq i$ in $g$ such that
  $\depth{g}{n'}<d$ for $n' = \nodeAtPos{g}{\pi}$. Hence, we can
  employ the induction hypothesis for $n'$ to obtain the following:
  \begin{align*}
    \phi_1(n) &= \gsuc^h_i(\phi_1(n'))
    \tag{successor condition for $\phi_1$}\\
    &= \gsuc^h_i(\phi_2(n'))
    \tag{ind. hyp.} \\
    &= \phi_2(n)
    \tag{successor condition for $\phi_2$}
  \end{align*}
\end{proof}
As a consequence, each $\Delta$-homomorphism is both monic and epic,
and whenever there are two $\Delta$-homomorphisms $\phi\fcolon g
\homto_\Delta h$ and $\psi\fcolon h \homto_\Delta g$, they are
inverses of each other, i.e.\ \emph{$\Delta$-isomorphisms}. If two
term graphs are \emph{$\Delta$-isomorphic}, we write $g \isom_\Delta
h$.

Note that injectivity is in general different from both being monic
and the existence of left-inverses. The same holds for surjectivity
and being epic resp.\ having right-inverses. However, each
$\Delta$-homomorphism is a $\Delta$-isomorphism iff it is bijective.

For the two special cases $\Delta = \emptyset$ and $\Delta =
\set{\sigma}$, we write $\phi\fcolon g \homto h$ resp.\ $\phi\fcolon g
\homto_\sigma h$ instead of $\phi\fcolon g \homto_\Delta h$ and call
$\phi$ a homomorphism resp.\ $\sigma$-homomorphism. The same
convention applies to $\Delta$-isomorphisms.

\begin{lemma}[homomorphisms are surjective]
  \label{lem:homSurj}
  Every homomorphism $\phi\fcolon g \to h$, with $g, h \in \itgraphs$,
  is surjective.
\end{lemma}
\begin{proof}
  Follows from an easy induction on the depth of the nodes in $h$.
\end{proof}

Note that a bijective $\Delta$-homomorphism is not necessarily a
$\Delta$-isomorphism. To realise this, consider two term graphs $g,h$,
each with one node only. Let the node in $g$ be labelled with $a$ and
the node in $h$ with $b$ then the only possible $a$-homomorphism from
$g$ to $h$ is clearly a bijection but not an $a$-isomorphism. On the
other hand, bijective homomorphisms are isomorphisms.

\begin{lemma}[bijective homomorphisms are isomorphisms]
  \label{lem:isomBij}
  Let $g,h \in \itgraphs$ and $\phi\fcolon g \homto h$. Then the following are
  equivalent
  \begin{enumerate}[(a)]
  \item $\phi$ is an isomorphism. \label{item:isomBij-a}
  \item $\phi$ is bijective. \label{item:isomBij-b}
  \item $\phi$ is injective. \label{item:isomBij-c}
  \end{enumerate}
\end{lemma}
\begin{proof}
  The implication (\ref{item:isomBij-a}) $\Rightarrow$
  (\ref{item:isomBij-b}) is trivial. The equivalence
  (\ref{item:isomBij-b}) $\Leftrightarrow$ (\ref{item:isomBij-c})
  follows from Lemma~\ref{lem:homSurj}. For the implication
  (\ref{item:isomBij-b}) $\Rightarrow$ (\ref{item:isomBij-a}), consider
  the inverse $\phi^{-1}$ of $\phi$. We need to show that $\phi^{-1}$
  is a homomorphism from $h$ to $g$. The root condition follows
  immediately from the root condition for $\phi$. Similarly, an easy
  equational reasoning reveals that the fact that $\phi$ is
  homomorphic in $N^g$ implies that $\phi^{-1}$ is homomorphic in
  $N^h$
\end{proof}

\subsection{Canonical Term Graphs}
\label{sec:canon-term-graphs}

In this section, we introduce a canonical representation of
isomorphism classes of term graphs. We use a well-known trick to
achieve this \cite{plump99hggcbgt}. As we shall see at the end of this
section, this will also enable us to construct term graphs modulo
isomorphism very easily.
\begin{definition}[canonical term graph]
  \label{def:canTgraph}
  A term graph $g$ is called \emph{canonical} if $n = \nodePos{g}{n}$
  holds for each $n \in N^g$. That is, each node is the set of its
  positions in the term graph. The set of all canonical term graphs
  over $\Sigma$ is denoted $\ictgraphs$.
\end{definition}

This structure allows a convenient characterisation of
$\Delta$-homomorphisms:
\begin{lemma}[characterisation of $\Delta$-homomorphisms]
  \label{lem:canhom}
  For $g,h\in \ictgraphs$, a function $\phi\fcolon N^g \funto N^h$ is
  a $\Delta$-homomorphism $\phi\fcolon g \homto_\Delta h$ iff the
  following holds for all $n \in N^g$:
  \begin{center}
    \begin{inparaenum}[(a)]
    \item $n \subseteq \phi(n)$,\quad and
      \label{item:canhom1}
      \qquad
    \item $\glab^g(n) = \glab^h(\phi(n))$ \quad whenever \quad
      $\glab^g(n)\nin\Delta$.
      \label{item:canhom2}
    \end{inparaenum}
  \end{center}
\end{lemma}
\begin{proof}
  \def\itema{(\ref{item:canhom1})}%
  \def\itemb{(\ref{item:canhom2})}%
  For the ``only if'' direction, assume that $\phi\fcolon g
  \homto_\Delta h$. \itemb{} is the labelling condition and is
  therefore satisfied by $\phi$. To establish \itema{}, we show the
  equivalent statement
  \[
  \forall \pi \in \pos{g}.\; \forall n \in N^g. \; \pi \in n \implies
  \pi \in \phi(n)
  \]
  We do so by induction on the length of $\pi$: If $\pi = \emptyseq$,
  then $\pi \in n$ implies $n = r^g$. By the root condition, we have
  $\phi(r^g)=r^h$ and, therefore, $\pi = \emptyseq \in r^h$. If $\pi =
  \pi' \concat \seq i$, then let $n' = \nodeAtPos{g}{\pi'}$. Consequently,
  $\pi' \in n'$ and, by induction hypothesis, $\pi' \in
  \phi(n')$. Since $\pi = \pi'\concat \seq i$, we have $\gsuc^g_i(n') =
  n$. By the successor condition we can conclude $\phi(n) =
  \gsuc^h_i(\phi(n'))$. This and $\pi' \in \phi(n')$ yields that
  $\pi'\concat \seq i \in \phi(n)$.

  For the ``if'' direction, we assume \itema{} and \itemb{}. The
  labelling condition follows immediately from \itemb{}. For the root
  condition, observe that since $\emptyseq \in r^g$, we also have
  $\emptyseq \in \phi(r^g)$. Hence, $\phi(r^g) = r^h$. In order to show
  the successor condition, let $n, n' \in N^g$ and $0 \le i <
  \rank{g}{n}$ such that $\gsuc^g_i(n) = n'$. Then there is a
  position $\pi \in n$ with $\pi \concat \seq i \in n'$. By \itema{}, we
  can conclude that $\pi \in \phi(n)$ and $\pi \concat \seq i \in \phi(n')$
  which implies that $\gsuc^h_i(\phi(n)) = \phi(n')$.
\end{proof}

By Proposition~\ref{prop:catgraph}, there is at most one
$\Delta$-homomorphism between two term graphs. The lemma above
uniquely defines this $\Delta$-homomorphism: If there is a
$\Delta$-homomorphism from $g$ to $h$, it is defined by $\phi(n) =
n'$, where $n'$ is the unique node $n' \in N^h$ with $n \subseteq n'$.

\begin{rem}
  \label{rem:canhom}
  Note that the lemma above is also applicable to non-canonical term
  graphs. It simply has to be rephrased such that instead of just
  referring to a node $n$, its set of positions $\nodePos{g}{n}$ is
  referred to whenever the ``inner structure'' of $n$ is used.
\end{rem}

The set of nodes in a canonical term graph forms a partition of the
set of positions. Hence, it defines an equivalence relation on the
set of positions. For a canonical term graph $g$, we write $\sim_g$
for this equivalence relation on $\pos{g}$. According to
Remark~\ref{rem:canhom}, we can extend this to arbitrary term graphs:
$\pi_1 \sim_g \pi_2$ iff $\nodeAtPos{g}{\pi_1} =
\nodeAtPos{g}{\pi_2}$. The characterisation of $\Delta$-homomorphisms
can thus be recast to obtain the following lemma that characterises
the \emph{existence} of $\Delta$-homomorphisms:
\begin{lemma}[characterisation of $\Delta$-homomorphisms]
  \label{lem:occrephom}
  Given $g,h\in \itgraphs$, there is a $\Delta$-homomorphism
  $\phi\fcolon g \homto_\Delta h$ iff, for all $\pi,\pi'\in\pos{g}$,
  we have
  \begin{center}
    \begin{inparaenum}[(a)]
    \item $\pi \sim_g \pi' \quad \implies \quad \pi \sim_h \pi'$, and
      \quad
      \label{item:occrephom1}
    \item $g(\pi) = h(\pi)$ \quad whenever \quad $g(\pi) \nin \Delta$.
      \label{item:occrephom2}
    \end{inparaenum}
  \end{center}

\end{lemma}
\begin{proof}
  \def\itema{(\ref{item:occrephom1})}%
  \def\itemb{(\ref{item:occrephom2})}%
  \def\itemca{(\ref{item:canhom1})}%
  \def\itemcb{(\ref{item:canhom2})}%
  \def\itemap{(\ref{item:canhom1}')}%
  \def\itembp{(\ref{item:canhom2}')}%
  W.l.o.g.\ we assume $g$ and $h$ to be canonical. For the ``only if''
  direction, assume that $\phi$ is a $\Delta$-homomorphism from $g$ to
  $h$. Then we can use the properties \itemca{} and \itemcb{} of
  Lemma~\ref{lem:canhom}, which we will refer to as \itemap{} and \itembp{} to
  avoid confusion. In order to show \itema{}, assume $\pi \sim_g \pi'$. Then
  there is some node $n \in N^g$ with $\pi,\pi' \in n$. \itemap{} yields
  $\pi,\pi' \in \phi(n)$ and, therefore, $\pi \sim_g \pi'$. To show
  \itemb{}, we assume some $\pi \in \pos{g}$ with $g(\pi) \nin \Delta$. Then
  we can reason as follows:
  \[
  g(\pi) = \glab^g(\nodeAtPos{g}{\pi}) \stackrel{\text{\itembp}}{=}
  \glab^h(\phi(\nodeAtPos{g}{\pi}))  \stackrel{\text{\itemap}}{=}
  \glab^h(\nodeAtPos{h}{\pi}) = h(\pi)
  \]

  For the converse direction, assume that both \itema{} and \itemb{}
  hold. Define the function $\phi\fcolon N^g \funto N^h$ by $\phi(n) =
  n'$ iff $n \subseteq n'$ for all $n \in N^g$ and $n' \in N^h$. To
  see that this is well-defined, we show at first that, for each $n \in
  N^g$, there is at most one $n' \in N^h$ with $n \subseteq
  n'$. Suppose there is another node $n'' \in N^h$ with $n \subseteq
  n''$. Since $n \neq \emptyset$, this implies $n' \cap n'' \neq
  \emptyset$. Hence, $n' = n''$. Secondly, we show that there is at
  least one such node $n'$. Choose some $\pi^* \in n$. Since then
  $\pi^* \sim_g \pi^*$ and, by \itema{}, also $\pi^* \sim_h \pi^*$
  holds, there is some $n' \in N^h$ with $\pi^* \in n'$. For each $\pi
  \in n$, we have $\pi^* \sim_g \pi$ and, therefore, $\pi^* \sim_h \pi$
  by \itema{}. Hence, $\pi \in n'$. So we know that $\phi$ is
  well-defined. By construction, $\phi$ satisfies \itemap{}. Moreover,
  because of \itemb{}, it is also easily seen to satisfy
  \itembp{}. Hence, $\phi$ is a homomorphism from $g$ to $h$.
\end{proof}

Intuitively, (a) means that $h$ has at least as
much sharing of nodes as $g$ has, whereas (b) means that $h$ has at
least the same non-$\Delta$-symbols as $g$.

\begin{corollary}[characterisation of $\Delta$-isomorphisms]
  \label{cor:isomOcc}
  Given $g,h \in \itgraphs$, the following holds:
  \begin{enumerate}[(i)]
  \item $\phi\fcolon N^g \funto N^h$ is a $\Delta$-isomorphism \quad
    iff \quad
    for all $n \in N^g$\\
    \begin{inparaenum}[(a)]
    \item $\nodePos{h}{\phi(n)} = \nodePos{g}{n}$, and \quad
    \item $\glab^g(n) = \glab^h(\phi(n))$ or
      $\glab^g(n),\glab^h(\phi(n))\in\Delta$.
    \end{inparaenum}
    \label{item:isomOcc1}
  \item  $g \isom_\Delta h$ \quad iff \quad
    \begin{inparaenum}[(a)]
    \item $\sim_g\ =\ \sim_h$, and \qquad
    \item $g(\pi) = h(\pi)$ or $g(\pi),h(\pi) \in \Delta$.
    \end{inparaenum}
    \label{item:isomOcc2}
  \end{enumerate}
\end{corollary}
\begin{proof}
  Immediate consequence of Lemma~\ref{lem:canhom} resp.\
  Lemma~\ref{lem:occrephom} and
  Proposition~\ref{prop:catgraph}.
\end{proof}
From (\ref{item:isomOcc2}) we immediately obtain the following
equivalence:
\begin{corollary}
  \label{cor:sig-isom-isom}
  Given $g,h \in \itgraphs$ and $\sigma \in \Sigma^{(0)}$, we have $g
  \isom h$ iff $g \isom_\sigma h$.
\end{corollary}

Now we can revisit the notion of canonical term graphs using the above
characterisation of $\Delta$-isomorphisms. We will define a function
$\canon{\cdot}\fcolon \itgraphs \funto \ictgraphs$ that maps a term
graph to its canonical representation. To this end, let $g =
(N,\glab,\gsuc,r)$ be a term graph and define
$\canon{g}=(N',\glab',\gsuc',r')$ as follows:
\begin{align*}
  &N' = \setcom{\nodePos{g}{n}}{n \in N} & &r' = \nodePos{g}{r} \\
  &\glab'(\nodePos{g}{n}) = \glab(n) & &\gsuc'_i(\nodePos{g}{n}) =
  \nodePos{g}{\gsuc_i(n)} & \text{for all } n \in N, 0 \le i < \rank{g}{n}
\end{align*}
$\canon{g}$ is obviously a well-defined canonical term graph.  With
this definition we indeed capture the idea of a canonical
representation of isomorphism classes:

\begin{proposition}[canonical partial term graphs are a canonical representation]
  \label{prop:canon}
  Given $g \in \itgraphs$, the term graph $\canon{g}$ canonically
  represents the equivalence class $\eqc{g}{\isom}$. More precisely,
  it holds that
  \begin{center}
    \begin{inparaenum}[(i)]
    \item $\eqc{g}{\isom} = \eqc{\canon{g}}{\isom}$, and\qquad
    \item $\eqc{g}{\isom} = \eqc{h}{\isom}$ \quad iff \quad $\canon{g}
      = \canon{h}$.
    \end{inparaenum}
  \end{center}
  In particular, we have, for all canonical term graphs $g,h$, that $g =
  h$ iff $g \isom h$.
\end{proposition}
\begin{proof}
  Straightforward consequence of Corollary~\ref{cor:isomOcc}.
\end{proof}

\begin{rem}
  \label{rem:catTgraphs}
  $\Delta$-homomorphisms can be naturally lifted to
  $\quotient{\itgraphs}{\isom}$: We say that two
  $\Delta$-homomorphisms $\phi\fcolon g \homto_\Delta h$,
  $\phi'\fcolon g' \homto_\Delta h'$, are isomorphic, written $\phi
  \isom \phi'$ iff there are isomorphisms $\psi_1\fcolon g \isoto g'$
  and $\psi_2\fcolon h' \isoto h$ such that $\phi = \psi_2 \circ \phi'
  \circ \psi_1$. Given a $\Delta$-homomorphism $\phi\fcolon g
  \homto_\Delta h$ in $\itgraphs$, $\eqc{\phi}\isom\fcolon
  \eqc{g}\isom \homto_\Delta \eqc{h}\isom$ is a $\Delta$-homomorphism
  in $\quotient{\itgraphs}{\isom}$. These $\Delta$-homomorphisms then
  form a category which can easily be show to be isomorphic to the
  category of $\Delta$-homomorphisms on $\ictgraphs$ via the mapping
  $\eqc{\cdot}\isom$.
\end{rem}

Corollary~\ref{cor:isomOcc} has shown that term graphs can be
characterised up to isomorphism by only giving the equivalence
$\sim_g$ and the labelling $g(\cdot)\fcolon \pi \mapsto g(\pi)$. This
observation gives rise to the following definition:

\begin{definition}[labelled quotient tree]
  \label{def:occRep}
  A \emph{labelled quotient tree} over signature $\Sigma$ is a triple
  $(P,l,\sim)$ consisting of a non-empty set $P \subseteq \nats^*$, a
  function $l\fcolon P \funto \Sigma$, and an equivalence relation
  $\sim$ on $P$ that satisfies the following conditions for all
  $\pi,\pi' \in P$ and $i \in \nats$:
  \begin{align*}
    \pi\concat \seq i \in P \quad &\implies \quad \pi \in P \quad \text{
      and } \quad i < \srank{l(\pi)}
    \tag{reachability} \\
    \pi \sim \pi' \quad &\implies \quad
    \begin{cases}
      l(\pi) = l(\pi') &\text{ and }\\
      \pi\concat \seq j \sim \pi' \concat \seq j &\text{ for all }\; j <
      \srank{l(\pi)}
    \end{cases}
    \tag{congruence}
  \end{align*}
\end{definition}

The following lemma confirms that labelled quotient trees uniquely
characterise any term graph up to isomorphism:
\begin{lemma}
  \label{lem:occrep}
  Each term graph $g \in \itgraphs$ induces a \emph{canonical labelled
    quotient tree} $(\pos{g},g(\cdot),\sim_g)$ over $\Sigma$. Vice
  versa, for each labelled quotient tree $(P,l,\sim)$ over $\Sigma$
  there is a unique canonical term graph $g\in \ictgraphs$ whose
  canonical labelled quotient tree is $(P,l,\sim)$, i.e.\ $\pos{g} =
  P$, $g(\pi) = l(\pi)$ for all $\pi \in P$, and $\sim_g\ =\ \sim$.
\end{lemma}
\begin{proof}
  The first part is trivial: $(\pos{g},g(\cdot),\sim_g)$ satisfies the
  conditions from Definition~\ref{def:occRep}.
  
  Let $(P,l, \sim)$ be a labelled quotient tree. Define the term graph
  $g = (N,\glab,\gsuc,r)$ by
  \begin{align*}
    N &= \quotient{P}{\sim} &
    \glab(n) = f \quad &\text{ iff } \quad \exists \pi \in n.\;l(\pi)
    = f\\
    r = n \quad &\text{ iff } \quad \emptyseq \in n &
    \gsuc_i(n) = n' \quad &\text{ iff } \quad \exists \pi \in
    n.\; \pi\concat \seq i \in n'
  \end{align*}
  The functions $\glab$ and $\gsuc$ are well-defined due to the
  congruence condition satisfied by $(P,l, \sim)$. Since $P$ is
  non-empty and closed under prefixes, it contains $\emptyseq$. Hence,
  $r$ is well-defined. Moreover, by the reachability condition, each
  node in $N$ is reachable from the root node. An easy induction proof
  shows that $\nodePos{g}{n} = n$ for each node $n \in N$. Thus, $g$
  is a well-defined canonical term graph. The canonical labelled
  quotient tree of $g$ is obviously $(P,l, \sim)$. Whenever there are
  two canonical term graphs with labelled quotient tree $(P,l, \sim)$,
  they are isomorphic due to Corollary~\ref{cor:isomOcc} and,
  therefore, have to be identical by Proposition~\ref{prop:canon}.
\end{proof}

Labelled quotient trees provide a valuable tool for constructing
canonical term graphs. Nevertheless, the original graph representation
remains convenient for practical purposes as it allows a
straightforward formalisation of term graph rewriting and provides a
finite representation of finite cyclic term graphs which induce an
infinite labelled quotient tree.

Before we continue, it is instructive to make the correspondence
between terms and term graphs clear. Note, that there is an obvious
one-to-one correspondence between canonical term \emph{trees} and
terms. For example, the term tree depicted in
Figure~\ref{fig:exTermTree} corresponds to the term $f(a,h(a,b))$. We
thus consider the set of terms $\iterms$ to be the subset of canonical
term trees of $\ictgraphs$.

With this correspondence in mind, we can define the \emph{unravelling}
of a term graph $g$ as the unique term $t$ such that there is a
homomorphism $\phi\fcolon t \homto g$. The unravelling of cyclic term
graphs yields infinite terms, e.g.\ in Figure~\ref{fig:gredEx} on
page~\pageref{fig:gredEx}, the term $h_\omega$ is the unravelling of
the term graph $g_2$. We use the notation $\unrav{g}$ for the
unravelling of $g$.

Another convenience for dealing with term graphs is a linear notation
that makes it easy to write down (canonical) term graphs instead of
using the formal definition or a drawing. The notation that we use is
based on the linear notation for graphs by Barendregt et
al.\cite{barendregt87parle}:
\begin{definition}
  \label{def:linGraphNot}
  Let $\Sigma$ be a signature, $\calN$ a countably infinite set (of
  names) disjoint from $\Sigma$ and $\oh \Sigma$ a signature such $n
  \in \oh\Sigma^{(0)}$ and $f, f^n \in \oh\Sigma^{(k)}$ for each
  $n\in\calN$, $k \in \nats$ and $f \in \Sigma^{(k)}$. A \emph{linear
    notation for a canonical term graph} in $\ictgraphs$ is a term $t
  \in \iterms[\oh\Sigma]$ such that for each $n\in\calN$ that occurs
  in $t$, there is exactly one occurrence of a function symbol of the
  form $\nn n f$ in $t$.

  For each such linear notation $t$ we define the corresponding
  canonical term graph $g$ as follows: Consider the term tree
  representation of $t$ with the root node $r$. Redirect every edge to
  a node labelled $n$ to the unique node labelled $\nn n f$. Then,
  change all labellings of the form $f^n$ to $f$. After removing all
  nodes not reachable from the node $r$, define $g$ as the canonical
  term graph of the thus obtained term graph rooted in $r$.

  We use $n,m$ and primed resp.\ indexed variants thereof to denote
  names in $\calN$.
\end{definition}

Intuitively, in a linear notation for a term graph, a subterm $n$
denotes a pointer to a subterm with the corresponding name $n$, i.e.\
a subterm of the form $f^n(t_1,\dots,t_k)$.

\begin{example}
  Consider the term graph in Figure~\ref{fig:exTermGraph}. This term
  graph can be described by the linear notation
  $\nn{n_1}f(\nn{n_2}h(n_1,\nn{n_3}c),f(n_2,n_3))$. On the other hand,
  $\nn{n_1}f(n_1,n_2)$ and $f(\nn n a, \nn n b)$ are not valid linear
  notations.
\end{example}

Note that every term $t\in \iterms$ is a linear notation for the
corresponding term tree in $\ictgraphs$.

\section{Partial Order on Term Graphs}
\label{sec:partial-order-lebot1}

In this section, we want to establish a partial order suitable for
formalising convergence of sequences of canonical term graphs
similarly to $\prs$-convergence on terms.

In previous work, we have studied several different partial orders on
term graphs and the notion of convergence they induce
\cite{bahr11rta}. All of these partial orders have in common that they
are based on $\bot$-homomorphisms. This approach is founded on the
observation that if we consider terms as term trees, then
$\bot$-homomorphisms characterise the partial order on terms:
\[
s \lebot t \iff \text{ there is a $\bot$-homomorphism } \phi\fcolon s
\homto_\bot t.
\]
Thus $\bot$-homomorphisms constitute the ideal tool to define a
partial order on partial term graphs, i.e.\ term graphs over the
signature $\Sigma_\bot = \Sigma \uplus \set{\bot}$.

In this paper, we focus on the simplest among these partial orders on
term graphs:
\begin{definition}
  The relation $\lebotg$ on $\iptgraphs$ is defined as follows: $g
  \lebotg h$ iff there is a $\bot$-homomorphism $\phi\fcolon g
  \homto_\bot h$.
\end{definition}
\begin{proposition}[partial order $\lebotg$]
  The relation $\lebotg$ is a partial order on $\ipctgraphs$.
\end{proposition}
\begin{proof}
  Transitivity and reflexivity of $\lebotg$ follows immediately from
  Proposition~\ref{prop:catgraph}. For antisymmetry, consider $g,h \in
  \ipctgraphs$ with $g \lebotg h$ and $h \lebotg g$. Then, by
  Proposition~\ref{prop:catgraph}, $g \isom_\bot h$. This is
  equivalent to $g \isom h$ by Corollary~\ref{cor:sig-isom-isom} from which we
  can conclude $g = h$ using Proposition~\ref{prop:canon}.
\end{proof}

In our previous attempts to formalise convergence on term graphs
\cite{bahr11rta}, this partial order was rejected as the induced
notion of convergence manifests some unintuitive behaviour. However,
as we will show in Section~\ref{sec:strong-convergence}, theses quirks
will vanish when we move to strong convergence.

Before we study the properties of the partial order $\lebotg$, it is
helpful to make its characterisation in terms of labelled quotient
trees explicit:
\begin{corollary}[characterisation of $\lebotg$]
  \label{cor:chaTgraphPoA}
  Let $g,h \in \ipctgraphs$. Then $g \lebotg h$ iff the following
  conditions are met:
  \begin{enumerate}[(a)]
  \item $\pi \sim_g \pi' \quad \implies \quad  \pi \sim_h \pi'$ \quad for all
    $\pi,\pi' \in \pos{g}$
    \label{item:chaTGraphPoA1}
 \item $g(\pi) = h(\pi)$ \quad for all $\pi\in \pos{g}$ with $g(\pi) \in
   \Sigma$.
    \label{item:chaTGraphPoA2}
  \end{enumerate}
\end{corollary}
\begin{proof}
  This follows immediately from Lemma~\ref{lem:occrephom}.
\end{proof}%

Note that the partial order $\lebot$ on terms is entirely
characterised by (\ref{item:chaTGraphPoA2}). That is, the partial
order $\lebotg$ is simply the partial order $\lebot$ on its underlying
tree structure (i.e.\ its unravelling) plus the preservation of
sharing as stipulated by (\ref{item:chaTGraphPoA1}).

Next, we will show that the partial order on term graphs has the
properties that make it suitable as a basis for $\prs$-convergence,
i.e.\ that it forms a complete semilattice. At first we show its cpo
structure:
\begin{theorem}
  \label{thr:lebot1cpo}
  The relation $\lebotg$ is a complete partial order on
  $\ipctgraphs$. In particular, it has the least element $\bot$, and
  the least upper bound of a directed set $G$ is given by the
  following labelled quotient tree $(P,l,\sim)$:
  \begin{gather*}
    P = \bigcup\limits_{g\in G} \pos{g} \hspace{30pt}%
    \sim\ = \bigcup\limits_{g\in G} \sim_g \hspace{30pt}%
    l(\pi) =
    \begin{cases}
      f & \text{ if } f \in \Sigma \text{ and } \exists g \in G. \;
      g(\pi) = f \\
      \bot &\text{ otherwise }
    \end{cases}
  \end{gather*}
\end{theorem}
\begin{proof}
  The least element of $\lebotg$ is obviously $\bot$. Hence, it
  remains to be shown that each each directed subset of $\ipctgraphs$
  has a least upper bound. To this end, suppose that $G$ is a directed
  subset of $\ipctgraphs$. We define a canonical term graph $\ol g$ by
  giving the labelled quotient tree $(P,l,\sim)$

  \def\lebota{(\ref{item:chaTGraphPoA1})}
  \def\lebotb{(\ref{item:chaTGraphPoA2})}

  In order to show that the canonical term graph $\ol g$ given by the
  labelled quotient tree $(P,l,\sim)$ above is indeed the lub of $G$,
  we will make extensive use of
  Corollary~\ref{cor:chaTgraphPoA}. Therefore, we use \lebota{} and
  \lebotb{} to refer to the conditions mentioned there.

  At first we need to show that $l$ is indeed well-defined. For this
  purpose, let $g_1,g_2 \in G$ and $\pi\in \pos{g_1}\cap\pos{g_2}$
  with $g_1(\pi), g_2(\pi) \in\Sigma$. Since $G$ is directed, there is
  some $g \in G$ such that $g_1, g_2 \lebotg g$. By \lebotb{}, we can
  conclude $g_1(\pi) = g(\pi) = g_2(\pi)$.
  
  Next we show that $(P,l,\sim)$ is indeed a labelled quotient
  tree. Recall that $\sim$ needs to be an equivalence relation. For
  the reflexivity, assume that $\pi \in P$. Then there is some $g \in
  G$ with $\pi \in \pos{g}$. Since $\sim_g$ is an equivalence
  relation, $\pi \sim_g \pi$ must hold and, therefore, $\pi \sim
  \pi$. For the symmetry, assume that $\pi_1 \sim \pi_2$. Then there
  is some $g\in G$ such that $\pi_1 \sim_g \pi_2$. Hence, we get
  $\pi_2 \sim_g \pi_1$ and, consequently, $\pi_2 \sim \pi_1$. In order
  to show transitivity, assume that $\pi_1 \sim \pi_2, \pi_2 \sim
  \pi_3$. That is, there are $g_1, g_2\in G$ with $\pi_1 \sim_{g_1}
  \pi_2$ and $\pi_2 \sim_{g_2} \pi_3$. Since $G$ is directed, we find
  some $g\in G$ such that $g_1,g_2\lebotg g$. By \lebota{}, this
  implies that also $\pi_1 \sim_{g} \pi_2$ and $\pi_2 \sim_{g}
  \pi_3$. Hence, $\pi_1 \sim_{g} \pi_3$ and, therefore, $\pi_1 \sim
  \pi_3$.

  For the reachability condition, let $\pi\concat \seq i \in P$. That is,
  there is a $g \in G$ with $\pi \concat \seq i \in \pos{g}$. Hence, $\pi
  \in \pos{g}$, which in turn implies $\pi \in P$. Moreover, $\pi
  \concat \seq i \in \pos{g}$ implies that $i < \srank{g(\pi)}$. Since
  $g(\pi)$ cannot be a nullary symbol and in particular not $\bot$, we
  obtain that $l(\pi) = g(\pi)$. Hence, $i < \srank{l(\pi)}$.

  For the congruence condition, assume that $\pi_1 \sim \pi_2$ and that
  $l(\pi_1) = f$. If $f \in \Sigma$, then there are $g_1,g_2 \in G$
  with $\pi_1 \sim_{g_1} \pi_2$ and $g_2(\pi_1) = f$. Since $G$ is
  directed, there is some $g\in G$ such that $g_1, g_2 \lebotg
  g$. Hence, by \lebota{} resp.\ \lebotb{}, we have $\pi_1 \sim_{g}
  \pi_2$ and $g(\pi_1) = f$. Using Lemma~\ref{lem:occrep} we can
  conclude that $g(\pi_2) = g(\pi_1) = f$ and that $\pi_1 \concat \seq i
  \sim_g \pi_2 \concat \seq i$ for all $i < \srank{g(\pi_1)}$. Because $g \in
  G$, it holds that $l(\pi_2) = f$ and that $\pi_1 \concat \seq i \sim \pi
  \concat \seq i$ for all $i < \srank{l(\pi_1)}$. If $f = \bot$, then also
  $l(\pi_2) = \bot$, for if $l(\pi_2) = f'$ for some $f' \in\Sigma$,
  then, by the symmetry of $\sim$ and the above argument (for the case
  $f\in\Sigma$), we would obtain $f = f'$ and, therefore, a
  contradiction. Since $\bot$ is a nullary symbol, the remainder of
  the condition is vacuously satisfied.

  This shows that $(P,l,\sim)$ is a labelled quotient tree which, by
  Lemma~\ref{lem:occrep}, uniquely defines a canonical term graph. In
  order to show that the thus obtained term graph $\overline{g}$ is an
  upper bound for $G$, we have to show that $g \lebotg \overline{g}$ by
  establishing \lebota{} and \lebotb. This is an immediate consequence
  of the construction.

  In the final part of this proof, we will show that $\ol g$ is the
  least upper bound of $G$. For this purpose, let $\hat{g}$ be an
  upper bound of $G$, i.e.\ $g \lebotg \hat{g}$ for all $g\in G$. We
  will show that $\overline{g} \lebotg \hat{g}$ by establishing \lebota{}
  and \lebotb. For \lebota{}, assume that $\pi_1 \sim \pi_2$. Hence,
  there is some $g \in G$ with $\pi_1 \sim_g \pi_2$. Since, by
  assumption, $g \lebotg \oh g$, we can conclude $\pi_1 \sim_{\oh g}
  \pi_2$ using \lebota{}. For \lebotb{}, assume $\pi \in P$ and $l(\pi)
  = f \in\Sigma$. Then there is some $g\in G$ with $g(\pi) =
  f$. Applying \lebotb{} then yields $\oh g(\pi) = f$ since $g \lebotg
  \oh g$.
\end{proof}

The following proposition shows that the partial order $\lebotg$ also
admits glbs of arbitrary non-empty sets:
\begin{proposition}
  \label{prop:lebot1glb}
  In the partially ordered set $(\ipctgraphs,\lebotg)$ every non-empty
  set has a glb. In particular, the glb of a non-empty set $G$ is
  given by the following labelled quotient tree $(P,l,\sim)$:
  \begin{align*}
    P &= \setcom{\pi \in \bigcap_{g \in G} \pos{g}}{\forall \pi' <
      \pi\exists f \in \Sigma_\bot\forall g \in G: g(\pi') = f}\\
    l(\pi) &=
    \begin{cases}
      f &\text{if } \forall g \in G: f = g(\pi)\\
      \bot &\text{otherwise}
    \end{cases}\qquad\qquad
    \sim\ = \bigcap_{g\in G} \sim_g \cap\ P \times P
  \end{align*}
\end{proposition}
\begin{proof}
  At first we need to prove that $(P,l,\sim)$ is in fact a
  well-defined labelled quotient tree. That $\sim$ is an equivalence
  relation follows straightforwardly from the fact that each $\sim_g$
  is an equivalence relation.

  Next, we show the reachability and congruence properties from
  Definition~\ref{def:occRep}. In order to show the reachability
  property, assume some $\pi \concat \seq i \in P$. Then, for each
  $\pi' \le \pi$ there is some $f_{\pi'} \in \Sigma_\bot$ such that
  $g(\pi') = f_{\pi'}$ for all $g\in G$. Hence, $\pi \in P$. Moreover,
  we have in particular that $i < \srank{f_{\pi}} = \srank{l(\pi)}$.

  For the congruence condition, assume that $\pi_1 \sim \pi_2$. Hence,
  $\pi_1 \sim_g \pi_2$ for all $g \in G$. Consequently, we have for
  each $g \in G$ that $g(\pi_1) = g(\pi_2)$ and that $\pi_1 \concat
  \seq i \sim_g \pi_2 \concat \seq i$ for all $i < \srank{g(\pi_1)}$. We
  distinguish two cases: At first assume that there are some $g_1, g_2
  \in G$ with $g_1(\pi_1) \neq g_2(\pi_1)$. Hence, $l(\pi_2) =
  \bot$. Since, we also have that $g_1(\pi_2) = g_1(\pi_1) \neq
  g_2(\pi_1) = g_2(\pi_2)$ we can conclude that $l(\pi_2) = \bot =
  l(\pi_1)$. Since $\srank{\bot} = 0$ we are done for this case. Next,
  consider the alternative case that there is some $f \in \Sigma_\bot$
  such that $g(\pi_1) = f$ for all $g \in G$. Consequently, $l(\pi_1)
  = f$ and since also $g(\pi_2) = g(\pi_1) = f$ for all $g \in G$, we
  can conclude that $l(\pi_2) = f = l(\pi_1)$. Moreover, we obtain
  from the initial assumption for this case, that $\pi_1\concat \seq i,
  \pi_2 \concat \seq i \in P$ for all $i < \srank{f}$ which implies that
  $\pi_1 \concat \seq i \sim \pi_2 \concat \seq i$ for all $i < \srank{f} =
  \srank{l(\pi_1)}$.

  Next, we show that the term graph $\ol g$ defined by $(P,l,\sim)$ is
  a lower bound of $G$, i.e.\ that $\ol g \lebotg g$ for all $g \in
  G$. By Lemma~\ref{lem:occrephom}, it suffices to show $\sim \cap\ P
  \times P \subseteq\ \sim_g$ and $l(\pi) = g(\pi)$ for all $\pi \in
  P$ with $l(\pi) \in\Sigma$. Both conditions follow immediately from
  the construction of $\ol g$.
  
  Finally, we show that $\ol g$ is the greatest lower bound of $G$. To
  this end, let $\oh g \in \ipctgraphs$ with $\oh g \lebotg g$ for
  each $g \in G$. We will show that then $\oh g \lebotg \ol g$ using
  Lemma~\ref{lem:occrephom}. At first we show that $\pos{\oh g}
  \subseteq P$. Let $\pi \in \pos{\oh g}$. We know that $\oh g(\pi')
  \in\Sigma$ for all $\pi' < \pi$. According to
  Lemma~\ref{lem:occrephom}, using the assumption that $\oh g \lebotg
  g$ for all $g \in G$, we obtain that $g(\pi') = \oh g(\pi')$ for all
  $\pi' < \pi$. Consequently, $\pi \in P$. Next, we show part
  (\ref{item:occrephom1}) of Lemma~\ref{lem:occrephom}. Let
  $\pi_1,\pi_2 \in \pos{\oh g} \subseteq P$ with $\pi_1 \sim_{\oh g}
  \pi_2$. Hence, using the assumption that $\oh g$ is a lower bound of
  $G$, we have $\pi_1 \sim_g \pi_2$ for all $g\in G$ according to
  Lemma~\ref{lem:occrephom}. Consequently, $\pi_1 \sim \pi_2$. For
  part (\ref{item:occrephom2}) of Lemma~\ref{lem:occrephom} let $\pi
  \in \pos{\oh g} \subseteq P$ with $\oh g(\pi) = f \in\Sigma$. Using
  Lemma~\ref{lem:occrephom}, we obtain that $g(\pi) = f$ for all $g\in
  G$. Hence, $l(\pi) = f$.
\end{proof}

From this we can immediately derive the complete semilattice structure
of $\lebotg$:
\begin{theorem}
  \label{thr:complSemilattice}
  The partially ordered set $(\ipctgraphs,\lebotg)$ forms a complete
  semilattice.
\end{theorem}
\begin{proof}
  Follows from Theorem~\ref{thr:lebot1cpo} and
  Proposition~\ref{prop:lebot1glb}.
\end{proof}

In particular, this means that the limit inferior is defined for every
sequence of term graphs. Moreover, from the constructions given in
Theorem~\ref{thr:lebot1cpo} and Proposition~\ref{prop:lebot1glb}, we
can derive the following direct construction of the limit inferior:
\begin{corollary}
  \label{cor:lebot1Liminf}%
  The limit inferior of a sequence $(g_\iota)_{\iota < \alpha}$ over
  $\ipctgraphs$ is given by the following labelled quotient tree
  $(P,\sim,l)$:
  \begin{align*}
    P &= \bigcup_{\beta<\alpha}
    \setcom{\pi\in\pos{g_\beta}}{\forall\pi'<\pi\forall\beta\le\iota<\alpha\colon
      g_\iota(\pi') = g_\beta(\pi')}\\
    \sim\ &= \left(\bigcup_{\beta<\alpha}
      \bigcap_{\beta\le\iota<\alpha} \sim_{g_\iota}\right) \cap P \times P\\
    l(\pi) &=
    \begin{cases}
      g_\beta(\pi) & \text{if } \exists \beta< \alpha \forall \beta
      \le \iota <
      \alpha\colon g_\iota(\pi) = g_\beta(\pi)\\
      \bot &\text{otherwise}
    \end{cases}
     \qquad \text{for all } \pi \in P
  \end{align*}
  In particular, given $\beta<\alpha$ and $\pi\in \pos{g_\beta}$, we
  have that $g(\pi) = g_\beta(\pi)$ if $g_\iota(\pi') = g_\beta(\pi')$
  for all $\pi' \le \pi$ and $\beta \le \iota < \alpha$.
\end{corollary}

\begin{example}
  \label{ex:liminf}
  Figure~\ref{fig:gtransRed} on page~\pageref{fig:gtransRed} illustrates a sequence of term graphs
  $(h_\iota)_{\iota<\omega}$. Except for the edge to the root that
  closes a cycle each term graph $h_\iota$ as a tree structure. Since
  this edge is pushed down as the sequence progresses, it vanishes in
  the the limit inferior of $(h_\iota)_{\iota<\omega}$, depicted as
  $h_\omega$ in Figure~\ref{fig:gtransRed}.

  Changing acyclic sharing on the other hand exposes an oddity of the
  partial order $\lebotg$. Let $(g_\iota)_{\iota<\omega}$ be the
  sequence of term graphs illustrated in
  Figure~\ref{fig:convWeird}. The sequence alternates between $g_0$
  and $g_1$ which differ only in the sharing of the two arguments of
  the $f$ function symbol. Hence, there is an obvious homomorphism
  from $g_0$ to $g_1$ and we thus have $g_0 \lebotg g_1$. Therefore,
  $g_0$ is the greatest lower bound of every suffix of
  $(g_\iota)_{\iota<\omega}$, which means that
  $\liminf_{\iota\limto\omega} g_\iota = g_0$.
\end{example}
\begin{figure}
  \centering
  \begin{tikzpicture}[node distance=15mm]%
    \node (r1) {$f$} %
    child{ node (n1) {$c$} }%
    child{ node (n2) {$c$} };%
    \node[right=of r1] (r2) {$f$}%
    child {%
      node (n2) {$c$}%
      edge from parent[transparent] %
    };%
    \draw[->] (r2)%
    edge [bend right=25] (n2)%
    edge [bend left=25] (n2);%
    
    \draw[single step,shorten=5mm] (r1) -- (r2);%
    
    \node[right=of r2] (r3) {$f$} %
    child{ node (n1) {$c$} }%
    child{ node (n2) {$c$} };%
    
    \draw[single step,shorten=5mm] (r2) -- (r3);%
    
    \node[right=of r3] (r4) {$f$}%
    child {%
      node (n2) {$c$}%
      edge from parent[transparent] %
    };%
    \draw[->] (r4)%
    edge [bend right=25] (n2)%
    edge [bend left=25] (n2);%
    \draw[single step,shorten=5mm] (r3) -- (r4);%

    \node[node distance=25mm,right=of r4] (r5) {$f$} %
    child{ node (n1) {$c$} }%
    child{ node (n2) {$c$} };%
    
    \draw[dotted,thick,shorten=10mm] (r4) -- (r5);%
    \begin{scope}[node distance=1cm]
      \node[below=of r1] {$(g_0)$};
      \node[below=of r2] {$(g_1)$};
      \node[below=of r3] {$(g_2)$};
      \node[below=of r4] {$(g_4)$};
      \node[below=of r5] {$(g_\omega)$};
    \end{scope}
  \end{tikzpicture}
  \caption{Limit inferior in the presence of acyclic sharing.}
  \label{fig:convWeird}
\end{figure}
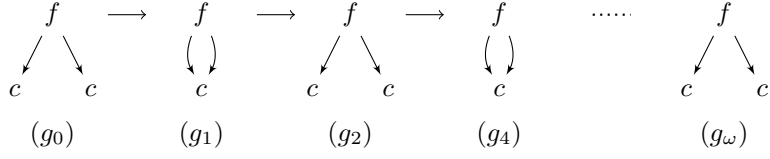

\section{Metric Spaces}
\label{sec:alternative-metric}

In this section, we shall define a metric space on canonical term
graphs. We base our approach to defining a metric distance on the
definition of the metric distance $\dd$ on terms.

Originally, Arnold and Nivat~\cite{arnold80fi} used a truncation
$\trunca{t}{d}$ of terms to define the metric on terms. The truncation
of a term $t$ at depth $d$ replaces all subterms at depth $d$ by
$\bot$:
\begin{align*}
  \trunca{t}{0} = \bot, \quad
  \trunca{f(t_1,\dots,t_k)}{d+1} =
  f(\trunca{t_1}{d},\dots,\trunca{t_k}{d}),\quad
  \trunca{t}{\infty} = t
\end{align*}

The similarity of two terms, on which the metric distance $\dd$ is
based, can thus be characterised via truncations:
\[
\similar{s}{t} = \max\setcom{d\in\nats\cup\set\infty}{\trunca{s}{d} =
  \trunca{t}{d}}
\]

We will adopt this approach for term graphs as well. To this end, we
will first define abstractly what a truncation on term graphs is and
how a metric distance can be derived from it. Then we show a concrete
truncation and show that the induced metric space is in fact
complete. We will conclude the section by showing that the metric
space we considered is robust in the sense that it is invariant under
small changes to the definition of truncation.

\subsection{Truncation Functions}
\label{sec:truncation-functions}

As we have seen above, the truncation on terms is a function that,
depending on a depth value $d$, transforms a term $t$ to a term
$\trunca{t}{d}$. We shall generalise this to term graphs and stipulate
some axioms that ensure that we can derive a metric distance in the
style of Arnold and Nivat~\cite{arnold80fi}:
\begin{definition}[truncation function]
  \label{def:truncFun}
  A family $\tau = (\tau_d\fcolon \iptgraphs \funto \iptgraphs)_{d \in \nats
    \cup \set{\infty}}$ of functions on term graphs is called a
  \emph{truncation function} if it satisfies the following properties
  for all $g,h \in \iptgraphs$ and $d \in \nats\cup\set{\infty}$:
  \begin{flushleft}
    \begin{inparaenum}[(a)]
    \item $\tau_0(g) \isom \bot$,\quad
      \label{item:truncFunI}
    \item $\tau_\infty(g) \isom g$, and\quad
      \label{item:truncFunII}
    \item $\tau_d(g) \isom \tau_d(h) \; \implies \; \tau_e(g) \isom
      \tau_e(h) \quad$ for all $e < d$.%
      \label{item:truncFunIII}
    \end{inparaenum}
  \end{flushleft}
\end{definition}

Note that from axioms (\ref{item:truncFunII}) and
(\ref{item:truncFunIII}) it follows that truncation functions must be
defined modulo isomorphism, i.e.\ $g\isom h$ implies $\tau_d(g) \isom
\tau_d(h)$ for all $d \in \nats\cup\set\infty$.

Given a truncation function, we can define a distance measure in the
style of Arnold and Nivat:
\begin{definition}[truncation-based similarity/distance]
  Let $\tau$ be a truncation function. The \emph{$\tau$-similarity} is
  the function $\abssim{\tau}\fcolon \iptgraphs \times \iptgraphs \funto
  \nats \cup \set{\infty}$ defined by
  \[
  \abssim{\tau}(g,h) =  \max\setcom{d \in \nats\cup\set\infty}{\tau_d(g) \isom \tau_d(h)}
  \]
  The \emph{$\tau$-distance} is the function $\dd_\tau\fcolon \iptgraphs
  \times \iptgraphs \funto \realsp$ defined by $\dd_\tau(g,h) =
  2^{-\abssim{\tau}(g,h)}$, where $2^{-\infty}$ is interpreted as
  $0$.
\end{definition}

Observe, that the similarity $\abssim{\tau}(g,h)$ induced by a
truncation function $\tau$ is well-defined since the
axiom~(\ref{item:truncFunI}) of Definition~\ref{def:truncFun}
guarantees that the set $\setcom{d \in \nats\cup\set\infty}{\tau_d(g)
  \isom \tau_d(h)}$ is not empty. The following proposition confirms
that the $\tau$-distance restricted to $\ictgraphs$ is indeed an
ultrametric:
\begin{proposition}[truncation-based ultrametric]
  \label{prop:truncMetric}
  For each truncation function $\tau$, the $\tau$-distance $\dd_\tau$
  constitutes an ultrametric on $\ictgraphs$.
\end{proposition}
\begin{proof}
  The identity resp.\ the symmetry condition follow by
  \begin{gather*}
    \dd_\tau(g,h) = 0 \iff \abssim{\tau}(g,h) = \infty \iff \tau_{\infty}(g)
    \isom \tau_{\infty}(h) \stackrel{(*)}\iff g \isom h
    \stackrel{\text{Prop.~\ref{prop:canon}}}\iff g = h,\quad\text{and}\\
    \dd_\tau(g,h) = 2^{-\abssim{\tau}(g,h)} = 2^{-\abssim{\tau}(h,g)} =
    \dd_\tau(h,g).
  \end{gather*}
  The equivalence ($*$) is valid by axiom~(\ref{item:truncFunII}) of
  Definition~\ref{def:truncFun}. For the strong triangle condition we
  have to show that
  \[
  \abssim{\tau}(g_1,g_3)  \ge \min\set{\abssim{\tau}(g_1,g_2), \abssim{\tau}(g_2,g_3)}.
  \]
  With $d = \min\set{\abssim{\tau}(g_1,g_2), \abssim{\tau}(g_2,g_3)}$
  we have, by axiom~(\ref{item:truncFunIII}) of
  Definition~\ref{def:truncFun}, that $\tau_d(g_1) \isom \tau_d(g_2)$
  and $\tau_d(g_2) \isom \tau_d(g_3)$. Since we have that $\tau_d(g_1)
  \isom \tau_d(g_3)$ then, we can conclude that
  $\abssim{\tau}(g_1,g_2) \ge d$.
\end{proof}

Given their particular structure, we can reformulate the definition of
Cauchy sequences and convergence in metric spaces induced by
truncation functions in terms of the truncation function itself:
\begin{lemma}
  \label{lem:cauchyTruncFun}%
  For each truncation function $\tau$, each $g \in (\ictgraphs,\dd_\tau)$,
  and each sequence $(g_\iota)_{\iota < \alpha}$ in
  $(\ictgraphs,\dd_\tau)$ the following holds:
  \begin{enumerate}[(i)]
  \item $(g_\iota)_{\iota < \alpha}$ is Cauchy iff for each
    $d\in\nats$ there is some $\beta < \alpha$ such that
    $\tau_d(g_\gamma) \isom \tau_d(g_{\iota})$ for all $\beta \le
    \gamma,\iota < \alpha$.
  \item $(g_\iota)_{\iota < \alpha}$ converges to $g$ iff for each
    $d\in\nats$ there is some $\beta < \alpha$ such that $\tau_d(g)
    \isom \tau_d(g_{\iota})$ for all $\beta \le \iota < \alpha$.
  \end{enumerate}
\end{lemma}
\begin{proof}
  We only show (i) as (ii) is essentially the same. For ``only if''
  direction assume that $(g_\iota)_{\iota < \alpha}$ is Cauchy and
  that $d\in\nats$. We then find some $\beta < \alpha$ such that
  $\dd_\tau(g_\gamma,g_{\iota}) < 2^{-d}$ for all $\beta \le
  \gamma,\iota < \alpha$. Hence, we obtain that
  $\abssim{\tau}(g_\gamma,g_{\iota}) > d$ for all $\beta \le
  \gamma,\iota < \alpha$. That is, $\tau_e(g_\gamma) \isom
  \tau_e(g_{\iota})$ for some $e > d$. According to axiom
  (\ref{item:truncFunIII}) of Definition~\ref{def:truncFun}, we can
  then conclude that $\tau_d(g_\gamma) \isom \tau_d(g_{\iota})$ for
  all $\beta \le \gamma,\iota < \alpha$.

  For the ``if'' direction assume some $\epsilon\in \realsp$. Then there
  is some $d \in \nats$ with $2^{-d}\le\epsilon$. By the initial
  assumption we find some $\beta < \alpha$ with $\tau_d(g_\gamma) \isom
  \tau_d(g_{\iota})$ for all $\beta \le \gamma,\iota < \alpha$, i.e.\
  $\abssim{\tau}(g_\gamma,g_\iota) \ge d$. Hence, we have that
  $\dd_\tau(g_\gamma,g_\iota) = 2^{\abssim{\tau}(g_\gamma,g_\iota)} <
  2^{-d}\le \epsilon$ for all $\beta \le \gamma,\iota < \alpha$.
\end{proof}

\subsection{The Strict Truncation and its Metric Space}
\label{sec:strict-trunc-funct}

In this section, we consider a straightforward truncation function
that simply cuts off all nodes at the given depth $d$.

\begin{definition}[strict truncation]
  \label{def:trunca}
  Let $g \in \iptgraphs$ and $d \in \nats \cup \set{\infty}$. The
  \emph{strict truncation} $\trunca{g}{d}$ of $g$ at $d$ is a term
  graph defined by
  \begin{align*}
    N^{\trunca{g}{d}} &= \setcom{n \in N^g}{\depth{g}{n} \le d}
    & r^{\trunca{g}{d}} &= r^g
    \\
    \glab^{\trunca{g}{d}}(n) &= 
    \begin{cases}
      \glab^g(n) &\text{if } \depth{g}{n} < d \\
      \bot  &\text{if } \depth{g}{n} = d
    \end{cases} &
    \gsuc^{\trunca{g}{d}}(n) &=
    \begin{cases}
      \gsuc^g(n) &\text{ if }\depth{g}{n} < d\\
      \emptyseq &\text{ if }\depth{g}{n} = d
    \end{cases}
  \end{align*}
\end{definition}

Figure~\ref{fig:exTrunc} on page~\pageref{fig:exTrunc} shows a term
graph $g$ and its strict truncation at depth $2$. Note that a node can
get truncated even though its successor is retained.

One can easily see that the truncated term graph $\trunca{g}{d}$ is
obtained from $g$ by relabelling all nodes at depth $d$ to $\bot$,
removing all their outgoing edges and then removing all nodes that
thus become unreachable from the root. This makes the strict
truncation a straightforward generalisation of the truncation on
terms.

The strict truncation indeed induces a truncation function:
\begin{proposition}
  \label{prop:truncaDown}
  Let $\trunca{}{}$ be the function with $\trunca{}{}_d(g) =
  \trunca{g}{d}$. Then $\trunca{}{}$ is a truncation function.
\end{proposition}
\begin{proof}
  (\ref{item:truncFunI}) and (\ref{item:truncFunII}) of
  Definition~\ref{def:truncFun} follow immediately from the
  construction of the truncation. For (\ref{item:truncFunIII}) assume
  that $\trunca{g}{d} \isom \trunca{h}{d}$. Let $0 \le e < d$ and let
  $\phi\fcolon \trunca{g}{d} \homto \trunca{h}{d}$ be the witnessing
  isomorphism. Note that strict truncations preserve the depth of
  nodes, i.e.\ $\depth{\trunca{g}{d}}{n} = \depth{g}{n}$ for all $n\in
  N^{\trunca{g}{d}}$. This can be shown by a straightforward induction
  on $\depth{g}{n}$. Moreover, by Corollary~\ref{cor:isomOcc} also
  isomorphisms preserve the depth of nodes. Hence, 
  \[
  \depth{h}{\phi(n)}
  = \depth{\trunca{h}{d}}{\phi(n)} = \depth{\trunca{g}{d}}{n} =
  \depth{g}{n}\quad\text{ for all } n \in N^{\trunca{g}{d}}
  \]
  Restricting $\phi$ to the nodes in $\trunca{g}{e}$ thus yields an
  isomorphism from $\trunca{g}{e}$ to $\trunca{h}{e}$.
\end{proof}

Next we show that the metric space $(\ictgraphs,\dda)$ that is induced
by the truncation function $\trunca{}{}$ is in fact complete. To do
this, we give a characterisation of the strict truncation in terms of
labelled quotient trees.
\begin{lemma}[labelled quotient tree of a strict truncation]
  \label{lem:truncaOccRep}
  Let $g \in \iptgraphs$ and $d \in \nats \cup \set{\infty}$. The
  strict truncation $\trunca{g}{d}$ is uniquely determined up to
  isomorphism by the labelled quotient tree $(P,l,\sim)$ with
  \begin{enumerate}[(a)]
  \item $P = \setcom{\pi \in \pos{g}}{\forall \pi_1 < \pi\exists \pi_2
      \sim_g \pi_1 \text{ with }\len{\pi_2} <
      d}$,\label{item:truncaOccRepI}
  \item $l(\pi) =
    \begin{cases}
      g(\pi) & \text{ if }\exists \pi'\sim_g\pi \text{ with }
      \len{\pi'} < d\\
      \bot &\text{ otherwise}
    \end{cases}$
    \label{item:truncaOccRepII}
  \item $\sim\ =\ \sim_g \cap\ P \times P$\label{item:truncaOccRepIII}
  \end{enumerate}
\end{lemma}
\begin{proof}
  We just have to show that $(P,l,\sim)$ is the canonical labelled
  quotient tree induced by $\trunca{g}{d}$. Then the lemma follows from
  Lemma~\ref{lem:occrep}. The case $d = \infty$ is trivial. In the
  following we assume that $d\in \nats$.

  Before continuing the proof, note that
  \begin{gather}
    \text{for each $\pi \in \pos{\trunca{g}{d}}$ we have that $\pi \in
      \pos{g}$ and $\nodeAtPos{\trunca{g}{d}}{\pi} =
      \nodeAtPos{g}{\pi}$.}  \tag{$*$}\label{eq:truncaOccRep}
  \end{gather}
  This can be shown by an induction on the length of $\pi$: The case
  $\pi = \emptyseq$ is trivial. If $\pi = \pi'\concat\seq{i}$, let $n
  = \nodeAtPos{\trunca{g}{d}}{\pi'}$ and $m =
  \nodeAtPos{\trunca{g}{d}}{\pi}$. Hence, $m =
  \gsuc_i^{\trunca{g}{d}}(n)$ and, by construction of $\trunca{g}{d}$,
  also $m = \gsuc_i^g(n)$. Since by induction hypothesis $n =
  \nodeAtPos{g}{\pi'}$, we can thus conclude that $\pi \in \pos{g}$
  and that $\nodeAtPos{g}{\pi} = m = \nodeAtPos{\trunca{g}{d}}{\pi}$.
  
  (\ref{item:truncaOccRepI}) $P = \pos{\trunca{g}{d}}$. For the ``$\subseteq$'' direction let $\pi \in
  P$. To show that $\pi \in \pos{\trunca{g}{d}}$, assume a $\pi_1 <
  \pi$ and let $n = \nodeAtPos{g}{\pi_1}$. Since $\pi \in P$, there is
  some $\pi_2 \sim_g \pi_1$ with $\len{\pi_2} < d$. That is,
  $\depth{g}{n} < d$. Therefore, we have that $n\in N^{\trunca{g}{d}}$
  and $\gsuc^{\trunca{g}{d}}(n) = \gsuc^g(n)$. Hence, each node on the
  path $\pi$ in $g$ is also a node in $\trunca{g}{d}$ and has the same
  successor nodes as in $g$. That is, $\pi \in \pos{\trunca{g}{d}}$.

  For the ``$\supseteq$'' direction, assume some $\pi \in
  \pos{\trunca{g}{d}}$. By \eqref{eq:truncaOccRep}, $\pi$ is also a
  position in $g$. To show that $\pi \in P$, let $\pi_1 < \pi$. Since
  only nodes of depth smaller than $d$ can have a successor node in
  $\trunca{g}{d}$, the node $\nodeAtPos{\trunca{g}{d}}{\pi_1}$ in
  $\trunca{g}{d}$ is at depth smaller than $d$. Hence, there is some
  $\pi_2 \sim_{\trunca{g}{d}} \pi_1$ with $\len{\pi_2} < d$. Because
  $\pi_2 \sim_{\trunca{g}{d}} \pi$ implies that $\pi_2 \sim_{g} \pi$,
  we can conclude that $\pi \in P$.

  (\ref{item:truncaOccRepII}) $l(\pi) = g(\pi)$ for all $\pi \in
  P$. Let $\pi \in P$ and $n = \nodeAtPos{g}{\pi}$. We distinguish two
  cases. At first suppose that there is some $\pi' \sim_g \pi$ with
  $\len{\pi'} < d$. Then $l(\pi) = g(\pi)$. Since $n =
  \nodeAtPos{g}{\pi'}$, we have that $\depth{g}{n} < d$. Consequently,
  $\glab^{\trunca{g}{d}}(n) = \glab^g(n)$ and, therefore,
  $\trunca{g}{d}(n) = g(\pi) = l(\pi)$. In the other case that there
  is no $\pi' \sim_g \pi$ with $\len{\pi} < d$, we have $l(\pi) =
  \bot$. This also means that $\depth{g}{n} = d$. Consequently,
  $\trunca{g}{d}(\pi) = \glab^{\trunca{g}{d}}(n) = \bot = l(\pi)$.

  (\ref{item:truncaOccRepIII}) $\sim\ =\ \sim_{\trunca{g}{d}}$. Using
  the fact that $P = \pos{\trunca{g}{d}}$, we can conclude for all
  $\pi_1,\pi_2 \in P$ that
  \[
  \pi_1 \sim_{\trunca{g}{d}} \pi_2 \iff
  \nodeAtPos{\trunca{g}{d}}{\pi_1} = \nodeAtPos{\trunca{g}{d}}{\pi_2}
  \stackrel{\text{\eqref{eq:truncaOccRep}}}\iff \nodeAtPos{g}{\pi_1} = \nodeAtPos{g}{\pi_2} \iff \pi_1 \sim_g
  \pi_2
  \]
\end{proof}

Notice that a position $\pi$ is retained by a truncation, i.e.\
$\pi\in P$, iff each node that $\pi$ passes through is at a depth
lower than $d$ (and is thus not truncated or relabelled).

From this characterisation we immediately obtain the following
relation between a term graph and its strict truncations:
\begin{corollary}
  \label{cor:truncaOccRep}
  Let $g \in \iptgraphs$ and $d \in \nats \cup \set{\infty}$. Then
  \begin{enumerate}[(i)]
  \item $\pi \in \pos{g}$ iff $\pi \in \pos{\trunca{g}{d}}$ for all
    $\pi$ with $\len{\pi} \le d$, and
    \label{item:truncaOccRep1}
  \item $\trunca{g}{d}(\pi) = g(\pi)$ for all $\pi \in \pos{g}$ with
    $\len{\pi} < d$.
    \label{item:truncaOccRep2}
  \end{enumerate}
\end{corollary}
\begin{proof}
  Using the reflexivity of $\sim_g$, (\ref{item:truncaOccRep1})
  follows immediately from
  Lemma~\ref{lem:truncaOccRep}~(\ref{item:truncaOccRepI}), and
  (\ref{item:truncaOccRep2}) follows immediately from
  Lemma~\ref{lem:truncaOccRep}~(\ref{item:truncaOccRepII}).
\end{proof}

We can now show that the metric space induced by the strict truncation
is complete:
\begin{theorem}
  \label{thr:smetricComplete}%
  The metric space $(\ictgraphs,\dda)$ is complete. In particular,
  each Cauchy sequence $(g_\iota)_{\iota < \alpha}$ in
  $(\ictgraphs,\dda)$ converges to the canonical term graph given by
  the following labelled quotient tree $(P,l,\sim)$:
  \begin{align*}
    P &= \liminf_{\iota \limto \alpha} \pos{g_\iota} = \bigcup_{\beta
      < \alpha}\bigcap_{\beta \le \iota < \alpha} \pos{g_\iota}
    \hspace{20pt} \sim\ = \liminf_{\iota \limto \alpha}
    \sim_{g_\iota}\ =
    \bigcup_{\beta < \alpha}\bigcap_{\beta \le \iota < \alpha} \sim_{g_\iota} \\
    l(\pi) &= g_\beta(\pi) \quad \text{for some } \beta<\alpha\text{
      with } g_\iota(\pi) = g_\beta(\pi) \text{ for each } \beta\le\iota<\alpha
    \qquad \text{for all } \pi \in P
  \end{align*}
\end{theorem}
\begin{proof}
  We need to check that $(P,l,\sim)$ is a well-defined labelled
  quotient tree. At first we show that $l$ is a well-defined function
  on $P$. In order to show that $l$ is functional, assume that there
  are $\beta_1,\beta_2 < \alpha$ such that there is a $\pi$ with
  $g_\iota(\pi) = g_{\beta_k}(\pi)$ for all $\beta_k \le \iota <
  \alpha$, $=1,2$. but then $g_{\beta_1}(\pi) = g_\beta(\pi) =
  g_{\beta_2}(\pi)$ for $\beta = \max\set{\beta_1,\beta_2}$.

  To show that $l$ is total on $P$, let $\pi \in P$ and $d =
  \len{\pi}$. By Lemma~\ref{lem:cauchyTruncFun}, there is some $\beta
  < \alpha$ such that $\trunca{g_\gamma}{d+1} \isom
  \trunca{g_{\iota}}{d+1}$ for all $\beta \le \gamma, \iota <
  \alpha$. According to Corollary~\ref{cor:truncaOccRep}, this means
  that all $g_\iota$ for $\beta \le \iota < \alpha$ agree on positions
  of length smaller than $d + 1$, in particular $\pi$. Hence,
  $g_\iota(\pi) = g_\beta(\pi)$ for all $\beta \le \iota <
  \alpha$, and we have $l(\pi) = g_\beta(\pi)$.

  One can easily see that $\sim$ is a binary relation on $P$: If
  $\pi_1 \sim \pi_2$, then there is some $\beta < \alpha$ with $\pi_1
  \sim_{g_\iota} \pi_2$ for all $\beta \le \iota < \alpha$. Hence,
  $\pi_1,\pi_2 \in \pos{g_\iota}$ for all $\beta \le \iota < \alpha$
  and thus $\pi_1,\pi_2 \in P$.

  Similarly follows that $\sim$ is an equivalence relation on $P$: To
  show reflexivity, assume $\pi \in P$. Then there is some $\beta <
  \alpha$ such that $\pi \in \pos{g_\iota}$ for all $\beta \le \iota <
  \alpha$. Hence, $\pi \sim_{g_\iota} \pi$ for all $\beta \le \iota <
  \alpha$ and, therefore, $\pi \sim \pi$. In the same way symmetry and
  transitivity follow from the symmetry and transitivity of
  $\sim_{g_\iota}$.

  Finally, we have to show the reachability and the congruence
  property from Definition~\ref{def:occRep}. To show reachability
  assume some $\pi\concat \seq i \in P$. Then there is some $\beta <
  \alpha$ such that $\pi\concat \seq i \in \pos{g_\iota}$ for all $\beta
  \le \iota < \alpha$. Hence, since then also $\pi \in \pos{g_\iota}$
  for all $\beta \le \iota < \alpha$, we have $\pi \in P$. According
  to the construction of $l$, there is also some $\beta \le \gamma <
  \alpha$ with $g_\gamma(\pi) = l(\pi)$. Since $\pi\concat \seq i \in
  \pos{g_\gamma}$ we can conclude that $i < \srank{l(\pi)}$.

  To establish congruence assume that $\pi_1 \sim
  \pi_2$. Consequently, there is some $\beta < \gamma$ such that
  $\pi_1 \sim_{g_\iota} \pi_2$ for all $\beta \le \iota <
  \alpha$. Therefore, we also have for each $\beta \le \iota < \alpha$
  that $\pi_1 \concat \seq i \sim_{g_\iota} \pi_2 \concat \seq i$ and that
  $g_\iota(\pi_1) = g_\iota(\pi_2)$. From the former we can
  immediately derive that $\pi_1\concat \seq i \sim \pi_2\concat \seq
  i$. Moreover, according to the construction of $l$, there some
  $\beta \le \gamma < \alpha$ such that $l(\pi_1) = g_\gamma(\pi_1) =
  g_\gamma(\pi_2) = l(\pi_2)$.

  This concludes the proof that $(P,l,\sim)$ is indeed a labelled
  quotient tree. Next, we show that the sequence $(g_\iota)_{\iota <
    \alpha}$ converges to the thus defined canonical term graph
  $g$. By Lemma~\ref{lem:cauchyTruncFun}, this amounts to giving for
  each $d \in \nats$ some $\beta < \alpha$ such that $\trunca{g}{d}
  \isom \trunca{g_\iota}{d}$ for each $\beta\le \iota < \alpha$.

  To this end, let $d \in \nats$. Since $(g_\iota)_{\iota < \alpha}$
  is Cauchy, there is, according to Lemma~\ref{lem:cauchyTruncFun},
  some $\beta < \alpha$ such that
  \begin{gather}
    \trunca{g_\iota}{d} \isom \trunca{g_{\iota'}}{d}\quad \text{ for all }
    \beta\le \iota,\iota' < \alpha.\tag{1}\label{eq:cauchy}
  \end{gather}
  In order to show that this implies that $\trunca{g}{d} \isom
  \trunca{g_\iota}{d}$ for each $\beta \le\iota < \alpha$, we show
  that the respective labelled quotient trees of $\trunca{g}{d}$ and
  $\trunca{g_\iota}{d}$ as characterised by
  Lemma~\ref{lem:truncaOccRep} coincide. The labelled quotient tree
  $(P_1,l_1,\sim_1)$ for $\trunca{g}{d}$ is given by
  \begin{gather*}
    \begin{aligned}
      P_1 &= \setcom{\pi \in P}{\forall \pi_1 < \pi\exists \pi_2 \sim
        \pi_1:\len{\pi_2} < d}\\%
      \sim_1\ &=\ \sim \cap\ P_1 \times P_1
    \end{aligned}%
    \hspace{1cm}
    l_1(\pi) =
    \begin{cases}
      l(\pi) &\text{if } \exists \pi'\sim \pi: \len{\pi'} < d\\
      \bot &\text{otherwise}
    \end{cases}
  \end{gather*}
  The labelled quotient tree $(P^\iota_2,l^\iota_2,\sim^\iota_2)$ for
  each $\trunca{g_\iota}{d}$ is given by
  \begin{align*}
    \begin{aligned}
      P^\iota_2 &= \setcom{\pi \in \pos{g_\iota}}{\forall \pi_1 <
        \pi\exists \pi_2 \sim_{g_\iota} \pi_1:\len{\pi_2} < d}\\%
      \sim^\iota_2\ &=\ \sim \cap\ P^\iota_2 \times P^\iota_2%
    \end{aligned}%
    \hspace{5mm}
    l^\iota_2(\pi) &=
    \begin{cases}
      g_\iota(\pi) &\text{if } \exists \pi'\sim_{g_\iota} \pi: \len{\pi'} < d\\
      \bot &\text{otherwise}
    \end{cases}
  \end{align*}
  Due to \eqref{eq:cauchy}, all $(P^\iota_2,l^\iota_2,\sim^\iota_2)$
  with $\beta \le \iota < \alpha$ are pairwise equal. Therefore, we
  write $(P_2,l_2,\sim_2)$ for this common labelled quotient
  tree. That is, it remains to be shown that $(P_1,l_1,\sim_1)$ and
  $(P_2,l_2,\sim_2)$ are equal.

  (a) $P_1 = P_2$. For the ``$\subseteq$'' direction let $\pi \in
  P_1$. If $\pi = \emptyseq$, we immediately have that $\pi \in
  P_2$. Hence, we can assume that $\pi$ is non-empty.  Since $\pi \in
  P_1$ implies $\pi \in P$, there is some $\beta \le \beta' < \alpha$
  with $\pi \in \pos{g_\iota}$ for all $\beta' \le \iota <
  \alpha$. Moreover this means that for each $\pi_1 < \pi$ there is
  some $\pi_2 \sim \pi_1$ with $\len{\pi_2} < d$. That is, there is
  some $\beta' \le \gamma_{\pi_1} < \alpha$ such that $\pi_2
  \sim_{g_\iota} \pi_1$ for all $\gamma_{\pi_1} \le \iota <
  \alpha$. Since there are only finitely many proper prefixes $\pi_1 <
  \pi$ but at least one, we can define $\gamma =
  \max\setcom{\gamma_{\pi_1}}{\pi_1 < \pi}$ such that we have for each
  $\pi_1 < \pi$ some $\pi_2 \sim_{g_\gamma} \pi_1$ with $\len{\pi_2} <
  d$. Hence, $\pi \in P^\gamma_2 = P_2$.

  To show the converse direction, assume that $\pi \in P_2$. Then $\pi
  \in P^\iota_2 \subseteq \pos{g_\iota}$ for all $\beta \le \iota <
  \alpha$. Hence, $\pi \in P$. To show that $\pi \in P_1$, assume some
  $\pi_1 < \pi$. Since $\pi \in P^\beta_2$, there is some $\pi_2
  \sim_{g_\beta} \pi_1$ with $\len{\pi_2} < d$. Then $\pi_1 \in P_2$
  because $P_2$ is closed under prefixes and $\pi_2 \in P_2$ because
  $\len{\pi_2} < d$. Thus, $\pi_2 \sim_2 \pi_1$ which implies $\pi_2
  \sim_{g_\iota} \pi_1$ for all $\beta \le \iota<
  \alpha$. Consequently, $\pi_2 \sim \pi_1$, which means that $\pi \in
  P_1$.

  (c) $\sim_1\ =\ \sim_2$. For the ``$\subseteq$'' direction assume
  $\pi_1 \sim_1 \pi_2$. Hence, $\pi_1 \sim \pi_2$ and $\pi_1,\pi_2 \in
  P_1 = P_2$. This means that there is some $\beta \le \gamma <
  \alpha$ with $\pi_1 \sim_{g_\gamma} \pi_2$. Consequently, $\pi_1
  \sim_2 \pi_2$. For the converse direction assume that $\pi_1 \sim_2
  \pi_2$. Then $\pi_1,\pi_2 \in P_2 = P_1$ and $\pi_1 \sim_{g_\iota}
  \pi_2$ for all $\beta \le \iota<\alpha$. Hence, $\pi_1 \sim \pi_2$
  and we can conclude that $\pi_1 \sim_1 \pi_2$.

  (b) $l_1 = l_2$. We show this by proving that, for all $\beta \le
  \iota < \alpha$, the condition $\exists \pi'\sim \pi: \len{\pi'} <
  d$ from the definition of $l_1$ is equivalent to the condition
  $\exists \pi'\sim_{g_\iota} \pi: \len{\pi'} < d$ from the definition
  of $l_2$ and that $l(\pi) = g_\iota(\pi)$ if either condition is
  satisfied. The latter is simple: Whenever there is some $\pi' \sim
  \pi$ with $\len{\pi'} < d$, then $g_\iota(\pi) = l^\iota_2(\pi) =
  l^\beta_2(\pi) = g_\beta(\pi)$ for all $\beta \le \iota <
  \alpha$. Hence, $l(\pi) = g_\beta(\pi) = g_\iota(\pi)$ for all
  $\beta \le \iota < \alpha$. For the former, we first consider the
  ``only if'' direction of the equivalence. Let $\pi \in P_1$ and
  $\pi' \sim \pi$ with $\len{\pi'} < d$. Then also $\pi' \in P_1$
  which means that $\pi' \sim_1 \pi$. Since then $\pi' \sim_2 \pi$, we
  can conclude that $\pi' \sim_{g_\iota} \pi$ for all $\beta \le \iota
  < \alpha$. For the converse direction assume that $\pi \in P_2$,
  $\pi' \sim_{g_\iota} \pi$ and $\len{\pi'} < d$. Then also $\pi' \in
  P_2$ which means that $\pi' \sim_2 \pi$ and, therefore, $\pi' \sim
  \pi$.
\end{proof}

\begin{example}
  \label{ex:limit}
  Reconsider the sequence of term graphs $(h_\iota)_{\iota<\omega}$
  Figure~\ref{fig:gtransRed} on page~\pageref{fig:gtransRed}. As we have noticed in
  Example~\ref{ex:liminf}, the edge that loops back to the root node
  is pushed down as the sequence progresses. Thus, we have for each $n
  \in \nats$, that the strict truncations of the term graphs $h_\iota$
  with $n \le \iota < \omega$ at depth $n+1$ coincide. Therefore, by
  Lemma~\ref{lem:cauchyTruncFun}, $(h_\iota)_{\iota<\omega}$ is
  Cauchy. In particular, we have that $(h_\iota)_{\iota<\omega}$
  converges to $h_\omega$.

  The limit inferior induced by $\lebotg$ showed some curios behaviour
  as soon as acyclic sharing changes as we have seen in
  Example~\ref{ex:liminf} with the convergence illustrated in
  Figure~\ref{fig:convWeird}. This is not the case for the metric
  $\dda$. In fact, there is no topological space in which
  $(g_\iota)_{\iota<\omega}$ from Figure~\ref{fig:convWeird} converges
  to a unique limit.
\end{example}

\subsection{Other Truncation Functions and Their Metric Spaces}
\label{sec:other-trunc-funct}

Generalising concepts from terms to term graphs is not a
straightforward matter as we have to decide how to deal with
additional sharing that term graphs offer. The definition of strict
truncation seems to be an obvious choice for a generalisation of tree
truncation. In this section, we shall formally argue that it is in
fact the case. More specifically, we show that no matter how we define
the sharing of the $\bot$-nodes that fill the holes caused by the
truncation, we obtain the same topology. We will then contrast this to
the metric that we have used in previous work~\cite{bahr11rta} by
showing that small changes to its definition also change the induced
topology.

The following lemma is a handy tool for comparing metric spaces
induced by truncation functions:
\begin{lemma}
  \label{lem:truncCont}
  Let $\tau,\upsilon$ be two truncation functions on $\iptgraphs$ and $f\fcolon
  \ictgraphs \funto \ictgraphs$ a function on $\ictgraphs$. Then the
  following are equivalent
  \begin{enumerate}[(i)]
  \item $f$ is a continuous mapping $f\fcolon (\ictgraphs,\dd_\tau)\homto
    (\ictgraphs,\dd_\upsilon)$
  \item For each $g \in \ictgraphs$ and $d \in \nats$ there is some $e
    \in \nats$ such that 
    \[
    \abssim{\tau}(g,h) \ge e \quad \implies \quad \abssim{\upsilon}(f(g),f(h)) \ge d \quad
    \text{for all } h \in \ictgraphs
    \]
  \item For each $g \in \ictgraphs$ and $d \in \nats$ there is some $e
    \in \nats$ such that 
    \[
    \tau_e(g) \isom \tau_e(h) \quad \implies \quad \upsilon_d(f(g)) \isom \upsilon_d(f(h))
    \quad \text{for all } h \in \ictgraphs
    \]
  \end{enumerate}
\end{lemma}
\begin{proof}
  Analogous to Lemma~\ref{lem:cauchyTruncFun}.
\end{proof}

An easy consequence of the above lemma is that if two truncation
functions only differ by a constant depth, they induce the same
topology:
\begin{proposition}
  \label{prop:truncTopology}
  Let $\tau,\upsilon$ be two truncation functions on $\iptgraphs$ such that
  there is a $\delta \in \nats$ with $\abs{\abssim{\tau}(g,h) -
    \abssim{\upsilon}(g,h)} \le \delta$ for all $g,h\in \ictgraphs$. Then
  $(\ictgraphs,\dd_\tau)$ and $(\ictgraphs,\dd_\upsilon)$ are topologically
  equivalent, i.e.\ induce the same topology.
\end{proposition}
\begin{proof}
  We show that the identity function $\id\fcolon \ictgraphs \funto
  \ictgraphs$ is a homeomorphism from $(\ictgraphs,\dd_\tau)$ to
  $(\ictgraphs,\dd_\upsilon)$, i.e.\ both $\id$ and $\id^{-1}$ are
  continuous. Due to the symmetry of the setting it suffices to show
  that $\id$ is continuous. To this end, let $g \in \ictgraphs$ and
  $d\in \nats$. Define $e = d + \delta$ and assume some $h \in
  \ictgraphs$ such that $\abssim{\tau}(g,h) \ge e$. By
  Lemma~\ref{lem:truncCont}, it remains to be shown that then
  $\abssim{\upsilon}(g,h) \ge d$. Indeed, we have $\abssim{\upsilon}(g,h) \ge
  \abssim{\tau}(g,h) - \delta \ge e - \delta = d$.
\end{proof}
This shows that metric spaces induced by truncation functions are
essentially invariant under changes in the truncation function bounded
by a constant margin.

\begin{rem}
  We should point out that the original definition of the metric on
  terms by Arnold and Nivat~\cite{arnold80fi} was slightly different
  from the one we showed here. Recall that we defined similarity as
  the maximum depth of truncation that ensures equality:
  \[
  \abssim{\tau}(g,h) =  \max\setcom{d \in \nats\cup\set\infty}{\tau_d(g) \isom \tau_d(h)}
  \]
  Arnold and Nivat, on the other hand, defined it as the minimum
  truncation depth that still shows inequality:
  \[
  \abssimp{\tau}(g,h) =  \min\setcom{d \in \nats\cup\set\infty}{\tau_d(g) \not\isom \tau_d(h)}
  \]
  However, it is easy to see that either both $\abssim{\tau}(g,h)$ and
  $\abssimp{\tau}(g,h)$ are $\infty$ or $\abssimp{\tau}(g,h) =
  \abssim{\tau}(g,h) + 1$. Hence, by
  Proposition~\ref{prop:truncTopology}, both definitions yield the
  same topology.
\end{rem}

Proposition~\ref{prop:truncTopology} also shows that two truncation
functions induce the same topology if they only differ in way they
treat ``fringe nodes'', i.e.\ nodes that are introduced in place of
the nodes that have been cut off. Since the definition of truncation
functions requires that $\tau_0(g) \isom \bot$ and $\tau_\infty(g)
\isom g$, we usually do not give the explicit construction of the
truncation for the depths $0$ and $\infty$.
\begin{example}
  Consider the following variant $\tau$ of the strict truncation
  function $\trunca{}{}$. Given a term graph $g \in \iptgraphs$ and
  depth $d \in \natsp$ we define the truncation $\tau_d(g)$ as
  follows:
  \begin{align*}
    N^g_{<d} &= \setcom{n \in N^g}{\depth{g}{n} < d}\\%
    N^g_{=d} &= \setcom{n^i}{n\in N^g_{<d},0\le i <
      \rank{g}{n},\gsuc^g_i(n) \nin N^g_{<d}}\\%
    N^{\tau_d(g)} &= N^g_{<d}\uplus N^g_{=d}\\%
    \glab^{\tau_d(g)} (n) &=
    \begin{cases}
      \glab^g(n) &\text{if } n \in N^g_{<d}\\
      \bot &\text{if } n \in N^g_{=d}
    \end{cases} \hspace{40pt}%
    \gsuc^{\tau_d(g)}_i(n) =
    \begin{cases}
      \gsuc^g_i(n) &\text{if } n^i \nin N^g_{=d}\\
      n^i &\text{if } n^i \in N^g_{=d}
    \end{cases}
  \end{align*}
  One can easily show that $\tau$ is in fact a truncation function.
  The difference between $\trunca{}{}$ and $\tau$ is that in the
  latter we create a fresh node $n^i$ whenever a node $n$ has a
  successor $\gsuc^g_i(n)$ that lies at the fringe, i.e.\ at depth
  $d$. Since this only affects the nodes at the fringe and, therefore,
  only nodes at the same depth $d$ we get the following:
  \begin{align*}
    \trunca{g}{d} \isom \trunca{h}{d} \quad &\implies \quad \tau_d(g)
    \isom \tau_d(h),\text{ and}\\%
    \tau_d(g) \isom \tau_d(h) \quad &\implies \quad \trunca{g}{d-1} \isom
    \trunca{h}{d-1}.
  \end{align*}
  Hence, the respectively induced similarities only differ by a
  constant margin of $1$, i.e.\ we have that $\abs{\similara{g}{h} -
    \abssim{\tau}(g,h)} = 1$. According to
  Proposition~\ref{prop:truncaDown}, this means that
  $(\ictgraphs,\dda)$ and $(\ictgraphs, \dd_\tau)$ are topologically
  equivalent.

  Consider another variant $\upsilon$ of the strict truncation
  function $\trunca{}{}$.  Given a term graph $g \in \iptgraphs$ and
  depth $d \in \natsp$, we define the truncation $\upsilon_d(g)$ as
  follows:
  \begin{align*}
    N^g_{<d} &= \setcom{n \in N^g}{\depth{g}{n} < d}\\%
    N^g_{=d} &= \setcom{n^i}{
        \begin{aligned}
          n \in N^g,\depth{g}{n} = d - 1, 0\le i < \rank{g}{n} \text{ with} \quad
          &\gsuc^g_i(n)\nin \tNodes{g}{d}\\\text{or }\quad
          &n\nin \predAcy{g}{\gsuc^g_i(n)}
        \end{aligned}}\\%
    N^{\upsilon_d(g)} &= N^g_{<d}\uplus N^g_{=d}\\%
    \glab^{\upsilon_d(g)} (n) &=
    \begin{cases}
      \glab^g(n) &\text{if } n \in N^g_{<d}\\
      \bot &\text{if } n \in N^g_{=d}
    \end{cases} \hspace{40pt}%
    \gsuc^{\upsilon_d(g)}(n) =
    \begin{cases}
      \gsuc^g_i(n) &\text{if } n^i \nin N^g_{=d}\\
      n^i &\text{if } n^i \in N^g_{=d}
    \end{cases}
  \end{align*}
  Also $\upsilon$ forms a truncation function as one can easily show.
  In addition to creating fresh nodes $n^i$ for each successor that is
  not in the retained nodes $N^g_{<d}$, the truncation function
  $\upsilon$ creates such new nodes $n^i$ for each cycle that created
  by a node just above the fringe. Again, as for the truncation
  function $\tau$, only the nodes at the fringe, i.e.\ at depth $d$
  are affected by this change. Hence, the respectively induced
  similarities of $\trunca{}{}$ and $\upsilon$ only differ by a
  constant margin of $1$, which makes the metric spaces
  $(\ictgraphs,\dda)$ and $(\ictgraphs, \dd_\upsilon)$ topologically
  equivalent as well.
\end{example}
The robustness of the metric space $(\ictgraphs,\dda)$ under the
changes illustrated above is due to the uniformity of the core
definition of the strict truncation which only takes into account the
depth. By simply increasing the depth by a constant number, we can
compensate for changes in the way fringe nodes are dealt with.

This is much different for the truncation function $\trunc{g}{d}$ that
induces the metric space considered in \cite{bahr11rta}:

\begin{definition} [truncation of term graphs]
  \label{def:truncGraph}
  Let $g \in \iptgraphs$ and $d \in \nats$.
  \begin{enumerate}[(i)]
  \item Given $n,m\in N^g$, $m$ is an \emph{acyclic predecessor} of
    $n$ in $g$ if there is an acyclic occurrence $\pi \concat \seq i \in
    \nodePosAcy{g}{n}$ with $\pi \in \nodePos{g}{m}$. The set of
    acyclic predecessors of $n$ in $g$ is denoted $\predAcy{g}{n}$.
  \item The set of \emph{retained nodes} of $g$ at $d$, denoted
    $\tNodes{g}{d}$, is the least subset $M$ of $N^g$ satisfying the
    following conditions for all $n\in N^g$:
    \begin{center}
      \begin{inparaenum}[(T1)]
        \def\theenumi{T}
      \item $\depth{g}{n} < d  \implies  n \in
        M$  \label{eq:truncNodes1} \qquad
      \item $n \in M  \implies  \predAcy{g}{n} \subseteq M$
        \label{eq:truncNodes2}
      \end{inparaenum}
    \end{center}
  \item For each $n\in N^g$ and $i\in \nats$, we use $n^i$ to denote a
    fresh node, i.e.\ $\setcom{n^i}{n \in N^g, i\in \nats}$ is a set
    of pairwise distinct nodes not occurring in $N^g$.  The set of
    \emph{fringe nodes} of $g$ at $d$, denoted $\fNodes{g}{d}$, is
    defined as the singleton set $\set{r^g}$ if $d = 0$, and otherwise
    as the set
    \begin{gather*}
      \setcom{n^i}{
        \begin{aligned}
          n \in \tNodes{g}{d}, 0\le i < \rank{g}{n} \text{ with} \quad
          &\gsuc^g_i(n)\nin \tNodes{g}{d}\\\text{or }\quad
          &\depth{g}{n} \ge d - 1, n\nin \predAcy{g}{\gsuc^g_i(n)}
        \end{aligned}
}
    \end{gather*}
  \item The \emph{truncation} of $g$ at $d$, denoted $\trunc{g}{d}$,
    is the term graph defined by
    \begin{align*}
      N^{\trunc{g}{d}} &= \tNodes{g}{d} \uplus \fNodes{g}{d}
      & r^{\trunc{g}{d}} &= r^g
      \\
      \glab^{\trunc{g}{d}}(n) &= 
      \begin{cases}
        \glab^g(n) &\text{if } n \in \tNodes{g}{d} \\
        \bot &\text{if } n \in \fNodes{g}{d}
      \end{cases} &
      \gsuc_i^{\trunc{g}{d}}(n) &= 
      \begin{cases}
        \gsuc_i^g(n) &\text{if } n^i \nin \fNodes{g}{d} \\
        n^i &\text{if } n^i \in \fNodes{g}{d}
      \end{cases}
    \end{align*}
    Additionally, we define $\trunc{g}{\infty}$ to be the term graph
    $g$ itself.
  \end{enumerate}
\end{definition}
The idea of this definition of truncation is that not only each node
at depth $<d$ is kept \eqref{eq:truncNodes1} but also every acyclic
predecessor of such a node \eqref{eq:truncNodes2}. In sum, every node
on an acyclic path from the root to a node at depth smaller than $d$
is kept. The difference between the two truncation functions
$\trunca{}{}$ and $\trunc{}{}$ are illustrated in
Figure~\ref{fig:exTrunc}.
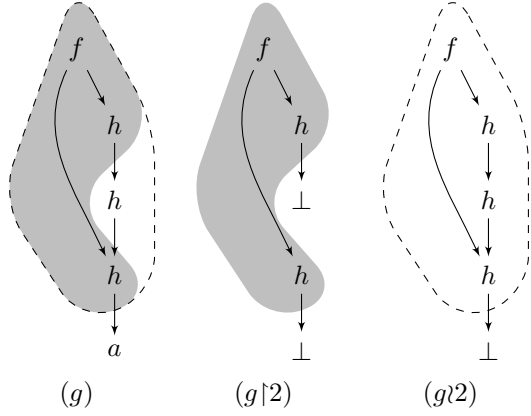
\begin{figure}
  \center
  \begin{tikzpicture}[->,node distance=2cm]
    \node[alias=r1] (f) {$f$}
    child [missing]
    child {
      node (g1) {$h$}
      child {
        node (g2) {$h$}
        child {
          node (g) {$h$}
          child {
            node {$a$}
          }
        }
      }
    };
    \draw (f) edge[out=-115,in=115] (g);
    \begin{pgfonlayer}{background}
      \fill[lightgray,rounded corners=.5cm] ($(f.north)+(0,.6)$) --
      ($(g1.east)+(.3,0)$) -- ($(g2.west)-(.3,0)$) -- ($(g.south
      east)+(.3,0)$) --($(g.south west)+(-.2,-.4)$) -- ($(f)+(-1,-2)$) --
      cycle;

      \draw[dashed ,rounded corners=.5cm] ($(f.north)+(0,.6)$) --
      ($(g1.east)+(.3,0)$)  -- ($(g.south
      east)+(.3,0)$) --($(g.south west)+(-.2,-.4)$) -- ($(f)+(-1,-2)$) --
      cycle;
    \end{pgfonlayer}

    \node[alias=r2,right=of f] (f) {$f$}
    child [missing]
    child {
      node (g1) {$h$}
      child {
        node (g2) {$\bot$}
        child[edge from parent/.style={}] {
          node (g) {$h$}
          child[edge from parent/.style={draw}] {
            node {$\bot$}
          }
        }
      }
    };
    \draw (f) edge[out=-115,in=115] (g);
    \begin{pgfonlayer}{background}
      \fill[lightgray,rounded corners=.5cm] ($(f.north)+(0,.6)$) --
      ($(g1.east)+(.3,0)$) -- ($(g2.west)-(.3,0)$) -- ($(g.south
      east)+(.3,0)$) --($(g.south west)+(-.2,-.4)$) -- ($(f)+(-1,-2)$) --
      cycle;
    \end{pgfonlayer}

    \node[alias=r3,right=of f] (f) {$f$}
    child [missing]
    child {
      node (g1) {$h$}
      child {
        node (g2) {$h$}
        child {
          node (g) {$h$}
          child {
            node {$\bot$}
          }
        }
      }
    };
    \draw (f) edge[out=-115,in=115] (g);
    \begin{pgfonlayer}{background}
      \draw[dashed ,rounded corners=.5cm] ($(f.north)+(0,.6)$) --
      ($(g1.east)+(.3,0)$)  -- ($(g.south
      east)+(.3,0)$) --($(g.south west)+(-.2,-.4)$) -- ($(f)+(-1,-2)$) --
      cycle;
    \end{pgfonlayer}
    \begin{scope}[node distance=4cm]
      \node[below=of r1] {($g$)};
      \node[below=of r2] {($\trunca{g}{2}$)};
      \node[below=of r3] {($\trunc{g}{2}$)};
    \end{scope}
  \end{tikzpicture}
  \caption{Comparison to strict truncation.}
  \label{fig:exTrunc}
\end{figure}

In contrast to $\trunca{}{}$, the truncation function $\trunc{}{}$ is
quite vulnerable to small changes:
\begin{example}
  Consider the following variant $\tau$ of the truncation function
  $\trunc{}{}$. Given a term graph $g \in \iptgraphs$ and depth $d \in
  \natsp$, we define the truncation $\tau_d(g)$ as follows: The set of
  retained nodes $N^g_{<d}$ is defined as for the truncation
  $\trunc{g}{d}$. For the rest we define
  \begin{align*}
    N^g_{=d} &= \setcom{\gsuc^g_i(n)}{n\in N^g_{<d},0\le i <
      \rank{g}{n},\gsuc^g_i(n) \nin N^g_{<d}}\\%
    N^{\tau_d(g)} &= N^g_{<d}\uplus N^g_{=d}\\%
    \glab^{\tau_d(g)} (n) &=
    \begin{cases}
      \glab^g(n) &\text{if } n \in N^g_{<d}\\
      \bot &\text{if } n \in N^g_{=d}
    \end{cases} \hspace{40pt}%
    \gsuc^{\tau_d(g)}(n) =
    \begin{cases}
      \gsuc^g(n) &\text{if } n \in N^g_{<d}\\
      \emptyseq &\text{if } n \in N^g_{=d}
    \end{cases}
  \end{align*}

  In this variant of truncation, some sharing of the retained nodes
  is preserved. Instead of creating fresh nodes for each successor
  node that is not in the set of retained nodes, we simply keep the
  successor node. Additionally loops back into the retained nodes
  are not cut off. This variant of the truncation deals with its
  retained nodes in essentially the same way as the strict
  truncation. However, opposed the strict truncation and their
  variants, this truncation function yields a topology different from
  the metric space $(\ictgraphs,\dd_{\trunc{}{}})$! To see this, consider the two
  families of term graphs $g_n$ and $h_n$ indicated in
  Figure~\ref{fig:truncFringe}. For both families we have that the
  $\tau$-truncations at depth $2$ to $n+2$ are the same, i.e.\ $\tau_d(g_n)
  = \tau_2(g_n)$ and $\tau_d(h_n) = \tau_2(h_n)$ for all $ 2 \le d \le
  n+2$. The same holds for the truncation function
  $\trunc{}{}$. Moreover, since the two leftmost successors of the
  $h$-node are not shared in $g_n$, both truncation functions coincide
  on $g_n$, i.e.\ $\trunc{g_n}{d} = \tau_d(g_n)$. This is not the case
  for $h_n$. In fact, they only coincide up to depth $1$. However, we
  have that $\trunc{h_n}{d} = \tau_d(g_n)$. In total, we can observe that
  $\similar{g_n}{h_n} = n + 2$ but $\abssim{\tau}(g_n,h_n) = 1$. This
  means, however, that the sequence $\seq{g_0,h_0,g_1,h_1,\dots}$
  converges in $(\ictgraphs,\dd_{\trunc{}{}})$ but not in $(\ictgraphs,\dd_\tau)$!
  
  A similar example can be constructed that uses the difference in the
  way the two truncation functions deal with fringe nodes created by
  cycles back into the set of retained nodes.
\end{example}
\begin{figure}
  \centering
  \begin{tikzpicture}[node distance=3cm, allow upside down]
    \node[alias=r1] (g) {$f$}
    child {
      node (f1) {$g$}
      child {
        node {$\vdots$}
        child {
          node (f2) {$g$}
          child {
            node (f3) {$h$}
            child {
              node (a) {$a$}
            }
          }
        }
      }
    };
    \draw[->] (g) edge[bend left=35] (a);
    \begin{scope}[node distance=10mm]
      \node [left of=f3] (a1) {$a$};
      \node [below left of=f3] (a2) {$a$};
      \draw[->] (f3) edge (a1) edge (a2);
    \end{scope}

    \draw[decorate,decoration=brace] (f2.south west) -- (f1.north
    west) node[midway,above,sloped] {$n$ times};

    \node[alias=r2, right of=r1] (g) {$f$}
    child {
      node (f1) {$g$}
      child {
        node {$\vdots$}
        child {
          node (f2) {$g$}
          child {
            node (f3) {$h$}
            child {
              node (a) {$a$}
            }
          }
        }
      }
    };
    \draw[->] (g) edge[bend left=35] (a);
    \begin{scope}[node distance=10mm]
      \node [left of=f3] (a1) {$\bot$};
      \node [below left of=f3] (a2) {$\bot$};
      \draw[->] (f3) edge (a1) edge (a2);
    \end{scope}

    \draw[decorate,decoration=brace] (f2.south west) -- (f1.north
    west) node[midway,above,sloped] {$n$ times};

    \node[alias=r3, right of=r2] (g) {$f$}
    child {
      node (f1) {$g$}
      child {
        node {$\vdots$}
        child {
          node (f2) {$g$}
          child {
            node (f3) {$h$}
            child {
              node (a) {$a$}
            }
          }
        }
      }
    };
    \draw[->] (g) edge[bend left=35] (a);
    \begin{scope}[node distance=10mm]
      \node [left of=f3] (a1) {$a$};
      \draw[->] (f3) edge[bend left] (a1) edge[bend right] (a1);
    \end{scope}

    \draw[decorate,decoration=brace] (f2.south west) -- (f1.north
    west) node[midway,above,sloped] {$n$ times};

    \node[alias=r4, right of=r3] (g) {$f$}
    child {
      node (f1) {$g$}
      child {
        node {$\vdots$}
        child {
          node (f2) {$g$}
          child {
            node (f3) {$h$}
            child {
              node (a) {$a$}
            }
          }
        }
      }
    };
    \draw[->] (g) edge[bend left=35] (a);
    \begin{scope}[node distance=10mm]
      \node [left of=f3] (a1) {$\bot$};
      \draw[->] (f3) edge[bend left] (a1) edge[bend right] (a1);
    \end{scope}

    \draw[decorate,decoration=brace] (f2.south west) -- (f1.north
    west) node[midway,above,sloped] {$n$ times};
    \begin{scope}[node distance=5.5cm]
      \node[below=of r1] {$(g_n)$};
      \node[below=of r2] {$(\tau_{2}(g_n) = \tau_{n+2}(g_n))$};
      \node[below=of r3] {$(h_n)$};
      \node[below=of r4] {$(\tau_{2}(h_n) = \tau_{n+2}(h_n))$};
    \end{scope}
  \end{tikzpicture}
  \caption{Variations in fringe nodes.}
\label{fig:truncFringe}
\end{figure}
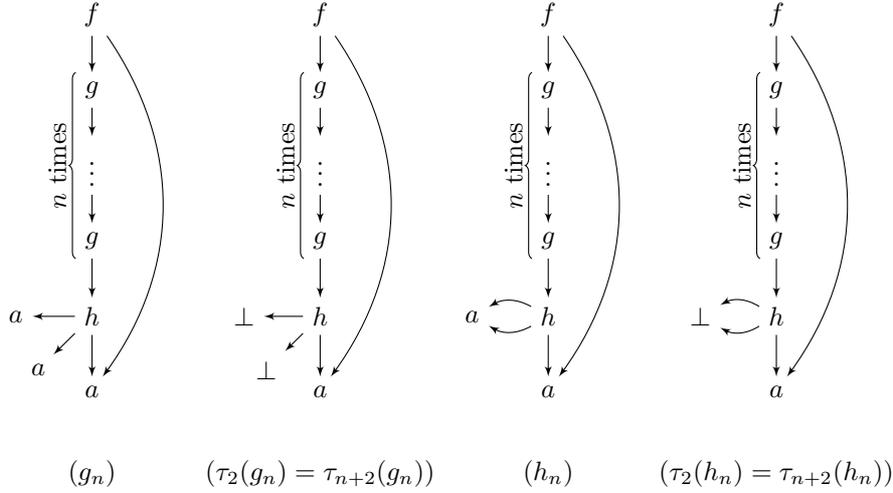

\section{Partial Order vs.\ Metric Space}
\label{sec:partial-order-vs}

Recall that $\prs$-convergence in term rewriting is a conservative
extension of $\mrs$-convergence (cf.\
Theorem~\ref{thr:strongExt}). The key property that makes this
possible is that for each sequence $(t_\iota)_{\iota<\alpha}$ in
$\iterms$, we have that $\lim_{\iota\limto\alpha} t_\iota =
\liminf_{\iota\limto\alpha} t_\iota$ whenever
$(t_\iota)_{\iota<\alpha}$ converges, or $\liminf_{\iota\limto\alpha}
t_\iota \in \iterms$.

Unfortunately, this is not the case for the metric space and the
partial order that we consider on term graphs. As we have shown in
Example~\ref{ex:limit}, the sequence of term graphs depicted in
Figure~\ref{fig:convWeird} has a total term graph as its limit
inferior although it does not converge in the metric space. This
example shows that we cannot hope to generalise the compatibility
property that we have for terms: Even if a sequence of total term
graphs has a total term graph as its limit inferior, it might not
converge. However, the other direction of the compatibility does hold
true:
\begin{theorem}
  \label{thr:limLiminf}
  If $(g_\iota)_{\iota<\alpha}$ converges, then
  $\lim_{\iota\limto\alpha} g_\iota = \liminf_{\iota\limto\alpha}
  g_\iota$.
\end{theorem}
\begin{proof}
  In order to prove this property, we will use the construction of the
  limit resp.\ the limit inferior of a sequence of term graphs which
  we have shown in Theorem~\ref{thr:smetricComplete} resp.\
  Corollary~\ref{cor:lebot1Liminf}.

  According to Theorem~\ref{thr:smetricComplete}, we have that the
  canonical term graph $\lim_{\iota\limto\alpha}g_\iota$ is given by
  the following labelled quotient tree $(P,\sim,l)$:
  \begin{gather*}
    P = \bigcup_{\beta<\alpha}\bigcap_{\beta\le\iota<\alpha}
    \pos{g_\iota}%
    \qquad
    {\sim} = \bigcup_{\beta<\alpha}\bigcap_{\beta\le\iota<\alpha}
    {\sim_{g_\iota}}%
    \\%
    l(\pi) = f \quad \text{iff}\quad \exists \beta<\alpha\forall
    \beta\le\iota<\alpha: g_\iota(\pi) = f%
  \end{gather*}
  We will show that $g = \liminf_{\iota\limto\alpha}g_\iota$ induces
  the same labelled quotient tree.
  
  From Corollary~\ref{cor:lebot1Liminf}, we immediately obtain that
  $\pos{g} \subseteq P$. To show the converse direction
  $\pos{g}\supseteq P$, we assume some $\pi\in P$. According to
  Corollary~\ref{cor:lebot1Liminf}, in order to show that $\pi \in
  \pos{g}$, we have to find a $\beta < \alpha$ such that
  $\pi\in\pos{g_\beta}$ and for each $\pi'<\pi$ there is some $f
  \in\Sigma_\bot$ such that $g_\iota(\pi') = f$ for all
  $\beta\le\iota<\alpha$.

  Because $\pi \in P$, there is some $\beta_1 < \alpha$ such that $\pi
  \in \pos{g_\iota}$ for all $\beta_1 \le \iota < \alpha$. Since
  $(g_\iota)_{\iota<\alpha}$ converges, it is also Cauchy. Hence, by
  Lemma~\ref{lem:cauchyTruncFun}, for each $d \in \nats$, there is
  some $\beta_2<\alpha$ such that $\trunca{g_\gamma}{d} \isom
  \trunca{g_\iota}{d}$ for all $\beta_2\le \gamma,\iota <
  \alpha$. Specialising this to $d = \len{\pi}$, we obtain some
  $\beta_2<\alpha$ with $\trunca{g_\gamma}{\len{\pi}} \isom
  \trunca{g_\iota}{\len{\pi}}$ for all $\beta_2\le \gamma,\iota <
  \alpha$. Let $\beta = \max\set{\beta_1,\beta_2}$. Then we have $\pi
  \in \pos{g_\iota}$ and $\trunca{g_\beta}{\len{\pi}} \isom
  \trunca{g_\iota}{\len{\pi}}$ for each $\beta \le \iota <
  \alpha$. Hence, for each $\pi' < \pi$, the symbol $f =
  g_\beta(\pi')$ is well-defined, and, according to
  Corollary~\ref{cor:truncaOccRep}, we have that $g_\iota(\pi') = f$
  for each $\beta \le \iota < \alpha$.

  The equalities ${\sim} = {\sim_g}$ and $l = g(\cdot)$ follow from
  Corollary~\ref{cor:lebot1Liminf} as $P = \pos{g}$.
\end{proof}

\section{Infinitary Term Graph Rewriting}
\label{sec:infin-term-graph}

In the previous sections, we have constructed and investigated the
necessary metric and partial order structures upon which the
infinitary calculus of term graph rewriting that we shall introduce in
this section is based. After describing the framework of term graph
rewriting that we consider, we will explore different modes
convergence on term graphs. In the same way that infinitary term
rewriting is based on the abstract notions of $\mrs$- and
$\prs$-convergence~\cite{bahr10rta}, infinitary term graph rewriting
is an instantiation of these abstract modes of convergence to term
graphs.

\subsection{Term Graph Rewriting Systems}
\label{sec:graph-rewr-syst}

The framework of term graph rewriting that we consider is that of
Barendregt et al.~\cite{barendregt87parle}. Similarly to term
rewriting systems, we have to deal with variables. That is, we
consider a signature $\Sigma_\calV$ extended with a set of variable
symbols $\calV$.
\begin{definition}[term graph rewriting system]
  \quad
  \begin{enumerate}[(i)]
  \item Given a signature $\Sigma$, a \emph{term graph rule} $\rho$
    over $\Sigma$ is a triple $(g,l,r)$ where $g$ is a graph over
    $\Sigma_\calV$ and $l,r \in N^g$, such that all nodes in $g$
    reachable from $l$ or $r$. We write $\lhs\rho$ resp.\ $\rhs\rho$
    to denote the left- resp.\ right-hand side of $\rho$, i.e.\ the
    term graph $\subgraph{g}{l}$ resp.\
    $\subgraph{g}{r}$. Additionally, we require that $\lhs\rho$ is
    finite and that for each variable $v\in\calV$ there is at most one
    node $n$ in $g$ labelled $v$ and $n$ is different but still
    reachable from $l$.
  \item A \emph{term graph rewriting system (GRS)} $\calR$ is a pair
    $(\Sigma,R)$ with $\Sigma$ a signature and $R$ a set of term graph
    rules.
  \end{enumerate}
\end{definition}

The requirement that the root $l$ of the left-hand side is not
labelled with a variable symbol is analogous to the requirement that
the left-hand side of a term rule is not a variable. Similarly, the
restriction that nodes labelled with variable symbols must be
reachable from the root of the left-hand side corresponds to the
restriction on term rules that every variable occurring on the
right-hand side must also occur on the left-hand side.

Term graphs can be used to compactly represent term. This
representation of terms is defined by the unravelling. This notion can
be extended to term graph rules.  Figure~\ref{fig:grules} illustrates
two term graph rules that both represent the term rule $a : x \to b :
a : x$ from Example~\ref{ex:termRewr} to which they unravel.
\begin{definition}[unravelling of term graph rules]
  Let $\rho$ be a term graph rule with $\rho_l$ and $\rho_r$ left-
  resp.\ right-hand side term graph. The \emph{unravelling} of $\rho$,
  denoted $\unrav{\rho}$ is the term rule $\unrav{\rho_l} \to
  \unrav{\rho_r}$. Let $\calR = (\Sigma,R)$ be a GRS. The unravelling
  of $\calR$, denoted $\unrav\calR$ is the TRS $(\Sigma,\unrav R)$
  with $\unrav R = \setcom{\unrav{\rho}}{\rho\in G}$.
\end{definition}
We will investigate the aspect of how term graph rules simulate their
unravellings in Section~\ref{sec:terms-vs.-term}.

The application of a rewrite rule $\rho$ (with root nodes $l$ and $r$)
to a term graph $g$ is performed in four steps: At first a suitable
sub-term graph of $g$ rooted in some node $n$ of $g$ is \emph{matched}
against the left-hand side of $\rho$. This amounts to finding a
$\calV$-homomorphism $\phi$ from the term graph rooted in $l$ to the
sub-term graph rooted in $n$, the \emph{redex}. The
$\calV$-homomorphism $\phi$ allows to instantiate variables in the
rule with sub-term graphs of the redex. In the second step, nodes and
edges in $\rho$ that are not reachable from $l$ are copied into $g$,
such that edges pointing to nodes in the term graph rooted in $l$ are
redirected to the image under $\phi$. In the last two steps, all edges
pointing to $n$ are redirected to (the copy of) $r$ and all nodes not
reachable from the root of (the now modified version of) $g$ are
removed.

\begin{definition}[application of a term graph rewrite
  rule, \cite{barendregt87parle}]
  \label{def:termGraphApp}
  Let $\rho = (N^\rho,\glab^\rho,\gsuc^\rho,l^\rho,r^\rho)$ be a term
  graph rewrite rule over signature $\Sigma$, $g \in \itgraphs$ and $n
  \in N^g$. $\rho$ is called \emph{applicable} to $g$ at $n$ if there
  is a $\calV$-homomorphism $\phi\fcolon\rho_l \homto_\calV
  \subgraph{g}{n}$. $\phi$ is called the \emph{matching
    $\calV$-homomorphism} of the rule application, and
  $\subgraph{g}{n}$ is called a \emph{$\rho$-redex}. Next, we define
  the \emph{result} of the application of the rule $\rho$ to $g$ at
  $n$ using the $\calV$-homomorphism $\phi$. This is done by
  constructing the intermediate graphs $g_1$ and $g_2$, and the final
  result $g_3$.
  \begin{enumerate}[(i)]
  \item The graph $g_1$ is obtained from $g$ by adding the part of
    $\rho$ not contained in the left-hand side:
    \begin{align*}
      N^{g_1} &= N^g \uplus (N^\rho \setminus N^{\rho_l})\\
      \glab^{g_1}(m) &= 
      \begin{cases}
        \glab^g(m) & \text{if } m \in N^g\\
        \glab^\rho(m) & \text{if } m \in N^\rho \setminus N^{\rho_l}
      \end{cases}\\
      \gsuc^{g_1}_i(m) &= 
      \begin{cases}
        \gsuc^g_i(m) & \text{if } m \in N^g\\
        \gsuc^\rho_i(m) & \text{if } m, \gsuc^\rho_i(m) \in N^\rho \setminus N^{\rho_l}\\
        \phi(\gsuc^\rho_i(m)) & \text{if } m \in N^\rho \setminus
        N^{\rho_l}, \gsuc^\rho_i(m) \in N^{\rho_l}
      \end{cases}
    \end{align*}
  \item Let $n' = \phi(r^\rho)$ if $r^\rho \in N^{\rho_l}$ and $n'=
    r^\rho$ otherwise. The graph $g_2$ is obtained from $g_1$ by
    redirecting edges ending in $n$ to $n'$:
    \begin{align*}
      N^{g_2} = N^{g_1} \qquad
      \glab^{g_2} = \glab^{g_1}\qquad
      \gsuc^{g_2}_i(m) = 
      \begin{cases}
        \gsuc_i^{g_1}(m) &\text{if } \gsuc^{g_1}_i(m) \neq n\\
        n' &\text{if } \gsuc^{g_1}_i(m) = n
      \end{cases}
    \end{align*}
  \item The term graph $g_3$ is obtained by setting the root node
    $r'$, which is $r$ if $l = r^g$, and otherwise $r^g$. That is,
    $g_3 = \subgraph{g_2}{r'}$. This also means that all nodes not
    reachable from $r'$ are removed.
  \end{enumerate}
  This induces a reduction step $\psi\fcolon g \to g_3$. In order to indicate the
  applied rule $\rho$ and the root nodes $n,n'$ of the redex resp.\ the
  reduct, we write $\psi\fcolon g \to[n,\rho,n'] g_3$.
\end{definition}

Examples for term graph rewriting steps are shown in
Figure~\ref{fig:gredEx}. We revisit them in more detail in
Example~\ref{ex:gRed} in the next section.

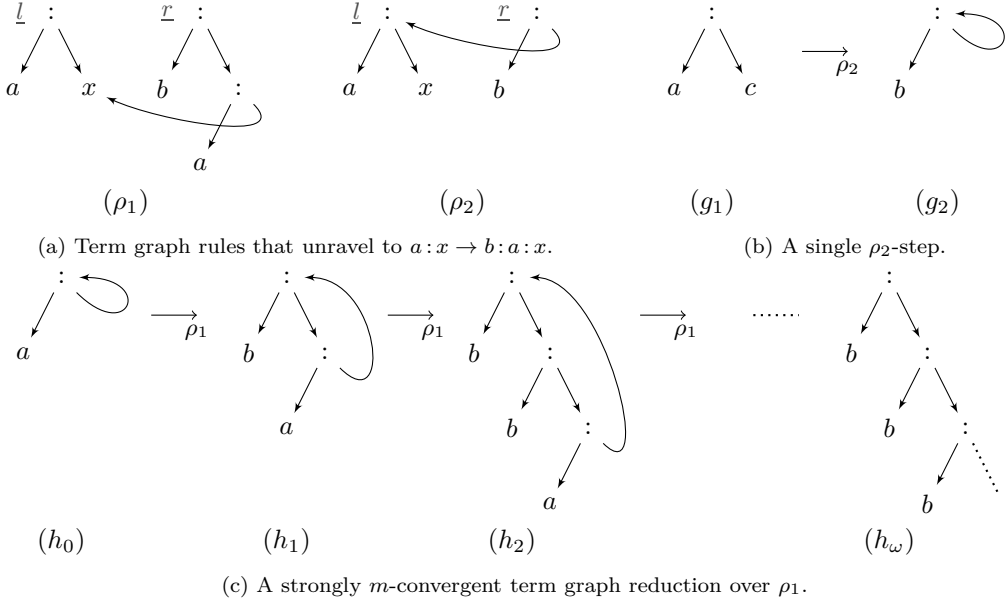
\begin{figure}
  \centering \subfloat[Term graph rules that unravel to $a \cons x
  \rightarrow b \cons a \cons x$.]{
      \begin{tikzpicture}
      \node [node name=180:\underline{l}] (l) {$\cons$}%
      child {%
        node (a) {$a$}%
      } child {%
        node (x) {$x$}%
      };%

      \node[node name=180:\underline{r},node distance=1.5cm,right=of l] (r) {$\cons$}%
      child {%
        node (b) {$b$}%
      } child {%
        node (c) {$\cons$}%
        child {%
          node (a') {$a$}%
        } child [missing]
      };%
      \draw%
      (c) edge[->,out=-45,in=-25] (x);%
      
      \node () at ($(l)!.5!(r) + (0,-2.5)$) {$(\rho_1)$};
      \node [node name=180:\underline{l},node distance=2cm,right=of r] (l) {$\cons$}%
      child {%
        node (a) {$a$}%
      } child {%
        node (x) {$x$}%
      };%

      \node[node name=180:\underline{r},node distance=1.5cm,right=of l] (r) {$\cons$}%
      child {%
        node (b) {$b$}%
      } child [missing];%
      \draw%
      (r) edge[->,out=-45,in=-25] (l);%
      
      \node () at ($(l)!.5!(r) + (0,-2.5)$) {$(\rho_2)$};
    \end{tikzpicture}
    \label{fig:grules}
  }
  \hspace{5mm}
  \subfloat[A single $\rho_2$-step.]{
      \begin{tikzpicture}
        
      \node (g1) {$\cons$}%
      child {%
        node {$a$}%
      } child {%
        node {$c$}
        };%

      \node () at ($(g1) + (0,-2.5)$) {$(g_1)$};

      \node [node distance=2.5cm,right=of g1] (g2) {$\cons$}%
      child {%
        node {$b$}%
      } child [missing];%
      \draw (g2) edge[->, min distance=10mm,out=-45,in=0] (g2);%

      \node () at ($(g2) + (0,-2.5)$) {$(g_2)$};
      \node (s1) at ($(g1)!.5!(g2)-(0,.5)$) {};

      \draw[single step] ($(s1)-(.3,0)$) -- ($(s1)+(.3,0)$)
      node[pos=1,below] {{\small$\rho_2$}};
    \end{tikzpicture}
    \label{fig:gsingleStep}
  }
  \hspace{1cm}%
  \subfloat[A strongly $\mrs$-convergent term graph reduction over $\rho_1$.]{%
      \begin{tikzpicture}
        
      \node (g1) {$\cons$}%
      child {%
        node (a) {$a$}%
      } child [missing];%
      \draw (g1) edge[->, min distance=10mm,out=-45,in=0] (g1);%

      \node () at ($(g1) + (0,-3.5)$) {$(h_0)$};
        
      \node [node distance=2.5cm,right=of g1] (g2) {$\cons$}%
      child {%
        node (b) {$b$}%
      } child {%
        node (c) {$\cons$}%
        child {%
          node (a) {$a$}%
        } child [missing]%
      };%
      \draw (c) edge[->, min distance=10mm,out=-45,in=0] (g2);%

      \node () at ($(g2) + (0,-3.5)$) {$(h_1)$};
        
      \node [node distance=2.5cm,right=of g2] (g3) {$\cons$}%
      child {%
        node (b) {$b$}%
      } child {%
        node (c) {$\cons$}%
        child {%
          node (b2) {$b$}%
        } child {%
          node (c2) {$\cons$}%
          child {%
            node (a) {$a$}%
          } child [missing]%
        }%
      };%
      \draw (c2) edge[->, min distance=10mm,out=-45,in=0] (g3);%

      \node () at ($(g3) + (0,-3.5)$) {$(h_2)$};
        
      \node [node distance=4.5cm,right=of g3] (go) {$\cons$}%
      child {%
        node (b) {$b$}%
      } child {%
        node (c) {$\cons$}%
        child {%
          node (b2) {$b$}%
        } child {%
          node (c2) {$\cons$}%
          child {%
            node (b3) {$b$}%
          } child[etc] {%
            node {} %
          }%
        }
      };%

      \node () at ($(go) + (0,-3.5)$) {$(h_\omega)$};

    \node (s1) at ($(g1)!.5!(g2)-(0,.5)$) {};
    \node (s2) at ($(g2)!.55!(g3)-(0,.5)$) {};
    \node (s3) at ($(g3)!.4!(go)-(0,.5)$) {};
    \node (s4) at ($(g3)!.7!(go)-(0,.5)$) {};

    \draw[single step] ($(s1)-(.3,0)$) -- ($(s1)+(.3,0)$)
    node[pos=1,below] {{\small$\rho_1$}};
    \draw[single step] ($(s2)-(.3,0)$) -- ($(s2)+(.3,0)$)
    node[pos=1,below] {{\small$\rho_1$}};
    \draw[single step] ($(s3)-(.3,0)$) -- ($(s3)+(.3,0)$)
    node[pos=1,below] {{\small$\rho_1$}};
    \draw[dotted,thick,-] ($(s4)-(.3,0)$) -- ($(s4)+(.3,0)$);
        
    \end{tikzpicture}
    \label{fig:gtransRed}
  }  
  \caption{Term graph rules and their reductions.}
  \label{fig:gredEx}
\end{figure}

Note that term graph rules do not provide a duplication
mechanism. Each variable is allowed to occur at most once. Duplication
must always be simulated by sharing. This means for example that
variables that should ``occur'' on the right-hand side must share the
occurrence of that variable on the left-hand side of the rule as seen
in the term graph rules in Figure~\ref{fig:grules}. This sharing can
be direct as in $\rho_1$ or indirect as in $\rho_2$. For variables
that are supposed to be duplicated on the right-hand side, for example
in the term rewrite rule $Y\,x \to x\,(Y\,x)$ that defines the fixed
point combinator $Y$ in an applicative language, we have to use
sharing in order to represent multiple occurrence of the same variable
as seen in the corresponding term graph rules in
Figure~\ref{fig:fixedPointCombA}.

As for term graphs, we also give a linear notation for term graph
rules:
\begin{definition}[linear notation of term graph rules]
  Let $\Sigma$ be a signature and $\oh\Sigma$ its extension as in
  Definition~\ref{def:linGraphNot}. A linear notation for a term graph
  rule over $\Sigma$ is a term rule $\rho\colon s \to t$ over
  $\oh\Sigma$ such that for each $n\in \calN$ that occurs in $\rho$
  there is exactly one occurrence of a function symbol of the form
  $\nn n
  f$ in $\rho$.

  The corresponding term graph rule $\rho'$ is defined as follows:
  Consider the term tree representations of $s$ and $t$. Let $l$ and
  $r$ be the root nodes of $s$ resp.\ $t$, and let $g$ be the disjoint
  union of $s$ and $t$. In $g$, redirect every edge to a node labelled
  $n$ to the unique node labelled $\nn n f$ for some
  $f\in\Sigma$. Then change all labellings of the form $\nn n f$ to
  $f$. In the resulting graph $g'$ do the following: For each $x\in
  \calV$ occurring in $g'$, redirect all edges to nodes labelled $x$
  to a single fresh node labelled $x$ provided $x \neq t$. If $x = t$,
  then redirect all edges to nodes labelled $x$ to the node $r$. Let
  $g''$ be the thus obtained graph after removing all nodes not
  reachable from $l$ or $r$. Then $\rho'$ is the term graph rule
  $(g'',l,r)$.
\end{definition}

As an example, the term graph rules in Figure~\ref{fig:grules} can be
written as $\rho_1\fcolon a\cons x \to b\cons a\cons  x$ resp.\
$\rho_2\fcolon a \nn n {:} x \to b\cons n$. Also note that each term
rule $\rho\fcolon l \to r$ can be interpreted as a linear notation for
a term graph rule $\rho'\fcolon l \to r$. This term graph rule $\rho'$ is, in
fact, the term graph rule with minimal sharing that unravels to $\rho$.

\subsection{Weak Convergence}
\label{sec:weak-convergence}

We start by first considering weak notions of convergences based on
the metric $\dda$ and the partial order $\lebotg$:
\begin{definition}
  Let $\calR = (\Sigma,R)$ be a GRS. Term graphs in $\calR$ range over
  $\ictgraphs$. Thus, we consider the applications of term graph to
  yield canonical term graphs. That is, when we write $g
  \to[n,\rho,n'] h$, we mean that there is a reduction step $g
  \to[n,\rho,\oh n] h'$ according to
  Definition~\ref{def:termGraphApp}, such that $h = \canon{h'}$ and
  $n' = \phi(\oh n)$ for the isomorphism $\phi\fcolon h' \to h$.
  
  \begin{enumerate}[(i)]
  \item A \emph{(transfinite) reduction} in $\calR$ is a sequence
    $(g_\iota \to[n_\iota,\rho_\iota] g_{\iota+1})_{i < \alpha}$ of
    rewriting steps in $\calR$. If $S$ is finite, we write $S\fcolon
    g_0 \fto{*} g_\alpha$.
  \item Let $S = (g_\iota \to_\calR g_{\iota+1})_{\iota < \alpha}$ be
    a reduction in $\calR$. $S$ is \emph{weakly $\mrs$-continuous},
    written $S\fcolon g_0 \wmacont[\calR]$, if the underlying sequence
    of term graphs $(g_\iota)_{\iota < \wsuc\alpha}$ is continuous,
    i.e.\ $\lim_{\iota\limto\lambda} g_\iota = g_\lambda$ for each
    limit ordinal $\lambda < \alpha$. $S$ \emph{weakly
      $\mrs$-converges} to $g \in \ictgraphs$ in $\calR$, written
    $S\fcolon g_0 \wmato[\calR] g$, if it is weakly $\mrs$-continuous
    and $\lim_{\iota\limto\wsuc\alpha} g_\iota = g$.
  \item Let $\calR_\bot$ be the GRS $(\Sigma_\bot, R)$ over the
    extended signature $\Sigma_\bot$ and $S = (g_\iota \to[\calR_\bot]
    g_{\iota+1})_{\iota < \alpha}$ a reduction in $\calR_\bot$. $S$ is
    weakly $\prs$-continuous, written $S\fcolon g_0 \wpato[\calR] g$,
    if $\liminf_{\iota<\lambda} g_i = g_\lambda$ for each limit
    ordinal $\lambda < \alpha$. $S$ \emph{weakly $\prs$-converges} to
    $g\in\ipctgraphs$ in $\calR$, written $S\fcolon g_0 \wpato[\calR]
    g$, if it is weakly $\prs$-continuous and
    $\liminf_{\iota<\wsuc\alpha} g_i = g$.
  \end{enumerate}

\end{definition}

Note that we have to extend the signature of $\calR$ to $\Sigma_\bot$
for the definition of weak $\prs$-convergence. Moreover, since the
partial order $\lebotg$ forms a complete semilattice on $\ipctgraphs$,
weak $\prs$-continuity coincides with weak $\prs$-convergence:
\begin{proposition}
  In a GRS, every weakly $\prs$-continuous reduction is weakly
  $\prs$-convergent.
\end{proposition}
\begin{proof}
  Follows immediately from Corollary~\ref{cor:lebot1Liminf}.
\end{proof}

\begin{example}
\label{ex:gRed}  
Consider the term graph rule $\rho_1$ in Figure~\ref{fig:grules} that
unravels to the term rule $a : x \to b : a : x$ from
Example~\ref{ex:termRewr}. Starting with the term tree $a : c$,
depicted as $g_1$ in Figure~\ref{fig:gsingleStep}, we obtain the same
transfinite reduction as in Example~\ref{ex:termRewr}:
\[
S\fcolon a \cons c \to[\rho_1] b \cons a \cons c \to[\rho_1] b \cons b \cons a \cons c
\to[\rho_1] \;\dots\; h_\omega
\]
Also in this setting, $S$ both weakly $\mrs$- and $\prs$-converges to
the term tree $h_\omega$ shown in
Figure~\ref{fig:gtransRed}. Similarly, we can reproduce the weakly
$\prs$-converging but not weakly $\mrs$-converging reduction $T$ from
Example~\ref{ex:termRewr2}. Notice that $h_\omega$ is a rational term
tree as it can be obtained by unravelling the finite term graph $g_2$
depicted in Figure~\ref{fig:gsingleStep}. In fact, if we use the rule
$\rho_2$, we can immediately rewrite $g_1$ to $g_2$, which unravels to
$h_\omega$. In $\rho_2$, not only the variable $x$ is shared but the
whole left-hand side of the rule. This causes each redex of $\rho_2$
to be \emph{captured} by the right-hand side.

Figure~\ref{fig:gtransRed} indicates a transfinite reduction starting
with a cyclic term graph $h_0$ that unravels to the rational term $t =
a \cons a \cons a\cons \dots$. This reduction both weakly $\mrs$- and
$\prs$-converges to the rational term tree $h_\omega$ as well. Again,
by using $\rho_2$ instead of $\rho_1$, we can rewrite $h_0$ to the
cyclic term graph $g_2$ in one step.
\end{example}

As for TRSs, we have that weak $\mrs$-convergence implies weak
$\prs$-convergence.
\begin{theorem}
  Let $S$ be a reduction in a GRS $\calR$.
  \[
  \text{If}\quad S\fcolon g \wmato[\calR] h \qquad \text{then} \qquad
  S\fcolon g \wpato[\calR] h.
  \]
\end{theorem}
\begin{proof}
  Follows straightforwardly from Theorem~\ref{thr:limLiminf}.
\end{proof}

However, as we have indicated, weak $\mrs$-convergence is not the
total fragment of weak $\prs$-convergence as it is the case for
TRS. The GRS with the two rules $f(c,c) \to f(\nn n c,n)$ and $f(c,c)
\to f(c,c)$ yields the reduction sequence shown in
Figure~\ref{fig:convWeird}. This reduction weakly $\prs$-converges to
$f(c,c)$ but is not weakly $\mrs$-convergent. However, this peculiar
behaviour can be ruled out by considering strong convergence, which is
the subject of the following sections.

\subsection{Reduction Contexts}
\label{sec:reduction-contexts}

The idea of strong convergence is to conservatively approximate the
convergence behaviour somewhat independent from the actual rules that
are applied. Strong $\mrs$-convergence in TRSs requires that the depth
of the redexes tends to infinity thereby assuming everything at the
depth of the redex or below can potentially be affected by a reduction
step. Strong $\prs$-convergence, on the other hand, uses a better
approximation that only assumes that the whole redex can be changed by
a reduction not however its siblings. To this end strong
$\prs$-convergence uses a notion of reduction contexts -- essentially
the term minus the redex -- for the formation of limits. In order to
define a suitable notion of strong $\prs$-convergence on term graphs,
we have to devise a corresponding notion of reduction contexts. In
this section we shall devise such a notion and argue for its adequacy.

The following definition provides the basic construction that we use
to remove nodes from a term graph:
\begin{definition}[local truncation]
  Let $g \in \iptgraphs$ and $N \subseteq N^g$. The \emph{local
    truncation} of $g$ by $N$, denoted $\truncl{g}{N}$, is given as
  follows:
  \begin{align*}
    &N^{\truncl{g}{N}} \text{ is the least set $M$ satisfying}\quad
    \begin{aligned}
      (a) \;&r^g \in M, and\\
      (b) \;&n \in M\setminus N \implies \gsuc^g(n) \subseteq M.
    \end{aligned}\\%
    &r^{\truncl{g}{N}} = r^g\qquad%
    \glab^{\truncl{g}{N}} =
    \begin{cases}
      \glab^g(n) &\text{if }n\nin N\\
      \bot &\text{if }n\nin N
    \end{cases}\qquad%
    \gsuc^{\truncl{g}{N}}(n) =
    \begin{cases}
      \gsuc^g(n) &\text{if } n\nin N\\
      \emptyseq &\text{if } n\in N
    \end{cases}
  \end{align*}
  By abuse of notation, we write $\truncl{g}{n}$ instead of
  $\truncl{g}{\set{n}}$.
\end{definition}
The goal for the rest of this section is to establish that
$\truncl{g}{n}$ is an adequate notion of reduction context for a
reduction step $g \to[n] h$ applied at node $n$ in $g$. According to
he abstract notion of strong $\prs$-convergence \cite{bahr10rta}, this
requires that $\truncl{g}{n}\lebotg g,h$.

The following lemma shows that local truncations only remove positions
from a term graph but do not alter them:
\begin{lemma}
\label{lem:locTruncNodes}
  Let $g \in \iptgraphs$, $N \subseteq N^g$ and $\pi \in
  \pos{\truncl{g}{N}}$. Then $\nodeAtPos{g}{\pi} =
  \nodeAtPos{\truncl{g}{N}}{\pi}$.
\end{lemma}
\begin{proof}
  We proceed by induction on the length of $\pi$. The case $\pi =
  \emptyseq$ follows from the definition $r^{\truncl{g}{N}} = r^g$. If
  $\pi = \pi'\concat\seq{i}$, we can use the induction hypothesis to
  obtain that $\nodeAtPos{g}{\pi'} =
  \nodeAtPos{\truncl{g}{N}}{\pi'}$. As $\pi' \concat\seq{i} \in
  \pos{\truncl{g}{N}}$, we know that $\nodeAtPos{\truncl{g}{N}}{\pi'}
  \nin N$. We can thus reason as follows:
  \[
  \nodeAtPos{g}{\pi} = \gsuc^g_i(\nodeAtPos{g}{\pi'}) =
  \gsuc^g_i(\nodeAtPos{\truncl{g}{N}}{\pi'}) =
  \gsuc^{\truncl{g}{N}}_i(\nodeAtPos{\truncl{g}{N}}{\pi'}) =
  \nodeAtPos{\truncl{g}{N}}{\pi}
  \]
\end{proof}

This leads immediately to the observation that local truncations
preserve sharing:
\begin{lemma}[local truncations preserve sharing]
  \label{lem:locTruncSim}%
  Let $g \in \iptgraphs$, $N \subseteq N^g$ and $\pi_1,\pi_2 \in
  \pos{\truncl{g}{N}}$. Then $\pi_1 \sim_g \pi_2$ iff $\pi_1 \sim_{\truncl{g}{N}} \pi_2$.
\end{lemma}
\begin{proof}
  \quad  \vspace{-23pt}
  \begin{align*}
    \pi_1 \sim_g \pi_2%
    &\iff \nodeAtPos{g}{\pi_1} = \nodeAtPos{g}{\pi_2}\\%
    &\iff\nodeAtPos{\truncl{g}{N}}{\pi_1} =
    \nodeAtPos{\truncl{g}{N}}{\pi_2}\tag{Lemma~\ref{lem:locTruncNodes}}\\%
    &\iff \pi_1 \sim_{\truncl{g}{N}} \pi_2%
  \end{align*}
\end{proof}

Most importantly, we obtain the intuitively expected property that
local truncations yield smaller term graphs w.r.t.\ $\lebotg$:
\begin{lemma}
  \label{lem:lebotLocTrunc}
  For each $g \in \iptgraphs$ and $N\subseteq N^g$, we have
  $\truncl{g}{N} \lebotg g$.
\end{lemma}
\begin{proof}
  We use Corollary~\ref{cor:chaTgraphPoA} to show
  this. (\ref{item:chaTGraphPoA1}) follows immediately from
  Lemma~\ref{lem:locTruncSim}. For (\ref{item:chaTGraphPoA2}), let
  $\pi \in \pos{\truncl{g}{N}}$ with $\truncl{g}{N}(\pi) \in
  \Sigma$. Hence, $\nodeAtPos{\truncl{g}{N}}{\pi} \nin N$ and we
  can reason as follows:
  \[
  \truncl{g}{N}(\pi) =
  \glab^{\truncl{g}{N}}(\nodeAtPos{\truncl{g}{N}}{\pi}) =
  \glab^g(\nodeAtPos{\truncl{g}{N}}{\pi}) \stackrel{\text{Lem.~\ref{lem:locTruncNodes}}}=
  \glab^g(\nodeAtPos{g}{\pi}) = g(\pi).
  \]
\end{proof}

The following property summarises the core property that we
require for an adequate notion of reduction context: The reduction
context of a reduction step is the maximal substructure that is
guaranteed to be preserved by the reduction.
\begin{lemma}
  \label{lem:stepLocTrunc}
  Given a graph reduction step $g \to[n,\rho,n'] h$ with
  $g,h\in\iptgraphs$, we have $\truncl{g}{n} \isom
  \truncl{h}{n'}$. The corresponding isomorphism is given by
    \[
  \phi(m) =
  \begin{cases}
    m &\text{if } m \neq n\\
    n' &\text{if } m = n
  \end{cases} \qquad \text{for all } m\in N^{\truncl{g}{n}}
  \]
\end{lemma}
\begin{proof}
  At first, observe that $n'$, the root of the reduct, is either a
  fresh node from $\rho$, or a node reachable from $n$ in $g$. Hence,
  we know that
  \begin{equation}
  m \in N^{\truncl g n} \setminus \set{n} \quad \text{implies that}
  \quad m \neq n'.
  \tag{$*$}\label{eq:stepLocTrunc}
\end{equation}

  In order to prove that $\phi\fcolon N^{\truncl{g}{n}} \to
  N^{\truncl{h}{n'}}$ is well-defined, we have to show that
  $N^{\truncl{g}{n}}\setminus\set{n} \subseteq N^{\truncl{h}{n'}}$:
  Let $m \in N^{\truncl{g}{n}}\setminus\set{n}$. We will show by
  induction on $\depth{\truncl{g}{n}}{m}$, that $m \in
  N^{\truncl{h}{n'}}$.

  If $\depth{\truncl{g}{n}}{m} = 0$, then $m = r^{\truncl{g}{n}} = r^g
  = r^h = r^{\truncl{h}{n}} \in N^{\truncl{h}{n}}$, where $r^g = r^h$
  holds because $n \neq r^g$. If $\depth{\truncl{g}{n}}{m} > 0$, then
  there is some $m' \in N^{\truncl g n}$ with $\depth{\truncl g n}{m'}
  < \depth{\truncl g n}{m}$ and $\gsuc^{\truncl g n}_i(m') = m$ for
  some $i \in \nats$. Hence, $m' \neq n$, which means that also
  $\gsuc^g_i(m') = m$ and that, by induction hypothesis, $m' \in
  N^{\truncl{h}{n'}}$. Since, in a graph reduction step, only edges to
  the redex node $n$ are redirected, we have that $\gsuc^h_i(m') \neq
  \gsuc^g_i(m')$ iff $\gsuc^g_i(m') = n$. Thus, as $\gsuc^g_i(m') = m
  \neq n$, we have $\gsuc^h_i(m') = \gsuc^g_i(m') = m$. Moreover, by
  \eqref{eq:stepLocTrunc}, we know that $m' \neq n'$. Thus, $m =
  \gsuc^h_i(m') \in N^{\truncl{h}{n'}}$.

  Next, we show that $\phi$ is a homomorphism from $\truncl{g}{n}$ to
  $\truncl{h}{n'}$. The root condition is satisfied as follows:
  \[
  \phi(r^{\truncl{g}{n}}) = \phi(r^g) = 
  \left\{
  \begin{aligned}
    &r^g &&\text{if } r^g \neq n\\
    &n' &&\text{if } r^g = n
  \end{aligned}
  \right\}
  = r^h = r^{\truncl{h}{n'}}.
  \]

  For the labelling and successor condition, assume some $m \in
  N^{\truncl{g}{n}}$. If $m = n$, then $\phi(m) = n'$ and the
  labelling and successor condition follow immediately from the
  construction of $\truncl{g}{n}$ and $\truncl{h}{n'}$. If $m \neq n$,
  then $\phi(m) = m$ and, by \eqref{eq:stepLocTrunc}, $m \neq
  n'$. Since the labelling of nodes is not changed by a reduction
  step, we have
  \[
  \glab^{\truncl{h}{n'}}(\phi(m)) = \glab^{\truncl{h}{n'}}(m) =
  \glab^h(m) = \glab^g(m) = \glab^{\truncl{g}{n}}(m) =
  \glab^{\truncl{g}{n}}(\phi(m)).
  \]
  For the successor condition, first assume that $\gsuc^g_i(m) =
  n$. Then the edge to $n$ is redirected to $n'$ by the reduction step,
  i.e.\ $\gsuc^h_i(m) = n'$, and we have
  \[
  \gsuc^{\truncl{h}{n'}}_i(\phi(m)) = \gsuc^{\truncl{h}{n'}}_i(m) =
  \gsuc^{h}_i(m) = n' = \phi(n) = \phi(\gsuc^g_i(m)) = \phi(\gsuc^{\truncl{g}{n}}_i(m)).
  \]
  If, on the other hand, $\gsuc^g_i(m) \neq n$, the edge is retained,
  i.e.\ $\gsuc^h_i(m) = \gsuc^g_i(m)$, and we have
  \[
  \gsuc^{\truncl{h}{n'}}_i(\phi(m)) = \gsuc^{\truncl{h}{n'}}_i(m) =
  \gsuc^{h}_i(m) = \gsuc^{g}_i(m) = \phi(\gsuc^g_i(m)) =
  \phi(\gsuc^{\truncl{g}{n}}_i(m)).
  \]

  The injectivity of $\phi$ follows from the fact that $\phi(m) = n'$
  if $m = n$ and that, by \eqref{eq:stepLocTrunc}, $\phi(m) = m \neq
  n'$ if $m \neq n$. Hence, according Lemma~\ref{lem:isomBij}, $\phi$ is an
  isomorphism, i.e.\ $\truncl{g}{n} \isom \truncl{h}{n'}$.
\end{proof}

As an easy consequence of this, we obtain that $\truncl{g}{n}$ is
indeed an adequate notion of reduction context.
\begin{proposition}
  \label{prop:stepContext}
  Given a graph reduction step $g \to[n,\rho,n'] h$, we have
  $\truncl{g}{n} \lebotg g, h$.
\end{proposition}
\begin{proof}
  Lemma~\ref{lem:lebotLocTrunc} yields $\truncl{g}{n} \lebotg g$. By
  Lemma~\ref{lem:stepLocTrunc} and Lemma~\ref{lem:lebotLocTrunc}, we
  get $\truncl{g}{n} \isom \truncl{h}{n'} \lebotg h$.
\end{proof}

The following lemma provides a convenient characterisation of local
truncations in terms of labelled quotient trees:
\begin{lemma}
  \label{lem:locTrunc}
  For each $g \in \iptgraphs$ and $n\in N^g$, the local truncation
  $\truncl{g}{n}$ has the following canonical labelled quotient tree
  $(P,l,\sim$):
  \begin{gather*}
    \begin{aligned}
      P &= \setcom{\pi \in \pos{g}}{\forall \pi' < \pi \colon \pi'
        \nin \nodePos{g}{n}}%
      \\%
      \sim &=\ \sim_g\cap\ P\times P\\
    \end{aligned}\qquad
    l(\pi) =
    \begin{cases}
      g(\pi) &\text{if } \pi \nin \nodePos{g}{n}\\
      \bot &\text{if } \pi \in \nodePos{g}{n}
    \end{cases}
    \quad \text{for all } \pi \in P
  \end{gather*}
  In particular, given $\pi \in \pos{g}$, we have that $g(\pi) =
  \truncl{g}{n}(\pi)$ if $\pi' \nin \nodePos{g}{n}$ for each $\forall
  \pi' \le \pi$.
\end{lemma}
\begin{proof}
  The last statement above follows immediately from the preceding
  characterisation of $(P,l,\sim)$. We will show in the following that
  $(P,l,\sim)$ is equal to $(\pos{\truncl g
    n},\truncl{g}{n}(\cdot),\sim_{\truncl g n})$.

  By Lemma~\ref{lem:locTruncNodes} $\pos{\truncl{g}{n}} \subseteq
  \pos{g}$. Therefore, in order to prove that $\pos{\truncl{g}{n}}
  \subseteq P$, we assume some $\pi \in \pos{\truncl{g}{n}}$ and show
  by induction on the length of $\pi$ that no proper prefix of $\pi$
  is a position of $n$ in $g$. The case $\pi=\emptyseq$ is trivial as
  $\emptyseq$ has no proper prefixes. If $\pi = \pi' \concat\seq{i}$,
  we can assume by induction that $\pi' \in P$ since $\pi' \in
  \pos{\truncl{g}{n}}$. Consequently, no proper prefix of $\pi'$ is in
  $\nodePos{g}{n}$. It thus remains to be shown that $\pi'$ itself is
  not in $\nodePos{g}{n}$. Since $\pi' \concat\seq{i} \in
  \pos{\truncl{g}{n}}$, we know that
  $\gsuc^{\truncl{g}{n}}_i(\nodeAtPos{\truncl{g}{n}}{\pi'})$ is
  defined. Therefore, $\nodeAtPos{\truncl{g}{n}}{\pi'}$ cannot be $n$,
  and since, by Lemma~\ref{lem:locTruncNodes},
  $\nodeAtPos{\truncl{g}{n}}{\pi'} = \nodeAtPos{g}{\pi'}$, neither can
  $\nodeAtPos{g}{\pi'}$. In other words, $\pi' \nin \nodePos{g}{n}$.

  For the converse direction $P \subseteq \pos{\truncl{g}{n}}$, assume
  some $\pi \in P$. We will show by induction on the length of $\pi$,
  that then $\pi \in \pos{\truncl{g}{n}}$. The case $\pi = \emptyseq$
  is trivial. If $\pi = \pi' \concat \seq{i}$, then also $\pi' \in P$
  which, by induction, implies that $\pi' \in
  \pos{\truncl{g}{n}}$. Since $\pi \in \pos{g}$, according to
  Lemma~\ref{lem:locTruncNodes}, we have that $\rank{g}{\pi'} >
  i$. Let $m = \nodeAtPos{g}{\pi'}$. According to
  Lemma~\ref{lem:locTruncNodes}, $m =
  \nodeAtPos{\truncl{g}{n}}{\pi'}$. Since $\pi \in P$, we have that
  $\pi' \nin \nodePos{g}{n}$ and thus $m \neq n$. Hence, according to
  the definition of $\truncl{g}{n}$, $\gsuc^{\truncl{g}{n}}(m) =
  \gsuc^g(m)$ which implies that $\rank{\truncl{g}{n}}{\pi'} >
  i$. Consequently, $\pi \in \pos{\truncl{g}{n}}$.

  The equality $\sim\ =\ \sim_{\truncl{g}{n}}$ follows directly from
  Lemma~\ref{lem:locTruncSim} and the equality $P =
  \pos{\truncl{g}{n}}$.

  For the equality $l = \truncl{g}{n}(\cdot)$, consider some $\pi \in
  \pos{\truncl{g}{n}}$. Since $\nodeAtPos{g}{\pi} = n$ iff $\pi \in
  \nodePos{g}{n}$, we can reason as follows:
  \[
  \truncl{g}{n}(\pi) =
  \glab^{\truncl{g}{n}}(\nodeAtPos{\truncl{g}{n}}{\pi}) \stackrel{\text{Lem.~\ref{lem:locTruncNodes}}}{=}
  \glab^{\truncl{g}{n}}(\nodeAtPos{g}{\pi}) =
  \begin{cases}
    g(\pi) &\text{if } \pi  \nin \nodePos{g}{n}\\
    \bot &\text{if } \pi  \in \nodePos{g}{n}
  \end{cases}
  \]
\end{proof}

\subsection{Strong Convergence}
\label{sec:strong-convergence}

Now that we have an adequate notion of reduction context, we can
define strong $\prs$-convergence on term graphs analogously to strong
$\prs$-convergence on terms. For strong $\mrs$-convergence, we simply
take the same notion of depth that we already used for the definition
of strict truncation and thus the metric space.
\begin{definition}
  Let $\calR$ be a GRS.
  \begin{enumerate}[(i)]
  \item The \emph{reduction context} $c$ of a graph reduction step
    $\phi\fcolon g \to[n] h$ is the term graph
    $\canon{\truncl{g}{n}}$. We write $\phi\fcolon g \to[c] h$ to
    indicate the reduction context of a graph reduction step.
  \item Let $S = (g_\iota \to[n_\iota] g_{\iota+1})_{\iota<\alpha}$ be
    a reduction in $\calR$. $S$ is strongly $\mrs$-continuous in
    $\calR$, denoted $S\fcolon g_0 \macont[\calR]$, if $\lim_{\iota
      \limto \lambda} g_\iota = g_\lambda$ and $\lim_{\iota
      \limto\lambda}\depth{g_\iota}{n_\iota} = \infty$ for each limit
    ordinal $\lambda < \alpha$. $S$ strongly $\mrs$-converges to $g$
    in $\calR$, denoted $S\fcolon g_0 \mato[\calR] g$, if it is
    strongly $\mrs$-continuous, $\lim_{\iota \limto \wsuc\alpha}
    g_\iota = g$, and $\lim_{\iota
      \limto\alpha}\depth{g_\iota}{n_\iota} = \infty$ in case $\alpha$
    is a limit ordinal.
  \item Let $S = (g_\iota \to[c_\iota] g_{\iota+1})_{\iota<\alpha}$ be
    a reduction in $\calR$. $S$ is strongly $\prs$-continuous in
    $\calR$, denoted $S\fcolon g_0 \pacont[\calR]$, if $\liminf_{\iota
      \limto \lambda} c_\iota = g_\lambda$ for each limit ordinal
    $\lambda < \alpha$. $S$ strongly $\prs$-converges to $g$ in
    $\calR$, denoted $S\fcolon g_0 \pato[\calR] g$, if it is strongly
    $\prs$-continuous and either $g = \liminf_{\iota \limto \alpha}
    c_\iota$ or $g = g_\alpha$ in case $S$ is closed.
  \end{enumerate}
\end{definition}

Note that we have to extend the signature of $\calR$ to $\Sigma_\bot$
for the definition of strong $\prs$-convergence. However, we can
obtain the total fragment of strong $\prs$-convergence if we restrict
ourselves to total term graphs in $\ictgraphs$: A reduction $(g_\iota
\to[\calR_\bot] g_{\iota+1})_{\iota < \alpha}$ $\prs$-converging to
$g$ is called \emph{total} if $g$ as well as each $g_\iota$ is total,
i.e.\ an element of $\ictgraphs$.

Since the partial order $\lebotg$ forms a complete semilattice on
$\ipctgraphs$, strong $\prs$-continuity coincides with strong
$\prs$-convergence:
\begin{proposition}
  \label{prop:strPContConv}
  Every strongly $\prs$-continuous reduction is strongly $\prs$-convergent.
\end{proposition}
\begin{proof}
  Follows immediately from Corollary~\ref{cor:lebot1Liminf}.
\end{proof}

The following technical lemma confirms the intuition that changes
during a continuous reduction must be caused by a reduction step that
was applied at the position where the difference is observed or above.
\begin{lemma}
  \label{lem:pContCxt}
  Let $(g_\iota \to[n_\iota,\rho,m_\iota]
  g_{\iota+1})_{\iota<\alpha}$ be a strongly $\prs$-continuous
  reduction in a GRS with its reduction contexts $c_\iota =
  \canon{\truncl{g_\iota}{n_\iota}}$ such that there are $\beta \le
  \gamma < \alpha$ and $\pi \in \pos{c_\beta}\cap\pos{c_\gamma}$ with
  $c_\beta(\pi) \neq c_\gamma(\pi)$. Then there is a position $\pi'
  \le \pi$ and an index $\beta \le \iota \le \gamma$ such that $\pi'
  \in \nodePos{g_\iota}{n_\iota}$.
\end{lemma}
\begin{proof}
  Throughout the poof, we can assume that
  \begin{equation}
    g_\iota(\pi) = c_\iota(\pi) \qquad \text{ if } \beta \le \iota \le \gamma
    \text{ and } \pi \in \pos{c_\iota}.
    \tag{$*$}
    \label{eq:pContCxt}
  \end{equation}
  If this would not be the case, then, by Lemma~\ref{lem:locTrunc},
  there is a $\pi'\le \pi$ such that $\pi'\in
  \nodePos{g_\iota}{n_\iota}$, i.e.\ the statement that we want to
  prove holds.

  We proceed with an induction on $\gamma$. The case $\gamma = \beta$
  is trivial.

  Let $\gamma = \iota+1> \beta$. We then consider two cases: If
  $\pi\nin \pos{c_{\iota}}$, we are done as this together with the
  assumption $\pi\in\pos{c_{\gamma}}$ implies, by the definition of
  reduction steps, that $\pi'\in\nodePos{g_{\iota}}{n_{\iota}}$ for
  some $\pi' < \pi$. If, on the other hand, $\pi\in \pos{c_{\iota}}$,
  then we can assume that $c_\beta(\pi) = c_{\iota}(\pi)$ since
  otherwise the proof goal follows immediately from the induction
  hypothesis. Consequently, we have that
  \[
  \truncl{g_\gamma}{m_{\iota}}(\pi)
  \stackrel{\text{Lem.~\ref{lem:stepLocTrunc}}}=
  \truncl{g_{\iota}}{n_{\iota}}(\pi) = c_{\iota}(\pi) =
  c_\beta(\pi) \neq c_\gamma(\pi) \stackrel{\eqref{eq:pContCxt}}=
  g_\gamma(\pi)
  \]
  The thus obtained inequality $\truncl{g_\gamma}{m_{\iota}}(\pi)
  \neq g_\gamma(\pi)$ implies, by Lemma~\ref{lem:locTrunc}, that there
  is a $\pi' \le \pi$ such that $\pi'\in
  \nodePos{g_\gamma}{m_{\iota}}$. According to
  Lemma~\ref{lem:stepLocTrunc} there is an isomorphism $\phi\fcolon
  \truncl{g_{\iota}}{n_{\iota}} \to \truncl{g_\gamma}{m_{\iota}}$
  with $\phi(n_{\iota}) = m_{\iota}$. This means, by
  Corollary~\ref{cor:isomOcc}, that $\nodePos{g_{\iota}}{n_{\iota}}
  = \nodePos{g_\gamma}{m_{\iota}}$. Hence, $\pi' \in
  \nodePos{g_{\iota}}{n_{\iota}}$.

  Let $\gamma$ be a limit ordinal. By \eqref{eq:pContCxt}, we know
  that $g_\gamma(\pi) = c_\gamma(\pi) \neq c_\beta(\pi)$. According to
  Corollary~\ref{cor:lebot1Liminf}, the inequality $g_\gamma(\pi) \neq
  c_\beta(\pi)$ is only possible if there is a $\pi'\le \pi$ and a
  $\beta\le \iota <\gamma$ such that $c_\iota(\pi') \neq
  c_\beta(\pi')$. Hence, we can invoke the induction hypothesis (for
  the position $\pi'$ instead of $\pi$) which immediately yields the
  proof goal.
\end{proof}

By combining the characterisation of the limit inferior from
Corollary~\ref{cor:lebot1Liminf} and the characterisation of local
truncations from Lemma~\ref{lem:locTrunc}, we obtain the following
characterisation of the limit of a strongly $\prs$-convergent
reduction:
\begin{lemma}
  \label{lem:strConvRes}
  Let $S = (g_\iota \to[n_\iota] g_{\iota+1})_{\iota<\alpha}$ be an
  open reduction in a GRS strongly $\prs$-converging to $g$. Then $g$
  has the following canonical labelled quotient tree $(P,l,\sim)$:
  \begin{align*}
    P &= \bigcup_{\beta<\alpha} \setcom{\pi \in \pos{g_\beta}}{\forall
      \pi' < \pi \forall \beta \le \iota < \alpha\colon \pi' \nin
      \nodePos{g_\iota}{n_\iota}}\\
    \sim &= \left(\bigcup_{\beta<\alpha}\bigcap_{\beta\le\iota<\alpha}
      \sim_{g_\iota}\right) \cap P \times P\\
    l(\pi) &=
    \begin{cases}
      g_\beta(\pi) &\text{if } \exists \beta < \alpha\forall
      \beta\le\iota<\alpha\colon \pi \nin \nodePos{g_\iota}{n_\iota}\\
      \bot &\text{otherwise}
    \end{cases}
    \quad \text{for all } \pi \in P
  \end{align*}
  In particular, given $\beta < \alpha$ and $\pi \in \pos{g_\beta}$,
  we have that $g(\pi) = g_\beta(\pi)$ if $\pi' \nin
  \nodePos{g_\iota}{n_\iota}$ for all $\pi' \le \pi$ and $\beta \le
  \iota < \alpha$.
\end{lemma}
\begin{proof}
  
  The last statement above follows immediately from the preceding
  characterisation of $(P,l,\sim)$. We will show in the following that
  $(P,l,\sim)$ is equal to $(\pos{g},g(\cdot),\sim_g)$.

  Let $c_\iota = \canon{\truncl{g_\iota}{n_\iota}}$ for each $\iota < \alpha$.
  At first we show that $\pos{g} \subseteq P$. To this end let $\pi
  \in \pos{g}$. Since $g = \liminf_{\iota\limto\alpha} c_\iota$, this
  means, by Corollary~\ref{cor:lebot1Liminf}, that
  \begin{gather*}
    \text{there is some } \beta< \alpha \text{ with } \pi \in
    \pos{c_\beta} \text{ and } c_\iota(\pi') = c_\beta(\pi')\text{ for
      all } \pi' < \pi \text{ and } \beta\le \iota <
    \alpha. \tag{1}\label{eq:strConvRes1}
  \end{gather*}
  Since, according to Lemma~\ref{lem:locTruncNodes}, $\pos{c_\beta}
  \subseteq \pos{g_\beta}$, we also have $\pi \in \pos{g_\beta}$. In
  order to prove that $\pi \in P$, we assume some $\pi' < \pi$ and
  $\beta \le \iota < \alpha$ and show that $\pi' \nin
  \nodePos{g_\iota}{n_\iota}$. Since $\pi'$ is a proper prefix of a
  position in $c_\beta$, we have that $c_\beta(\pi') \in \Sigma$. By
  \eqref{eq:strConvRes1}, also $c_\iota(\pi') \in \Sigma$. Hence,
  according to Lemma~\ref{lem:locTrunc}, $\pi' \nin
  \nodePos{g_\iota}{n_\iota}$.

  For the converse direction $P \subseteq \pos{g}$, we assume some
  $\pi \in P$ and show that then $\pi \in \pos{g}$. Since $\pi \in P$,
  we have that
  \begin{gather*}
    \text{there is some } \beta < \alpha \text{ with } \pi \in
    \pos{g_\beta} \text{ and } \pi' \nin \nodePos{g_\iota}{n_\iota}
    \text{ for all } \pi' < \pi \text{ and } \beta \le \iota
    <\alpha. \tag{2}\label{eq:strConvRes2}
  \end{gather*}
  In particular, we have that $\pi' \nin \nodePos{g_\beta}{n_\beta}$
  for all $\pi' < \pi$. Hence, by Lemma~\ref{lem:locTrunc}, $\pi \in
  \pos{c_\beta}$. According to Corollary~\ref{cor:lebot1Liminf}, it
  remains to be shown that $c_\gamma(\pi') = c_\beta(\pi')$ for all
  $\pi' <\pi$ and $\beta \le \gamma <\alpha$. We will do that by an
  induction on $\gamma$:
  
  The case $\gamma = \beta$ is trivial. For $\gamma = \iota + 1 >
  \beta$, let $g_\iota \to[n_\iota,\rho_\iota,n'_\iota] g_\gamma$ be
  the $\iota$-th reduction step and $\pi'<\pi$. By
  Lemma~\ref{lem:stepLocTrunc}, we then have $c_\iota \isom
  \truncl{g_\iota}{n_\iota} \isom \truncl{g_\gamma}{n'_\iota}$. We can
  thus reason as follows:
  \[
  c_\beta(\pi') \stackrel{\text{ind.\ hyp.}}{=} c_\iota(\pi') \stackrel{\text{Lem.~\ref{lem:stepLocTrunc}}}{=}
  \truncl{g_\gamma}{n'_\iota}(\pi') \stackrel{\text{Lem.~\ref{lem:locTrunc}}}{=} g_\gamma(\pi') \stackrel{\text{Lem.~\ref{lem:locTrunc}}}{=} 
  \truncl{g_\gamma}{n_\gamma}(\pi') = c_\gamma(\pi')
  \]
  The first application of Lemma~\ref{lem:locTrunc} above is justified
  by the fact that $\pi'<\pi\in\pos{c_\beta}$ and thus $c_\beta(\pi')
  \neq \bot$. The second application of Lemma~\ref{lem:locTrunc} is
  justified by \eqref{eq:strConvRes2}.

  If $\gamma > \beta$ is a limit ordinal, then $g_\gamma =
  \liminf_{\iota \limto \gamma} c_\iota$ and we can apply
  Corollary~\ref{cor:lebot1Liminf}. Since $\pi' \in \pos{c_\beta}$
  and, by induction hypothesis, $c_\iota(\pi'') = c_\beta(\pi'')$ for
  all $\pi'' \le \pi'$, $\beta \le \iota < \gamma$, we thus obtain
  that $g_\gamma(\pi') = c_\beta(\pi')$. Since, according to
  \eqref{eq:strConvRes2}, $\pi'' \nin \nodePos{g_\gamma}{n_\gamma}$
  for each $\pi'' \le \pi'$, we have by Lemma~\ref{lem:locTrunc} that
  $g_\gamma(\pi') = c_\gamma(\pi')$. Hence, $c_\gamma(\pi') =
  c_\beta(\pi')$.

  The inclusion $\sim_g\ \subseteq\ \sim$ follows immediately from
  Corollary~\ref{cor:lebot1Liminf} and the equality $P = \pos{g}$
  since $\sim_{c_\iota}\ \subseteq\ \sim_{g_\iota}$ for all $\iota <
  \alpha$ according to Lemma~\ref{lem:locTruncSim}.

  For the reverse inclusion $\sim\ \subseteq\ \sim_g$, assume that
  $\pi_1 \sim \pi_2$. That is, $\pi_1,\pi_2 \in P$ and there is some
  $\beta_0 < \alpha$ such that $\pi_1 \sim_{g_\iota} \pi_2$ for all
  $\beta_0 \le \iota < \alpha$. Since $\pi_1,\pi_2 \in P = \pos{g}$,
  we know, by Corollary~\ref{cor:lebot1Liminf}, that there are
  $\beta_1,\beta_2 < \alpha$ such that $\pi_k \in \pos{c_\iota}$ for
  all $\beta_k \le \iota < \alpha$. Let $\beta =
  \max\set{\beta_0,\beta_1,\beta_2}$. For each $\beta \le \iota <
  \alpha$, we then obtain that $\pi_1 \sim_{g_\iota} \pi_2$ and
  $\pi_1,\pi_2 \in \pos{c_\iota}$. By Lemma~\ref{lem:locTruncSim},
  this is equivalent to $\pi_1 \sim_{c_\iota} \pi_2$. Applying
  Corollary~\ref{cor:lebot1Liminf} then yields $\pi_1 \sim_g \pi_2$.

  Finally, we show that $l = g(\cdot)$. To this end, let $\pi \in
  P$. We distinguish two mutually exclusive cases. For the first case,
  we assume that
  \begin{gather}
    \text{there is some $\beta < \alpha$ such that $c_\iota(\pi) =
      c_\beta(\pi)$ for all $\beta \le \iota <
      \alpha$.}\tag{3}\label{eq:strConvRes3}
  \end{gather}
  By Corollary~\ref{cor:lebot1Liminf}, we know that then $g(\pi) =
  c_\beta(\pi)$. Next, assume that there is some $\beta'< \alpha$ with
  $\pi\nin\nodePos{g_\iota}{n_\iota}$ for all $\beta' \le \iota <
  \alpha$. W.l.o.g.\ we can assume that $\beta = \beta'$. Hence,
  $l(\pi) = g_\beta(\pi)$. Moreover, since $\pi \nin
  \nodePos{g_{\beta}}{n_{\beta}}$, we have that $g_\beta(\pi) =
  c_\beta(\pi)$ according to Lemma~\ref{lem:locTrunc}. We thus conclude that $l(\pi) = g_\beta(\pi) =
  c_\beta(\pi) = g(\pi)$. Now assume there is no such $\beta'$, i.e.\
  for each $\beta' < \alpha$ there is some $\beta' \le \iota < \alpha$
  with $\pi \in \nodePos{g_\iota}{n_\iota}$. Consequently, $l(\pi) =
  \bot$ and, by Lemma~\ref{lem:locTrunc}, we have for each $\beta' <
  \alpha$ some $\beta' \le \iota < \alpha$ such that $c_\iota(\pi) =
  \bot$. According to \eqref{eq:strConvRes3}, the latter implies that
  $c_\iota(\pi) = \bot$ for all $\beta \le \iota < \alpha$. By
  Corollary~\ref{cor:lebot1Liminf}, we thus obtain that $g(\pi) =
  \bot = l(\pi)$.

  Next, we consider the negation of \eqref{eq:strConvRes3}, i.e.\ that
  \begin{gather}
    \text{for all $\beta < \alpha$ there is a $\beta \le \iota <
      \alpha$ such that $\pi \in \pos{c_\iota}\cap \pos{c_\beta}$
      implies $c_\iota(\pi) \neq
      c_\beta(\pi)$.}\tag{4}\label{eq:strConvRes4}
  \end{gather}
  By Corollary~\ref{cor:lebot1Liminf}, we have that $g(\pi) =
  \bot$. Since $\pi\in P = \pos{g}$, we can apply
  Corollary~\ref{cor:lebot1Liminf} again to obtain a $\gamma < \alpha$
  with $\pi\in\pos{c_\iota}$ and $c_\iota(\pi') = c_{\gamma}(\pi')$
  for all $\pi' < \pi$ and $\gamma \le \iota < \alpha$. Combining this
  with \eqref{eq:strConvRes4} yields that for each $\gamma \le \beta <
  \alpha$ there is a $\beta \le \iota <\alpha$ with $c_\iota(\pi) \neq
  c_\beta(\pi)$. According to Lemma~\ref{lem:pContCxt}, this can only
  happen if there is a $\beta\le \gamma' \le \iota$ and a $\pi'\le\pi$
  such that $\pi' \in \nodePos{g_{\gamma'}}{n_{\gamma'}}$. Since $\pi$
  has only finitely many prefixes, we can apply the infinite pigeon
  hole principle to obtain a single prefix $\pi' \le \pi$ such that
  for each $\beta < \alpha$ there is some $\beta \le\iota < \alpha$
  with $\pi' \in \nodePos{g_\iota}{n_\iota}$. However, $\pi'$ cannot
  be a proper prefix of $\pi$ since this would imply that $\pi \nin
  P$. Thus we can conclude that for each $\beta < \alpha$ there is
  some $\beta \le\iota < \alpha$ such that $\pi \in
  \nodePos{g_\iota}{n_\iota}$. Hence, $l(\pi) = \bot = g(\pi)$.
\end{proof}

The benefit of strong $\prs$-convergence over strong
$\mrs$-convergence is that the former has a more fine-grained
characterisation of divergence. Strong $\prs$-convergence allows for
local divergence, i.e.\ parts of a term graph that do not become
persistent along a transfinite reduction. We will call such parts
volatile:
\begin{definition}[volatility]
  Let $S = (g_\iota \to[n_\iota] g_{\iota+1})_{\iota<\alpha}$ be an
  open graph reduction. A position $\pi \in \nats^*$ is said to be
  \emph{volatile} in $S$ if, for each $\beta < \alpha$, there is some
  $\beta \le \gamma < \alpha$ such that $\pi \in
  \nodePos{g_\gamma}{n_\gamma}$. If $\pi$ is volatile in $S$ and no proper
  prefix of $\pi$ is volatile in $S$, then $\pi$ is called
  \emph{outermost-volatile} in $S$.
\end{definition}

As for infinitary term rewriting~\cite{bahr10rta2}, local divergence in
a strongly $\prs$-converging reduction can be characterised by
volatile positions:
\begin{lemma}
  \label{lem:volBot}
  Let $S = (g_\iota \to[n_\iota] g_{\iota+1})_{\iota<\alpha}$ be an
  open reduction in a GRS strongly $\prs$-converging to $g$. Then, for
  every $\pi\in \nats^*$, we have the following:
  \begin{enumerate}[(i)]
  \item If $\pi$ is volatile in $S$, then $\pi \in \pos{g}$ implies $g(\pi) =
    \bot$.
    \label{item:volBot1}
  \item $g(\pi) = \bot$ iff
    \label{item:volBot2}
    \begin{enumerate}[(a)]
    \item $\pi$ is outermost-volatile in $S$, or
      \label{item:volBot2a}
    \item there is some $\beta < \alpha$ such that $g_\beta(\pi) = \bot$
      and $\pi' \nin \nodePos{g_\iota}{n_\iota}$ for all $\pi' \le
      \pi$ and $\beta \le \iota < \alpha$.
      \label{item:volBot2b}
    \end{enumerate}
  \item Let $g_\iota$ be total for all $\iota < \alpha$. Then $g(\pi)
    = \bot$ iff $\pi$ is outermost-volatile in $S$.
    \label{item:volBot3}
  \end{enumerate}
\end{lemma}
\begin{proof}
  (\ref{item:volBot1}) follows immediately from
  Lemma~\ref{lem:strConvRes}.

  (\ref{item:volBot2}) At first consider the ``only if'' direction:
  Suppose that $g(\pi) = \bot$. We will show that
  (\ref{item:volBot2b}) holds whenever (\ref{item:volBot2a}) does not
  hold. To this end, suppose that $\pi$ is not outermost-volatile in
  $S$. Since $g(\pi) = \bot$, we know that $g(\pi') \in \Sigma$ for
  all $\pi' < \pi$. By Clause (\ref{item:volBot1}), this implies that
  no prefix of $\pi$ is volatile. Consequently, $\pi$ itself is not
  volatile in $S$ either as it would be outermost-volatile
  otherwise. Hence, no prefix of $\pi$ is volatile in $S$, i.e.\ there
  is some $\beta < \alpha$ such that $\pi' \nin
  \nodePos{g_\iota}{n_\iota}$ for all $\pi' \le \pi, \beta\le \iota<
  \alpha$. Additionally, by Lemma~\ref{lem:strConvRes}, we obtain that
  $g_\beta(\pi) = g(\pi) = \bot$. That is, (\ref{item:volBot2b})
  holds.

  For the ``if'' direction, we show that both (\ref{item:volBot2a})
  and (\ref{item:volBot2b}) independently imply that $g(\pi) = \bot$:
  The implication from (\ref{item:volBot2b}) follows immediately from
  Lemma~\ref{lem:strConvRes}. For the implication from
  (\ref{item:volBot2a}), let $\pi$ be outermost-volatile in $S$. Since
  no proper prefix of $\pi$ is volatile in $S$, we find some $\beta <
  \alpha$ such that $\pi' \nin \nodePos{g_\iota}{n_\iota}$ for all
  $\pi' < \pi, \beta \le \iota < \alpha$. Since $\pi$ itself is
  volatile, there is some $\beta \le \gamma < \alpha$ such that $\pi
  \in \pos{g_\gamma}$. As we have, in particular, that $\pi' \nin
  \nodePos{g_\iota}{n_\iota}$ for all $\pi' < \pi, \gamma \le \iota <
  \alpha$, we have, by Lemma~\ref{lem:strConvRes}, that $\pi \in
  \pos{g}$. Consequently, according to Clause~(\ref{item:volBot1}), we
  have that $g(\pi) = \bot$.

  (\ref{item:volBot3}) is a special case of (\ref{item:volBot2}): If
  each $g_\iota$ is total, then (\ref{item:volBot2b}) cannot be true.
\end{proof}

With this in mind, we can characterise total reductions as exactly
those that lack volatile positions:
\begin{lemma}[total reductions]
  \label{lem:totalRed}
  Let $\calR$ be a GRS, $g$ a total term graph in $\calR$, and
  $S\fcolon g \pato[\calR] h$. $S\fcolon g \pato[\calR] h$ is total
  iff no prefix of $S$ has a volatile position.
\end{lemma}
\begin{proof}
  The ``only if'' direction follows straightforwardly from
  Lemma~\ref{lem:volBot}.

  We prove the ``if'' direction by induction on the length of $S$. If
  $\len{S} = 0$, then the totality of $S$ follows from the assumption
  of $g$ being total. If $\len{S}$ is a successor ordinal, then the
  totality of $S$ follows from the induction hypothesis since single
  reduction steps preserve totality. If $\len{S}$ is a limit ordinal,
  then the totality of $S$ follows from the induction hypothesis using
  Lemma~\ref{lem:volBot}.
\end{proof}

Next we want to compare strong $\mrs$- and $\prs$-convergence with the
ultimate goal of establishing the same relation between them as for
term rewriting (cf.\ Theorem~\ref{thr:strongExt}).

\begin{definition}[minimal positions]
  Let $g \in \itgraphs$ and $n \in N^g$. A position $\pi \in
  \nodePos{g}{n}$ is called minimal if no proper prefix $\pi' < \pi$
  is in $\nodePos{g}{n}$. The set of all minimal positions of $n$ in
  $g$ is denoted $\nodePosMin{g}{n}$.
\end{definition}

Minimal positions have the nice property that they are not affected by
term graph reductions:
\begin{lemma}
  \label{lem:stepDepth}
  Given a term graph reduction step $g \to[n,\rho,n'] h$, we have
  $\nodePosMin{g}{n} = \nodePosMin{h}{n'}$.
\end{lemma}
\begin{proof}
  We will show that $\nodePosMin{g}{n} \subseteq
  \nodePosMin{h}{n'}$. The converse inclusion is symmetric. Let $\pi
  \in \nodePosMin{g}{n}$. Hence, $\pi' \nin \nodePos{g}{n}$ for all
  $\pi' < \pi$ which, by Lemma~\ref{lem:locTrunc}, implies that $\pi
  \in \pos{\truncl{g}{n}}$. According to
  Lemma~\ref{lem:locTruncNodes}, $\pi \in
  \nodePos{\truncl{g}{n}}{n}$. Since, by Lemma~\ref{lem:stepLocTrunc},
  there is an isomorphism $\phi\fcolon\truncl{g}{n} \homto
  \truncl{h}{n'}$ with $\phi(n) = n'$, we obtain, by
  Corollary~\ref{cor:isomOcc}, that $\pi \in
  \nodePos{\truncl{h}{n'}}{n'}$. By Lemma~\ref{lem:locTruncNodes},
  this implies $\pi \in \nodePos{h}{n'}$. Since $\pi \in
  \pos{\truncl{h}{n'}}$, we know, by Lemma~\ref{lem:locTrunc}, that
  $\pi' \nin\nodePos{h}{n'}$ for all $\pi' < \pi$. Combined, this
  means that $\pi \in \nodePosMin{h}{n'}$.
\end{proof}

In order to compare strong $\mrs$- and $\prs$-convergence, we consider
positions bounded by a certain depth.
\begin{definition}[bounded positions]
  Let $g \in \itgraphs$ and $d \in \nats$. We write $\posBound{d}{g}$
  for the set $\setcom{\pi \in \pos{g}}{\len{\pi} \le d}$ of
  positions of length at most $\pi$.
\end{definition}

Local truncations do not change positions bounded by the same depth or
above:
\begin{lemma}
  \label{lem:locTruncDepth}
  Let $g \in \iptgraphs$, $n \in N^G$ and $d \le \depth{g}{n}$. Then
  $\posBound{d}{\truncl{g}{n}} = \posBound{d}{g}$.
\end{lemma}
\begin{proof}
  $\posBound{d}{\truncl{g}{n}} \subseteq \posBound{d}{g}$ follows from
  Lemma~\ref{lem:locTrunc}. For the converse inclusion, assume some
  $\pi \in \posBound{d}{g}$. Since $\len\pi \le d \le \depth{g}{n}$,
  we know for each $\pi' < \pi$ that $\len{\pi'}< \depth{g}{n}$ and
  thus $\pi' \nin \nodePos{g}{n}$ . By Lemma~\ref{lem:locTrunc}, this
  implies that $\pi$ is in $\pos{\truncl{g}{n}}$ and thus also in
  $\posBound{d}{\truncl{g}{n}}$
\end{proof}

Additionally, reductions that only contract redexes at a depth $d$ or
below do not affect the positions bounded by $d$.
\begin{lemma}
  \label{lem:redBound}
  Let $S = (g_\iota \to[n_\iota] g_{\iota+1})_{\iota<\alpha}$ be a
  strongly $\prs$-convergent reduction in a GRS and $d \in \nats$ such
  that $\depth{g_\iota}{n_\iota} \ge d$ for all $\iota < \alpha$. Then
  $\posBound{d}{g_0} = \posBound{d}{g_\iota}$ for all $\iota \le
  \alpha$.
\end{lemma}
\begin{proof}
  We prove the statement by an induction on $\alpha$. The case $\alpha
  = 0$ is trivial.

  For $\alpha = \beta + 1$, let $g_\beta \to[n_\beta,n'_\beta]
  g_\alpha$ be the $\beta$-th step of $S$. Due to the induction
  hypothesis, it suffices to show that $\posBound{d}{g_0} =
  \posBound{d}{g_\alpha}$. By Lemma~\ref{lem:stepLocTrunc},
  $\truncl{g_\beta}{n_\beta} \isom \truncl{g_\alpha}{n'_\beta}$, and
  by Lemma~\ref{lem:stepDepth}, $\depth{g_\alpha}{n'_\beta} =
  \depth{g_\beta}{n_\beta} \ge d$. Hence, according to
  Lemma~\ref{lem:locTruncDepth}, we have both
  $\posBound{d}{\truncl{g_\beta}{n_\beta}} = \posBound{d}{g_\beta}$
  and $\posBound{d}{\truncl{g_\alpha}{n'_\beta}} =
  \posBound{d}{g_\alpha}$. We thus obtain the desired equation:
  \[
  \posBound{d}{g_0} \stackrel{\text{ind.\ hyp.}}=
  \posBound{d}{g_\beta}
  \stackrel{\text{Lem.~\ref{lem:locTruncDepth}}}=
  \posBound{d}{\truncl{g_\beta}{n_\beta}}
  \stackrel{\text{Lem.~\ref{lem:stepLocTrunc}}}=
  \posBound{d}{\truncl{g_\alpha}{n'_\beta}}
  \stackrel{\text{Lem.~\ref{lem:locTruncDepth}}}=
  \posBound{d}{g_\alpha}
  \]

  Lastly, let $\alpha$ be a limit ordinal. By the induction
  hypothesis, we only need to show $\posBound{d}{g_0} =
  \posBound{d}{g_\alpha}$. At first assume $\pi \in
  \posBound{d}{g_\alpha}$. Hence, by Lemma~\ref{lem:strConvRes}, there
  is some $\beta < \alpha$ such that $\pi \in
  \pos{g_\beta}$. Therefore, $\pi$ is in $\posBound{d}{g_\beta}$ and,
  by induction hypothesis, also in $\posBound{d}{g_0}$. Conversely,
  assume that $\pi \in \posBound{d}{g_0}$. Because
  $\depth{g_\iota}{n_\iota} \ge d$ for all $\iota < \alpha$, we have
  that $\pi' \nin \nodePos{g_\iota}{n_\iota}$ for all $\pi' <\pi$ and
  $\iota < \alpha$. According to Lemma~\ref{lem:strConvRes}, this
  implies that $\pi$ is in $\pos{g_\alpha}$ and thus also in
  $\posBound{d}{g_\alpha}$.
\end{proof}

The following two lemmas form the central properties that link strong
$\mrs$- and $\prs$-convergence:
\begin{lemma}
  \label{lem:volDepthInf}
  Let $S = (g_\iota \to[n_\iota] g_{\iota+1})_{\iota<\alpha}$ be a
  strongly $\prs$-convergent open reduction in a GRS. If $S$ has no
  volatile positions then $(\depth{g_\iota}{n_\iota})_{\iota<\alpha}$
  tends to infinity.
\end{lemma}
\begin{proof}
  We will prove the contraposition. To this end, assume that
  $(\depth{g_\iota}{n_\iota})_{\iota<\alpha}$ does not tend to
  infinity. That is, there is some $d \in \nats$ such that for each
  $\gamma < \alpha$ there is a $\gamma \le \iota < \alpha$ with
  $\depth{g_\iota}{n_\iota} \le d$. Let $d^*$ be the smallest such
  $d$. Hence, there is a $\beta < \alpha$ such that
  $\depth{g_\iota}{n_\iota} \ge d^*$ for all $\beta \le \iota <
  \alpha$. Thus we can apply Lemma~\ref{lem:redBound} to the suffix of
  $S$ starting from $\beta$ to obtain that $\posBound{d^*}{g_\beta} =
  \posBound{d^*}{g_\iota}$ for all $\beta \le \iota < \alpha$. Since
  we find for each $\gamma < \alpha$ some $\gamma \le \iota < \alpha$
  with $\depth{g_\iota}{n_\iota} \le d^*$, we know that for each
  $\gamma < \alpha$ there is a $\gamma \le \iota < \alpha$ and a $\pi
  \in \posBound{d^*}{g_\beta}$ with $\pi \in
  \nodePos{g_\iota}{n_\iota}$. Because $\posBound{d^*}{g_\beta}$ is
  finite, the infinite pigeon hole principle yields single
  $\pi^*\in\posBound{d^*}{g_\beta}$ such that for each $\gamma <
  \alpha$ there is a $\gamma \le \iota <\alpha$ with $\pi^*\in
  \nodePos{g_\iota}{n_\iota}$. That is, $\pi^*$ is volatile in $S$.
\end{proof}

\begin{lemma}
  \label{lem:strongPMConv}
  Let $S = (g_\iota \to[n_\iota] g_{\iota+1})_{\iota < \alpha}$ be an
  open reduction in a GRS strongly $\prs$-converging to $g$. If
  $(\depth{g_\iota}{n_\iota})_{\iota<\alpha}$ tends to infinity, then
  $g \isom \lim_{\iota \limto \alpha} g_\iota$.
\end{lemma}
\begin{proof}
  Let $h = \lim_{\iota \limto \alpha} g_\iota$ and let
  $(c_\iota)_{\iota<\alpha}$ be the reduction contexts of $S$. We will
  prove that $g \isom h$ by showing that their respective labelled
  quotient trees coincide.

  For the inclusion $\pos{g} \subseteq \pos{h}$, assume some $\pi \in
  \pos{g}$. According to Corollary~\ref{cor:lebot1Liminf}, there is
  some $\beta < \alpha$ such that $\pi \in \pos{c_\beta}$ and
  $c_\iota(\pi) = c_\beta(\pi)$ for all $\pi'<\pi$ and $\beta \le
  \iota < \alpha$. Thus, $\pi \in \pos{c_\iota}$ for all $\beta \le
  \iota < \alpha$. Since $c_\iota \isom \truncl{g_\iota}{n_\iota}$
  and, therefore, by Lemma~\ref{lem:locTrunc}, $\pos{c_\iota} =
  \pos{g_\iota}$, we have that $\pi \in \pos{g_\iota}$ for all $\beta
  \le \iota < \alpha$. This implies, by
  Theorem~\ref{thr:smetricComplete}, that $\pi \in \pos{h}$.

  For the converse inclusion $\pos{h}\subseteq\pos{g}$, assume some
  $\pi \in \pos{h}$. According to Theorem~\ref{thr:smetricComplete},
  there is some $\beta < \alpha$ such that $\pi \in \pos{g_\iota}$ for
  all $\beta \le \iota < \alpha$. Since $(\depth{g_\iota}{n_\iota})_{\iota<\alpha}$
  tends to infinity, we find some $\beta \le \gamma < \alpha$ such that
  $\depth{g_\iota}{n_\iota} \ge \len{\pi}$ for all $\gamma \le \iota
  < \alpha$, i.e.\ $\pi' \nin \nodePos{g_\iota}{n_\iota}$ for all
  $\pi'<\pi$. This means, by Lemma~\ref{lem:strConvRes},
  that $\pi \in \pos{g}$.

  By Lemma~\ref{lem:strConvRes} and Theorem~\ref{thr:smetricComplete},
  $\sim_g\ =\ \sim_h$ follows from the equality $\pos{g} = \pos{h}$.

  In order to show the equality $g(\cdot) = g(\cdot)$, assume some
  $\pi \in \pos{h}$. According to Theorem~\ref{thr:smetricComplete},
  there is some $\beta<\alpha$ such that $h(\pi) = g_{\iota}(\pi)$ for
  all $\beta \le \iota < \alpha$. Additionally, since
  $(\depth{g_\iota}{n_\iota})_{\iota<\alpha}$ tends to infinity, there is some
  $\beta\le\gamma<\alpha$ such that $\depth{g_\iota}{n_\iota} > \len\pi$ for
  all $\gamma \le \iota < \alpha$, i.e. $\pi \nin
  \nodePos{g_\iota}{n_\iota}$. Thus, by Lemma~\ref{lem:strConvRes},
  $g(\pi)=g_{\gamma}(\pi)$. Since $h(\pi) = g_{\gamma}(\pi)$,
  we can conclude that $g(\pi) = h(\pi)$.
\end{proof}

The following property, which relates strong $\mrs$-convergence and
-continuity, follows from the fact that our notion of strong
$\mrs$-convergence on term graphs instantiates the abstract model of
strong $\mrs$-convergence from our previous work~\cite{bahr10rta}:
\begin{lemma}
  \label{lem:mContConv}
  Let $S = (g_\iota \to[n_\iota] g_{\iota+1})_{\iota<\alpha}$ be an
  open strongly $\mrs$-continuous reduction in a GRS. If
  $(\depth{g_\iota}{n_\iota})_{\iota<\alpha}$ tends to infinity, then
  $S$ is strongly $\mrs$-convergent.
\end{lemma}
\begin{proof}
  This is a special case of Proposition~5.5 from \cite{bahr10rta}.
\end{proof}

Now, we have everything in place to prove that strong
$\prs$-convergence conservatively extends strong $\mrs$-convergence.
\begin{theorem}
  \label{thr:graphExt}
  Let $\calR$ be a GRS and $S$ a reduction in $\calR$. We then have
  that
  \[
  S\fcolon g \mato[\calR] h \qquad \text{iff} \qquad S\fcolon g
  \pato[\calR] h \text{ is total}.
  \]
\end{theorem}
\begin{proof}
  Let $S = (g_\iota \to[n_\iota] g_{\iota+1})_{\iota<\alpha}$. At first, we
  prove the ``only if'' direction by induction on $\alpha$:

  The case $\alpha = 0$ is trivial. If $\alpha$ is a successor
  ordinal, the statement follows immediately from the induction
  hypothesis.

  Let $\alpha$ be a limit ordinal. Since $S\fcolon g\mato g_\alpha$,
  we know that $\prefix{S}{\gamma}\fcolon g \mato g_\gamma$ for all
  $\gamma < \alpha$. Hence, we can apply the induction hypothesis in
  order to obtain that $\prefix{S}{\gamma}\fcolon g \pato g_\gamma$
  for each $\gamma < \alpha$. Consequently, $S$ is strongly
  $\prs$-continuous, which means, by
  Proposition~\ref{prop:strPContConv}, that $S$ strongly
  $\prs$-converges to some term graph $h'$. However, since $S$
  strongly $\mrs$-converges, we know that
  $(\depth{g_\iota}{n_\iota})_{\iota<\alpha}$ tends to
  infinity. Consequently, we can apply Lemma~\ref{lem:strongPMConv} to
  obtain that $h' = \lim_{\iota\limto\alpha} = h$, i.e.\ $S\fcolon g
  \pato h$. Since $S\fcolon g \mato h$, of course, $S\fcolon g \pato
  h$ must be total.

  We will also prove the ``if'' direction by induction on $\alpha$:
  Again, the case $\alpha = 0$ is trivial and the case that $\alpha$
  is a successor ordinal follows immediately from the induction
  hypothesis.

  Let $\alpha$ be a limit ordinal. Since $S$ is strongly
  $\prs$-convergent, we know that $\prefix{S}{\gamma}\fcolon g \pato
  g_\gamma$ is total for each $\gamma < \alpha$. Therefore, we can
  apply the induction hypothesis to obtain that
  $\prefix{S}{\gamma}\fcolon g \mato g_\gamma$ for each $\gamma <
  \alpha$. Hence, $S$ is strongly $\mrs$-continuous. Since $S$ is
  total, we know from Lemma~\ref{lem:volBot}, that $S$ has no volatile
  positions. Hence, by Lemma~\ref{lem:volDepthInf},
  $(\depth{g_\iota}{n_\iota})_{\iota<\alpha}$ tends to infinity. Together with the
  strong $\mrs$-continuity of $S$, this yields, according to
  Lemma~\ref{lem:mContConv}, that $S$ strongly $\mrs$-converges to
  some $h'$. With Lemma~\ref{lem:strongPMConv}, we can then conclude
  that $h' = h$, i.e.\ $S\fcolon g \mato h$.
\end{proof}

\section{Terms vs.\ Term Graphs}
\label{sec:terms-vs.-term}

Term graph rewriting is an efficient implementation technique for term
rewriting that uses pointers in order to avoid duplication. This is
also used as the basis for the implementation of functional
programming languages. A prominent example is the implementation of
the fixed point combinator $Y$ defined by the term rule $\rho_0\fcolon
Y\, x \to x\, (Y\, x)$, where we write function application as
juxtaposition. Written as a term graph rule $\rho_1\fcolon Y\, x \to
x\, (Y\, x)$ depicted in Figure~\ref{fig:fixedPointCombA}, we can see
that the two occurrences of the variables $x$ on the right-hand side
are shared. In fact, since term graph rewriting does not provide a
mechanism for duplication, this is the only way to represent non-right
linear rules. With the rule $\rho_1$ applied repeatedly as shown in
Figure~\ref{fig:fixedPointCombC}, more and more pointers to the same
occurrence of the function symbol $f$ are created. This reduction
strongly $\mrs$-converges to the infinite term graph $g_\omega = \nn n
f\,(n\,(n\,(\dots)))$, which has infinitely many edges to the
$f$-node. Note, however, that the term graph rule $\rho_1$ is not
maximally shared. If we apply the maximally shared rule $\rho_2\fcolon
\nn n {(Y\, x)} \to x\, n$, we obtain in one step the cyclic term
graph $h_0 = \nn n {(f\, n)}$. This is, in fact, how the fixed point
combinator is typically implemented in functional programming
languages~\cite{jones87book,turner79spe}. Although, the resulting term
graphs $g_\omega$ and $h_0$ are different, they both unravel to the
same term $f\,(f\,(f\,(\dots)))$.

\begin{figure}
  \centering%
  \subfloat[Term graph rules that unravel to $Y\, x \rightarrow x\, (Y\,
  x)$.]{%
    \label{fig:fixedPointCombA}%
    \begin{tikzpicture}%
      \node[node name=180:\underline{l}] (l) {$@$}%
      child {%
        node (Y) {$Y$}%
      } child {%
        node (f) {$x$}%
      };%
      \node[node distance=1.5cm,right=of l,node name=180:\underline{r}] (r) {$@$}%
      child[missing]%
      child {%
        node (a) {$@$}%
        child {%
          node {$Y$}%
        } child[missing]
      };%
      \draw%
      (r) edge[->,out=-135,in=35] (f)
      (a) edge[->,out=-45,in=-35] (f);%
      \node at ($(l)!.5!(r) + (0,-2.5)$) {$(\rho_1)$};
    \end{tikzpicture}
    \quad
    \begin{tikzpicture}
      \node[node name=180:\underline{l}] (l) {$@$}%
      child {%
        node (Y) {$Y$}%
      } child {%
        node (f) {$x$}%
      };%
      \node[node distance=1.5cm,right=of l,node name=180:\underline{r}] (r) {$@$};%
      \draw%
      (r) edge[->,out=-45,in=-15] (l);%
      \node at ($(l)!.5!(r) + (0,-2.5)$) {$(\rho_2)$};
      \draw%
      (r) edge[->,out=-135,in=35] (f);%
    \end{tikzpicture}%
  }%
  \quad
  \subfloat[A single $\rho_2$-step.]{%
    \label{fig:fixedPointCombB}%
    \begin{tikzpicture}
      \node (g1) {$@$}%
      child {%
        node {$Y$}%
      } child {%
        node {$f$}
        };%

      \node at ($(g1) + (0,-2.5)$) {$(g_0)$};      
      \node [node distance=2.5cm,right=of g1] (g2) {$@$}%
      child {%
        node {$f$}%
      } child [missing];%
      \draw (g2) edge[->, min distance=10mm,out=-45,in=0] (g2);%

      \node at ($(g2) + (0,-2.5)$) {$(h_0)$};
      \node (s1) at ($(g1)!.5!(g2)-(0,.5)$) {};

      \draw[single step] ($(s1)-(.3,0)$) -- ($(s1)+(.3,0)$)
      node[pos=1,below] {{\small$\rho_2$}};
    \end{tikzpicture}
  }%

    \subfloat[A strongly $\mrs$-convergent term graph reduction over
    $\rho_1$.]{%
    \label{fig:fixedPointCombC}%
      \begin{tikzpicture}  
      \node (g1) {$@$}%
      child {%
        node (a) {$Y$}%
      } child {%
        node {$f$}
      };%

      \node at ($(g1) + (0,-3.5)$) {$(g_0)$};
        
      \node [node distance=2.5cm,right=of g1] (g2) {$@$}%
      child {%
        node (f) {$f$}%
      } child {%
        node (c) {$@$}%
        child {%
          node (a) {$Y$}%
        } child [missing]%
      };%
      \draw (c) edge[->, min distance=7mm,out=-45,in=-45] (f);%

      \node () at ($(g2) + (0,-3.5)$) {$(g_1)$};
        
      \node [node distance=2.5cm,right=of g2] (g3) {$@$}%
      child {%
        node (f) {$f$}%
      } child {%
        node (c) {$@$}%
        child [missing]%
        child {%
          node (c2) {$@$}%
          child {%
            node (a) {$Y$}%
          } child [missing]%
        }%
      };%
      \draw (c) edge[->, min distance=3mm,out=-125,in=-45] (f);%
      \draw (c2) edge[->, min distance=10mm,out=-45,in=-75] (f);%

      \node () at ($(g3) + (0,-3.5)$) {$(g_2)$};
        
      \node [node distance=4.5cm,right=of g3] (go) {$@$}%
      child {%
        node (f) {$f$}%
      } child {%
        node (c) {$@$}%
        child [missing]%
        child {%
          node (c2) {$@$}%
          child [missing]%
          child[etc] {%
            node {} %
          }%
        }
      };%
      \draw (c) edge[->, min distance=3mm,out=-125,in=-45] (f);%
      \draw (c2) edge[->, min distance=10mm,out=-125,in=-75] (f);%

      \node () at ($(go) + (0,-3.5)$) {$(g_\omega)$};

    \node (s1) at ($(g1)!.5!(g2)-(0,.5)$) {};
    \node (s2) at ($(g2)!.55!(g3)-(0,.5)$) {};
    \node (s3) at ($(g3)!.4!(go)-(0,.5)$) {};
    \node (s4) at ($(g3)!.7!(go)-(0,.5)$) {};

    \draw[single step] ($(s1)-(.3,0)$) -- ($(s1)+(.3,0)$)
    node[pos=1,below] {{\small$\rho_1$}};
    \draw[single step] ($(s2)-(.3,0)$) -- ($(s2)+(.3,0)$)
    node[pos=1,below] {{\small$\rho_1$}};
    \draw[single step] ($(s3)-(.3,0)$) -- ($(s3)+(.3,0)$)
    node[pos=1,below] {{\small$\rho_1$}};
    \draw[dotted,thick,-] ($(s4)-(.3,0)$) -- ($(s4)+(.3,0)$);
        
    \end{tikzpicture}
  }%
  \caption{Implementation of the fixed point combinator as a term
    graph rewrite rule.}
  \label{fig:fixedPointComb}
\end{figure}
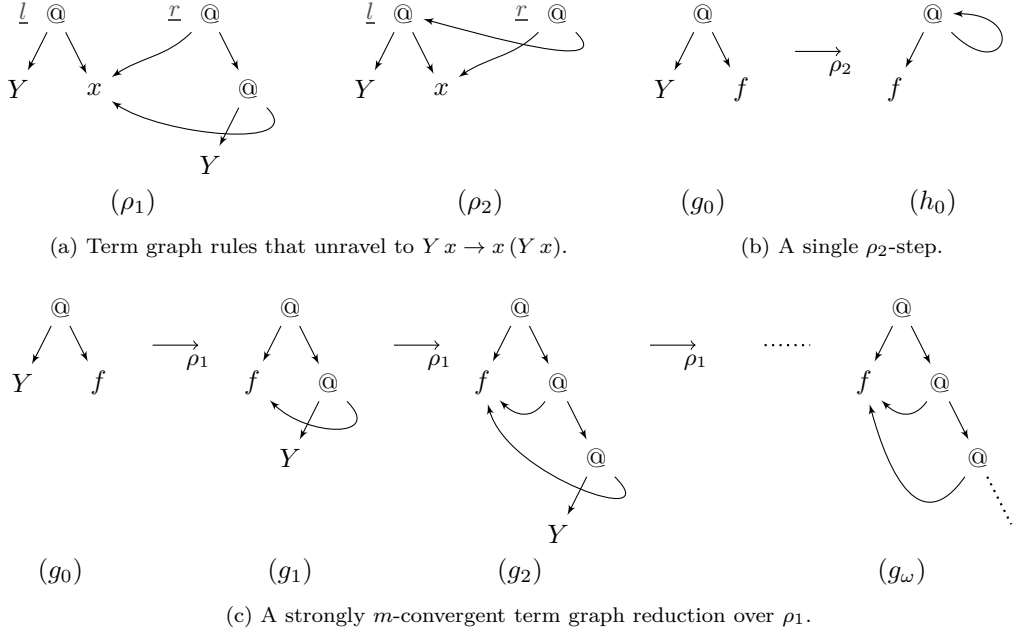

In this section, we will study the relationship between GRSs and the
corresponding TRSs they simulate. In particular, we will show the
soundness of GRSs w.r.t.\ strong convergence and a restricted form of
completeness. To this end we make use of the isomorphism between terms
and canonical term trees as outlined at the end of
Section~\ref{sec:canon-term-graphs}.

Note that term trees have an obvious characterisation in terms of
their equivalence on positions:
\begin{fact}
  \label{fact:labelledTree}
  A term graph $g \in \itgraphs$ is a term tree iff\, $\sim_g$ is the
  identity relation, i.e.\ $\pi_1 \sim_g \pi_2$ iff $\pi_1 = \pi_2$.
\end{fact}

When giving the labelled quotient tree for a term tree $t$ we can thus
omit the equivalence $\sim_t$. We refer to the remaining pair
$(\pos{t},t(\cdot))$ as \emph{labelled tree}.

Recall that the unravelling $\unrav g$ of a term graph $g$ is the
uniquely determined term $t$ such that there is a homomorphism from
$t$ to $g$. Labelled trees give a concrete characterisation of
unravellings:
\begin{proposition}
  \label{prop:unravTree}%
  The unravelling $\unrav{g}$ of a term graph $g \in \itgraphs$ is
  given by the labelled tree $(P,l)$ with $P = \pos{g}$ and $l(\pi) =
  g(\pi)$ for all $\pi \in P$.
\end{proposition}
\begin{proof}
  Since the implicit equivalence $\sim_{\unrav{g}}$ is reflexive and a
  subrelation of $\sim_g$, the triple $(P,l,\sim_{\unrav{g}})$ is a
  labelled quotient tree. Let $t$ be the term represented by
  $(P,l)$. By Lemma~\ref{lem:occrephom}, there is a homomorphism from
  $t$ to $g$. Thus, $\unrav{g} = t$.
\end{proof}

Before start investigating the correspondences between term rewriting
and term graph rewriting, we need to transfer the notions of
left-linearity and orthogonality to GRSs:
\begin{definition}[left-linearity, orthogonality
  \cite{barendregt87parle}]
  Let $\calR = (\Sigma,R)$ be a GRS.
  \begin{enumerate}[(i)]
  \item A rule $\rho \in R$ is called \emph{left-linear} if its
    left-hand side $\lhs\rho$ is a term tree. The GRS $\calR$ is
    called \emph{left-linear} if all its rules are left-linear.
  \item A $\rho$-redex $g$ and a $\rho'$-redex $g'$ in a common term
    graph, with matching $\calV$-homomorphisms $\phi$ resp.\ $\phi'$
    are \emph{disjoint}, if $r^g \neq \phi'(n)$ for all non-$\calV$
    nodes $n$ in $\lhs\rho'$ and, symmetrically, $r^{g'} \neq \phi(n)$
    for all non-$\calV$ nodes $n$ in $\lhs\rho$. In other words, the
    root of either redexes must not be matched by the respective other
    rule.
  \item The GRS $\calR$ is called \emph{non-overlapping} if all its redexes
    are pairwise disjoint.
  \item The GRS $\calR$ is called \emph{orthogonal} if it is
    left-linear and non-overlapping.
  \end{enumerate}
\end{definition}

It is obvious that the unravelling $\unrav\calR$ of a GRS is
left-linear if $\calR$ is left-linear, that and $\unrav\calR$ is
orthogonal if $\calR$ is orthogonal.

We have to single out a particular kind of term graph redexes that
manifest a peculiar behaviour.
\begin{definition}[circular redex]
  Let $\rho = (g,l,r)$ be a term graph rule. A $\rho$-redex is called
  \emph{circular} if $l$ and $r$ are distinct but the matching
  $\calV$-homomorphism $\phi$ maps them to the same node, i.e.\ $l\neq
  r$ but $\phi(l) = \phi(r)$.
\end{definition}

Kennaway et al.\ \cite{kennaway94toplas} show that circular redexes
only reduce to themselves:
\begin{proposition}
  \label{prop:circId}
  For every circular $\rho$-redex $\subgraph{g}{n}$, we have $g
  \to[\rho,n] h$.
\end{proposition}

However, contracting the unravellings of a circular redex also yields
the same term:
\begin{lemma}
  \label{lem:circUnravRed}
  Let $g$ be a term graph with a circular $\rho$-redex rooted in
  $n$. Then $\unrav g \to[\unrav\rho,\pi] \unrav g$ for all $\pi \in
  \nodePos{g}{n}$.
\end{lemma}
\begin{proof}
  Since there is a circular $\rho$-redex, we know that the right-hand
  side root $r^\rho$ is reachable from the left-hand side root
  $l^\rho$ of $\rho$. Let $\pi^*$ be a path from $l^\rho$ to
  $r^\rho$. Because $\subgraph{g}{n}$ is a circular redex, the
  corresponding matching $\calV$-homomorphism maps both $l^\rho$ and
  $r^\rho$ to $n$. Since $\Delta$-homomorphisms preserve paths, we
  thus know that $\pi^*$ is also a path from $n$ to itself in $g$. In
  other words $\pi\in \nodePos{g}{n}$ implies $\pi\concat\pi^*\in
  \nodePos{g}{n}$. Consequently, for each $\pi\in \nodePos{g}{n}$ we
  have that $\atPos{\unrav g}{\pi} = \atPos{\unrav
    g}{\pi\concat\pi^*}$.

  Since there is a path $\pi^*$ from $l^\rho$ to $r^\rho$, the
  unravelling $\unrav\rho$ of $\rho$ is of the form $l \to
  \atPos{l}{\pi^*}$. Hence, we know that each application of
  $\unrav\rho$ at a position $\pi$ in some term $t$ replaces the
  subterm at $\pi$ with the subterm at $\pi\concat\pi^*$ in $t$, i.e.\
  $t \to[\unrav\rho,\pi]
  \substAtPos{t}{\pi}{\atPos{t}{\pi\concat\pi^*}}$.

  Combining the two findings above, we obtain that
  \[
  \unrav{g} \to[\unrav\rho,\pi]
  \substAtPos{\unrav g}{\pi}{\atPos{\unrav g}{\pi\concat\pi^*}}
  = \substAtPos{\unrav g}{\pi}{\atPos{\unrav g}{\pi}}
  = \unrav g
  \quad \text{for all }\pi \in \nodePos{g}{n}
  \]
\end{proof}

The following two properties due to Kennaway et
al.~\cite{kennaway94toplas} show how single term graph rewrite steps
relate to term reductions in the corresponding unravelling.
\begin{proposition}
  \label{prop:unravNF}
  Given a left-linear GRS $\calR$ and a term graph $g$ in $\calR$, it
  holds that $g$ is a normal form in $\calR$ iff $\unrav{g}$ is a
  normal form in $\unrav{\calR}$.
\end{proposition}

\begin{theorem}
  \label{thr:graphStepTermConv}
  Let $\calR$ be a left-linear GRS with a reduction step $g
  \to[n,\rho] h$. Then $S\fcolon \unrav{g} \mato[\unrav\calR]
  \unrav{h}$ such that the depth of every redex reduced in $S$ is
  greater or equal to $\depth{g}{n}$. In particular, if the
  $\rho$-redex $\subgraph{g}{n}$ is not circular, then $S$ is a
  complete development of the set of redex occurrences
  $\nodePos{g}{n}$ in $\unrav g$.
\end{theorem}

The goal of the following two sections is to generalise the above
soundness theorem to strong $\mrs$- and $\prs$-convergence.

\subsection{Strong $\mrs$-Convergence}
\label{sec:metric-space}

At first we shall study correspondences w.r.t.\ strong
$\mrs$-convergence. To this end, we first show that the metric $\dda$
on term graphs generalises the metric $\dd$ on terms.

\begin{lemma}
  \label{lem:truncaTree}
  Let $t \in \ipterms$ and $d \in \nats \cup \set{\infty}$. The strict
  truncation $\trunca{t}{d}$ is given by the labelled tree $(P,l)$
  with
  \begin{center}
    \begin{inparaenum}[(a)]
    \item $P = \setcom{\pi \in \pos{t}}{\len{\pi} \le
        d}$,\label{item:truncaTreeI}
      \qquad
    \item $l(\pi) =
      \begin{cases}
        t(\pi) & \text{ if }\len{\pi} < d\\
        \bot & \text{ if }\len{\pi} \ge d
      \end{cases}$
      \label{item:truncaTreeII}
    \end{inparaenum}
  \end{center}
\end{lemma}
\begin{proof}
  Immediate from Lemma~\ref{lem:truncaOccRep} and
  Fact~\ref{fact:labelledTree}.
\end{proof}
This shows that the metric $\dda$ restricted to terms coincides with
the metric $\dd$ on terms. Moreover, we can use this in order to
relate the metric distance between term graphs and the metric distance
between their unravellings.
\begin{lemma}
  \label{lem:unravMetric}
  For all $g,h \in \itgraphs$, we have that $\dda(g,h) \ge
  \dda(\unrav{g},\unrav{h})$.
\end{lemma}
\begin{proof}
  Let $d = \similara{g}{h}$. Hence, $\trunca{g}{d} \isom
  \trunca{h}{d}$ and we can assume that the corresponding labelled
  quotient trees as characterised by Lemma~\ref{lem:truncaOccRep}
  coincide. We only need to show that $\trunca{\unrav{g}}{d} \isom
  \trunca{\unrav{h}}{d}$ since then $\similara{\unrav{g}}{\unrav{h}}
  \ge d$ and thus $\dda(\unrav{g},\unrav{h}) \le 2^{-d} = \dda(g,h)$. In order
  to show this, we show that the labelled trees of
  $\trunca{\unrav{g}}{d}$ and $\trunca{\unrav{h}}{d}$ as characterised
  by Lemma~\ref{lem:truncaTree} coincide. For the set of positions we
  have the following:
  \begin{align*}
    &\pi \in \pos{\trunca{\unrav{g}}{d}}\\%
    \iff &\pi \in \pos{\unrav{g}},\quad \len{\pi} \le d%
    \tag{Lemma~\ref{lem:truncaTree}}\\%
    \iff &\pi \in \pos{g},\quad \len{\pi} \le d%
    \tag{Proposition~\ref{prop:unravTree}}\\%
    \iff &\pi \in \pos{\trunca{g}{d}},\quad \len{\pi} \le d%
    \tag{Lemma~\ref{lem:truncaOccRep}}\\%
    \iff &\pi \in \pos{\trunca{h}{d}},\quad \len{\pi} \le d%
    \tag{$\trunca{g}{d} \isom \trunca{h}{d}$}\\%
    \iff &\pi \in \pos{h},\quad \len{\pi} \le d%
    \tag{Lemma~\ref{lem:truncaOccRep}}\\%
    \iff &\pi \in \pos{\unrav{h}},\quad\len{\pi} \le d%
    \tag{Proposition~\ref{prop:unravTree}}\\%
    \iff &\pi \in \pos{\trunca{\unrav{h}}{d}}%
    \tag{Lemma~\ref{lem:truncaTree}}
  \end{align*}
  In order to show that the labellings are equal, consider some $\pi
  \in \pos{\trunca{\unrav{g}}{d}}$ and assume at first that $\len{\pi}
  \ge d$. By Lemma~\ref{lem:truncaTree}, we then have $\left(\trunca{\unrav{g}}{d}\right)(\pi) = \bot =
  \left(\trunca{\unrav{h}}{d}\right)(\pi)$.
  Otherwise, if $\len{\pi} < d$, we obtain that
    
  \begin{align*}
    \left(\trunca{\unrav{g}}{d}\right)(\pi)%
    &\stackrel{\text{Lem.~\ref{lem:truncaTree}}}{=}%
    \unrav{g}(\pi)%
    \stackrel{\text{Prop.~\ref{prop:unravTree}}}{=}%
    g(\pi)%
    \stackrel{\text{Lem.~\ref{lem:truncaOccRep}}}{=}%
    \trunca{g}{d}(\pi)%
    \\%
    &\stackrel{\trunca{g}{d} \isom \trunca{h}{d}}{=}%
    \trunca{h}{d}(\pi)%
    \stackrel{\text{Lem.~\ref{lem:truncaOccRep}}}{=}%
    h(\pi)%
    \stackrel{\text{Prop.~\ref{prop:unravTree}}}{=}%
    \unrav{h}(\pi)%
    \stackrel{\text{Lem.~\ref{lem:truncaTree}}}{=}%
    \left(\trunca{\unrav{h}}{d}\right)(\pi)
  \end{align*}
\end{proof}

This immediately yields that Cauchy sequences are preserved by
unravelling:
\begin{lemma}
  \label{lem:unravCauchy}
  If $(g_\iota)_{\iota<\alpha}$ is a Cauchy sequence in
  $(\ictgraphs,\dda)$, then so is $(\unrav{g_\iota})_{\iota<\alpha}$.
\end{lemma}
\begin{proof}
  This follows immediately from Lemma~\ref{lem:unravMetric}.
\end{proof}

Additionally, also limits are preserved by unravellings.
\begin{proposition}
  \label{prop:unravLim}
  For every sequence $(g_\iota)_{\iota<\alpha}$ in $(\ictgraphs,\dda)$,
  we have that $\lim_{\iota\limto\alpha} g_\iota = g$ implies
  $\lim_{\iota\limto\alpha} \unrav{g_\iota} = \unrav g$
\end{proposition}
\begin{proof}
  According to Theorem~\ref{thr:smetricComplete}, we have that
  $\pos{g} = \liminf_{\iota\limto\alpha} \pos{g_\iota}$, and that
  $g(\pi) = g_\beta(\pi)$ for some $\beta < \alpha$ with $g_\iota(\pi)
  = g_\beta(\pi)$ for all $\beta \le \iota < \alpha$. By
  Proposition~\ref{prop:unravTree}, we then obtain $\pos{\unrav{g}} =
  \liminf_{\iota\limto\alpha} \pos{\unrav{g_\iota}}$, and that $\unrav
  g(\pi) = \unrav{g_\beta}(\pi)$ for some $\beta < \alpha$ with
  $\unrav{g_\iota}(\pi) = \unrav{g_\beta}(\pi)$ for all $\beta \le
  \iota < \alpha$. Since by Lemma~\ref{lem:unravCauchy},
  $(\unrav{g_\iota})_{\iota<\alpha}$ is Cauchy, we can apply
  Theorem~\ref{thr:smetricComplete} to obtain that
  $\lim_{\iota\limto\alpha} \unrav{g_\iota} = \unrav g$.
\end{proof}

We can now show that term graph reductions are sound w.r.t.\
reductions in the unravelled system.
\begin{theorem}
  \label{thr:mConvSound}
  Let $\calR$ be a left-linear GRS. If $g \mato[\calR] h$, then
  $\unrav{g} \mato[\unrav{\calR}] \unrav{h}$.
\end{theorem}
\begin{proof}
  Let $S = (g_\iota \to[d_\iota] g_{\iota+1})_{\iota<\alpha}$ be a
  reduction strongly $\mrs$-converging to $g_\alpha$ in $\calR$, i.e.\
  $S\fcolon g_0 \mato[\calR] g_\alpha$. According to
  Theorem~\ref{thr:graphStepTermConv}, there is for each $\iota <
  \alpha$ a reduction $T_\iota \fcolon \unrav{g_\iota}
  \mato[\unrav{\calR}] \unrav{g_{\iota+1}}$ such that
  \[
  \text{all steps in $T_\iota$ contract a redex at depth $\ge
    d_\iota$.}\tag{$*$}
  \label{eq:depth}
  \]
  Define for each $\delta \le \alpha$ the concatenation $U_\delta =
  \Concat_{\iota<\delta} T_\iota$. We will show that $U_\delta\fcolon
  \unrav{g_0} \mato[\unrav{\calR}] \unrav{g_\delta}$ for each $\delta
  \le \alpha$ by induction on $\delta$. The theorem is then obtained
  by instantiating $\delta = \alpha$.

  The case $\delta = 0$ is trivial. If $\delta = \delta' + 1$, then we
  have by induction hypothesis that $U_{\delta'}\fcolon \unrav{g_0}
  \mato[\unrav{\calR}] \unrav{g_{\delta'}}$. Since $T_{\delta'}\fcolon
  \unrav{g_{\delta'}} \mato[\unrav{\calR}] \unrav{g_{\delta}}$, and
  $U_\delta = U_{\delta'} \concat T_{\delta'}$, we have
  $U_\delta\fcolon \unrav{g_0} \mato[\unrav{\calR}] \unrav{g_\delta}$.

  For the case that $\delta$ is a limit ordinal, let $U_\delta =
  (t_\iota \to[e_\iota] t_{\iota+1})_{\iota<\beta}$. For each $\gamma
  < \beta$ we find some $\gamma' < \delta$ with
  $\prefix{U_\delta}{\gamma} < U_{\gamma'}$. By induction hypothesis,
  we can assume that $U_{\gamma'}$ is strongly
  $\mrs$-continuous. According to Proposition~\ref{prop:contConv},
  this means that the proper prefix $\prefix{U_\delta}{\gamma}$
  strongly $\mrs$-converges to $t_\gamma$. This shows that each proper
  prefix $\prefix{U_\delta}{\gamma}$ of $U_\delta$ strongly
  $\mrs$-converges to $t_\gamma$. Hence, by
  Proposition~\ref{prop:contConv}, $U_\delta$ is strongly
  $\mrs$-continuous.

  Since $S$ is strongly $\mrs$-convergent, $(d_\iota)_{\iota<\delta}$
  tends to infinity. By \eqref{eq:depth}, also
  $(e_\iota)_{\iota<\alpha}$ tends to infinity. Hence, $U_\delta$ is
  strongly $\mrs$-convergent according to
  Proposition~\ref{prop:strConvDepth}. Let $t$ be the term $U_\delta$
  is strongly $\mrs$-converging to, i.e.\ $\lim_{\iota \limto \beta}
  t_\iota = t$. Since $(\unrav{g_\iota})_{\iota<\delta}$ is a cofinal
  subsequence of $(t_\iota)_{\iota<\beta}$, we have by
  Proposition~\ref{prop:convSubseq} that $\lim_{\iota\limto\delta}
  \unrav{g_\iota} = t$. Since $S$ is strongly $\mrs$-convergent, we
  also have that $\lim_{\iota \limto \delta} g_\iota = g_\delta$. By
  Proposition~\ref{prop:unravLim}, this yields that $\lim_{\iota
    \limto \delta} \unrav{g_\iota} = \unrav{g_\delta}$. Consequently,
  we have that $t = \unrav{g_\delta}$, i.e.\ $U_\delta\fcolon
  \unrav{g_0} \mato[\unrav{\calR}] \unrav{g_\delta}$.
\end{proof}

Unfortunately, we will not be able to obtain a full completeness
result. Even the weak completeness that was considered by Kennaway et
al.~\cite{kennaway94toplas} for finitary term graph reductions does
not hold for infinitary term graph reductions. This weaker
completeness property is satisfied by infinitary term graph rewriting
iff
\[
\unrav g \mato[\unrav\calR] t \implies \text{ there is a term graph } h
\text{ with } g \mato[\calR]
h,\quad t \mato[\unrav\calR] \unrav{h}
\]
Kennaway et al.~\cite{kennaway94toplas} consider an informal notion of
infinitary term graph rewriting and give a counterexample for the
above weak completeness property. This counterexample also applies to
strongly $\mrs$-convergent term graph reductions:
\begin{example}
  \label{ex:complCounterEx}
  We consider an infinite alphabet $\Sigma$ with $\ul n \in
  \Sigma^{(2)}$ for each $n\in\nats$. Let $g$ be the term graph
  depicted in Figure~\ref{fig:complCounterExA}. The root node is
  labelled $\ul 0$ and each node labelled $\ul n$ has as its left
  successor itself and as its right successor a node labelled
  $\ul{n+1}$. Let $\calR$ be the GRS that for each natural number $n$
  has a rule $\rho_n\colon \ul n(x,y) \to \ul{n+1}(x,y)$. A single
  reduction step in $\calR$ increments the label of exactly one
  node. Figure~\ref{fig:complCounterExB} shows the unravelling of
  $g$. In each row of $\unrav g$, the rightmost node has the largest
  label. However, each node to its left can be incremented by
  performing finitely many reduction steps in the TRS $\unrav \calR$
  so that it has the same label as the rightmost node. Doing this for
  each row yields a reduction strongly $\mrs$-converging to the term
  $t$ depicted in Figure~\ref{fig:complCounterExC}. Note that for each
  $n\in\nats$, there are only finitely many occurrences of $\ul n$ in
  $t$. Therefore, also the number of occurrences of labels $\ul m$
  with $m < n$ is finite for each $n\in \nats$. Since it is only
  possible to obtain a node labelled $\ul n$ by repeated reduction on
  a node labelled $\ul m$ with $m<n$, this means that every term $t'$
  with $t \mato[\unrav\calR] t'$ also has only finitely many
  occurrences of $\ul n$ for any $n\in\nats$. On the other hand, there
  is no strongly $\mrs$-converging reduction from $g$ to a term graph
  that unravels to a term with finitely many occurrences of $\ul n$
  for each $n\in\nats$. This is because the structure of $g$ cannot be
  changed by reductions in $\calR$. In particular, the loops in $g$
  are maintained.
\end{example}

\begin{figure}
  \centering
  \subfloat[A term graph $g$.]{%
    \label{fig:complCounterExA}%
    \hspace{1cm}
    \begin{tikzpicture}
      \node (n0) {$\ul 0$}%
      child {%
        node (n1) {$\ul 1$}%
        child {%
          node (n2) {$\ul 2$}%
          child {%
            node (n3) {$\ul 3$}%
            child[etc] {%
              node {} %
            }%
          }%
        }%
      };%
      \foreach \n in {n0,n1,n2,n3} {%
        \draw (\n) edge[->, out=-125,in=180,min distance=5mm] (\n);%
      }%
    \end{tikzpicture}
    \hspace{1cm}
  }%
  \quad%
  \subfloat[The unravelling of $g$.]{%
    \label{fig:complCounterExB}%
    \begin{tikzpicture}[%
      level 1/.style={sibling distance=20mm},%
      level 2/.style={sibling distance=10mm},%
      level 3/.style={sibling distance=5mm},%
      ]%
      \node {$\ul 0$}%
      child {%
        node {$\ul 0$}%
        child {%
          node {$\ul 0$}%
          child {%
            node {$\ul 0$}%
            child [etc] {%
              node {}%
            } child [etc] { %
              node {}%
            }%
          } child {%
            node {$\ul 1$}%
            child [etc] {%
              node {}%
            } child [etc] { %
              node {}%
            }%
          }%
        } child {%
          node {$\ul 1$}%
          child {%
            node {$\ul 1$}%
            child [etc] {%
              node {}%
            } child [etc] { %
              node {}%
            }%
          } child { %
            node {$\ul 2$}%
            child [etc] {%
              node {}%
            } child [etc] { %
              node {}%
            }%
          }%
        }%
      } child {%
        node {$\ul 1$}%
        child {%
          node {$\ul 1$}%
          child {%
            node {$\ul 1$}%
            child [etc] {%
              node {}%
            } child [etc] { %
              node {}%
            }%
          } child { %
            node {$\ul 2$}%
            child [etc] {%
              node {}%
            } child [etc] { %
              node {}%
            }%
          }%
        } child { %
          node {$\ul 2$}%
          child {%
            node {$\ul 2$}%
            child [etc] {%
              node {}%
            } child [etc] { %
              node {}%
            }%
          } child { %
            node {$\ul 3$}%
            child [etc] {%
              node {}%
            } child [etc] { %
              node {}%
            }%
          }%
        }%
      };%
    \end{tikzpicture}
  }%
  \quad%
  \subfloat[Term reduct $t$ of $\unrav g$.]{%
    \label{fig:complCounterExC}%
    \begin{tikzpicture}[%
      level 1/.style={sibling distance=20mm},%
      level 2/.style={sibling distance=10mm},%
      level 3/.style={sibling distance=5mm},%
      ]%
      \node {$\ul 0$}%
      child {%
        node {$\ul 1$}%
        child {%
          node {$\ul 2$}%
          child {%
            node {$\ul 3$}%
            child [etc] {%
              node {}%
            } child [etc] { %
              node {}%
            }%
          } child {%
            node {$\ul 3$}%
            child [etc] {%
              node {}%
            } child [etc] { %
              node {}%
            }%
          }%
        } child {%
          node {$\ul 2$}%
          child {%
            node {$\ul 3$}%
            child [etc] {%
              node {}%
            } child [etc] { %
              node {}%
            }%
          } child { %
            node {$\ul 3$}%
            child [etc] {%
              node {}%
            } child [etc] { %
              node {}%
            }%
          }%
        }%
      } child {%
        node {$\ul 1$}%
        child {%
          node {$\ul 2$}%
          child {%
            node {$\ul 3$}%
            child [etc] {%
              node {}%
            } child [etc] { %
              node {}%
            }%
          } child { %
            node {$\ul 3$}%
            child [etc] {%
              node {}%
            } child [etc] { %
              node {}%
            }%
          }%
        } child { %
          node {$\ul 2$}%
          child {%
            node {$\ul 3$}%
            child [etc] {%
              node {}%
            } child [etc] { %
              node {}%
            }%
          } child { %
            node {$\ul 3$}%
            child [etc] {%
              node {}%
            } child [etc] { %
              node {}%
            }%
          }%
        }%
      };%
    \end{tikzpicture}
  }%
  \caption{Counterexample for weak completeness.}
  \label{fig:complCounterEx}
\end{figure}
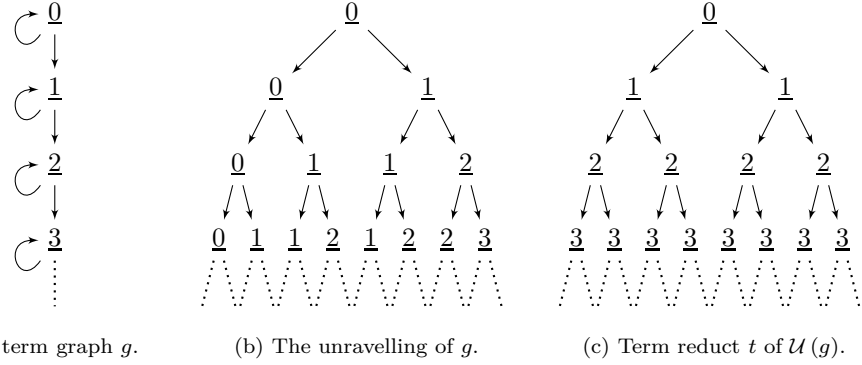
The above counterexample requires an infinite set of function symbols
and rules. Kennaway et al.~\cite{kennaway94toplas} sketch a variant of
this example that gets along with only two function symbols and one
rule. However, after closer inspection one can see that this system is
not a counterexample! We do not know whether a restriction to finitely
many rules may in fact yield weak completeness of infinitary term
graph rewriting.

We think, however, that a completeness property w.r.t.\ normalising
reductions can be obtained. To this end, consider the following
property of strong $\mrs$-convergence in TRSs:
\begin{theorem}[\cite{kennaway95ic}]
  \label{thr:nfMConv}
  Every orthogonal TRS has the normal form property w.r.t.\ strong
  $\mrs$-convergence. That is, for each term $t$ with $t \mato t_1$
  and $t \mato t_2$, we have $t_1 = t_2$, whenever $t_1,t_2$ are
  normal forms.
\end{theorem}

We then obtain as a corollary that a term graph has the same normal
forms as its unravelling, provided it has one:
\begin{corollary}
  For every orthogonal GRS $\calR$, we have that two normalising
  reductions $g \mato[\calR] h$ and $\unrav g \mato[\unrav\calR] t$ imply
  that $t = \unrav{h}$.
\end{corollary}
\begin{proof}
  According to Theorem~\ref{thr:mConvSound}, $g \mato[\calR] h$
  implies $\unrav{g} \mato[\unrav\calR] \unrav{h}$. Note that, by
  Proposition~\ref{prop:unravNF}, $\unrav{h}$ is a normal form in
  $\unrav{\calR}$. Hence, the reduction $\unrav{g} \mato[\unrav\calR]
  \unrav{h}$ together with $\unrav g \mato[\unrav\calR] t$ implies $t =
  \unrav{h}$ according to Theorem~\ref{thr:nfMConv}.
\end{proof}

We conjecture that this can be generalised such that infinitary term
graph rewriting is complete w.r.t.\ normalising reductions. That is,
if $\unrav g \mato[\unrav\calR] t$ with $t$ a normal form of
$\unrav\calR$, then there is a term graph $h$ with $\unrav h =t$ and
$g \mato[\calR] h$.

\subsection{Strong $\prs$-Convergence}
\label{sec:common-partial-order}

In this section, we replicate the results that we have obtained in the
preceding section for the case of strong $\prs$-convergence. Since
$\prs$-convergence is a conservative extension of $\mrs$-convergence,
cf.\ Theorems~\ref{thr:strongExt} and \ref{thr:graphExt}, this will in
fact generalise the soundness and completeness results for infinitary
term graph rewriting.

At first we derive a characterisation of the partial order $\lebotg$
on terms:
\begin{lemma}
  \label{lem:poTrees}
  Given two terms $s,t \in \ipterms$, we have $s \lebotg t$ iff
  $s(\pi) = t(\pi)$ for all $\pi \in \pos{s}$ with $g(\pi) \in \Sigma$.
\end{lemma}
\begin{proof}
  Immediate from Corollary~\ref{cor:chaTgraphPoA} and Fact~\ref{fact:labelledTree}.
\end{proof}
This shows that the partial order $\lebotg$ on term graphs generalises
the partial order on terms.

From this we easily obtain that the partial order $\lebotg$ as well as
its induced limits are preserved by unravelling:
\begin{theorem}
  \label{thr:poUnravel}
  In the partially ordered set $(\ipctgraphs,\lebotg)$ the following
  holds:
  \begin{enumerate}[(i)]
  \item Given two term graphs $g,h$, we have that $g \lebotg h$
    implies $\unrav{g} \lebotg \unrav{h}$.
  \item For each directed set $G$, we have
    that $\unrav{\Lub_{g\in G} g} = \Lub_{g\in G} \unrav g$.
  \item For each non-empty set $G$, we have that $\unrav{\Glb_{g\in G}
      g} = \Glb_{g\in G} \unrav g$.
  \item For each sequence $(g_\iota)_{\iota < \alpha}$, we have that
    $\unrav{\liminf_{\iota\limto\alpha}g_\iota} =
    \liminf_{\iota\limto\alpha}\unrav{g_\iota}$.
  \end{enumerate}
\end{theorem}
\begin{proof}
  (i) By Corollary~\ref{cor:chaTgraphPoA}, $g \lebotg h$ implies that
  $g(\pi) = h(\pi)$ for all $\pi \in \pos{g}$ with $g(\pi) \in
  \Sigma$. By Proposition~\ref{prop:unravTree}, we then have $\unrav
  g(\pi) = \unrav h(\pi)$ for all $\pi \in \pos{\unrav g}$ with
  $\unrav g(\pi) \in \Sigma$ which, by Lemma~\ref{lem:poTrees},
  implies $\unrav g \lebotg \unrav h$. 
  
  By a similar argument (ii) and (iii) follow from the
  characterisation of least upper bounds and greatest lower bounds in
  Theorem~\ref{thr:lebot1cpo} resp.\ Proposition~\ref{prop:lebot1glb}
  by using Proposition~\ref{prop:unravTree}.

  (iv) Follows from (ii) and (iii).
\end{proof}

In order to proof the soundness w.r.t.\ strong $\prs$-convergence we
need a stronger variant of Theorem~\ref{thr:graphStepTermConv} that
does not only make a statement about the depth of the redexes
contracted in the term reduction, but also the corresponding reduction
contexts. 

\begin{theorem}
  \label{thr:graphStepTermCxt}
  Let $\calR$ be a left-linear GRS with a reduction step $g \to[c]
  h$. Then there is a non-empty reduction $S = (t_\iota \to[c_\iota]
  t_{\iota+1})_{\iota<\alpha}$ with $S\fcolon \unrav{g}
  \pato[\unrav{\calR}] \unrav{h}$ such that $\unrav{c} =
  \Glb_{\iota<\alpha} c_\iota$.
\end{theorem}
\begin{proof}
  By Theorem~\ref{thr:graphStepTermConv}, there is a reduction
  $S\fcolon \unrav{g} \mato[\unrav{\calR}] \unrav{h}$. At first we
  assume that the redex $\subgraph{g}{n}$ contracted in $g \to[n] h$
  is not a circular redex. Hence $S$, is complete development of the
  set of redex occurrences $\nodePos{g}{n}$ in $\unrav g$. By
  Theorem~\ref{thr:strongExt}, we then obtain $S\fcolon \unrav{g}
  \pato[\unrav{\calR}] \unrav{h}$. From Lemma~\ref{lem:locTrunc} and
  Proposition~\ref{prop:unravTree} it follows that
  $\unrav{\truncl{g}{n}}$ is obtained from $\unrav{g}$ by replacing
  each subterm of $\unrav{g}$ at a position in $\nodePosMin{g}{n}$,
  i.e.\ a minimal position of $n$, by $\bot$. Since each step $t_\iota
  \to[\pi_\iota] t_{\iota+1}$ in $S$ contracts a redex at a position
  $\pi_\iota$ that has a prefix in $\nodePosMin{g}{n}$, we have
  $\unrav{\truncl{g}{n}} \lebotg \substAtPos{t_\iota}{\pi_\iota}{\bot}
  = c_\iota$. Moreover, for each $\pi\in \nodePosMin{g}{n}$ there is a
  step at $\iota(\pi) < \alpha$ in $S$ that takes place at $\pi$. From
  Proposition~\ref{prop:lebot1glb}, it is thus clear that
  $\unrav{\truncl{g}{n}} = \Glb_{\pi \in P} c_{\iota(\pi)}$. Together
  with $\unrav{\truncl{g}{n}} \lebotg c_\iota$ for all $\iota <
  \alpha$, this yields $\unrav{\truncl{g}{n}} = \Glb_{\iota < \alpha}
  c_{\iota}$. Then $\unrav{c} = \Glb_{\iota < \alpha} c_{\iota}$
  follows from the fact that $c \isom \truncl{g}{n}$.

  If the $\rho$-redex $\subgraph{g}{n}$ contracted in $g \to[\rho,n]
  h$ is a circular redex, then $g = h$ according to
  Proposition~\ref{prop:circId}. However, by
  Lemma~\ref{lem:circUnravRed}, each $\unrav{\rho}$-redex at positions
  in $\nodePos{g}{n}$ in $\unrav g$ reduces to itself as well. Hence,
  we get a reduction $\unrav{g} \pato[\unrav{\rho}] \unrav{h}$ via a
  complete development of the redexes at the minimal positions
  $\nodePosMin{g}{n}$. The equality $\unrav{c} = \Glb_{\iota < \alpha}
  c_{\iota}$ then follows as for the first case above.
\end{proof}

Before we prove the soundness of strongly $\prs$-converging term graph
reductions, we show the following technical lemma:

\begin{lemma}
  \label{lem:liminfSegm}
  Let $(a_\iota)_{\iota<\alpha}$ be a sequence in a complete
  semilattice $(A,\le)$ and $(\gamma_\iota)_{\iota<\delta}$ a strictly
  monotone sequence in the ordinal $\alpha$ such that
  $\Lub_{\iota<\delta}\gamma_\iota = \alpha$. Then
  \[
  \liminf_{\iota\limto\alpha} a_\iota =
  \liminf_{\beta\limto\delta}\left(\Glb_{\gamma_\beta \le \iota <
    \gamma_{\beta+1}} a_\iota\right).
  \]
\end{lemma}
\begin{proof}
  At first we show that
  \begin{equation}
    \Glb_{\beta\le\beta'<\delta}\Glb_{\gamma_{\beta'}\le \iota <
      \gamma_{\beta'+1}} a_\iota = \Glb_{\gamma_\beta \le \iota <
      \alpha} a_\iota\qquad\text{for all } \beta<\delta
    \tag{$*$}
    \label{eq:liminfSegm}
  \end{equation}
  by using the antisymmetry of the partial order $\le$ on $A$.
  
  Since for all $\beta \le \beta' < \delta$, we have that
  $\Glb_{\gamma_{\beta'}\le \iota<\gamma_{\beta'+1}} a_\iota \ge
  \Glb_{\gamma_\beta \le \iota < \alpha} a_\iota$, we obtain that
  $\Glb_{\beta\le\beta'<\delta}\Glb_{\gamma_{\beta'}\le \iota<\gamma_{\beta'+1}} a_\iota \ge
  \Glb_{\gamma_\beta \le \iota < \alpha} a_\iota$
  
  On the other hand, since $(\gamma_\iota)_{\iota<\delta}$ is strictly
  monotone and $\Lub_{\iota<\delta} \gamma_\iota = \alpha$, we find
  for each $\gamma_\beta \le \gamma < \alpha$ some $\beta\le \beta' <
  \delta$ such that $\gamma_{\beta'} \le \gamma < \gamma_{\beta'+1}$
  and, thus, $\Glb_{\gamma_{\beta'}\le \iota < \gamma_{\beta'+1}}
  a_\iota \le a_\gamma$. Consequently, we obtain that
  $\Glb_{\beta\le\beta'<\delta}\Glb_{\gamma_{\beta'}\le \iota <
    \gamma_{\beta'+1}} a_\iota \le \Glb_{\gamma_\beta\le\iota<\alpha}
  a_{\iota}$.
  
  With the thus obtained equation \eqref{eq:liminfSegm}, it remains to
  be shown that $\Lub_{\beta<\alpha}\Glb_{\beta\le\iota<\alpha}
  a_\iota= \Lub_{\beta'<\delta}\Glb_{\gamma_\beta'\le\iota<\alpha}
  a_\iota$. Again, we use the antisymmetry of $\le$.

  Since $\Glb_{\iota< \delta}\gamma_\iota = \alpha$, we find for each
  $\beta < \alpha$ some $\beta' < \delta$ with $\gamma_{\beta'} \ge
  \beta$. Consequently, we have for each $\beta < \alpha$ some $\beta'
  < \delta$ with $\Glb_{\beta\le\iota<\alpha} a_\iota \le
  \Glb_{\gamma_\beta'\le\iota<\alpha} a_\iota$. Hence
  $\Lub_{\beta<\alpha}\Glb_{\beta\le\iota<\alpha} a_\iota \le
  \Lub_{\beta'<\delta}\Glb_{\gamma_\beta'\le\iota<\alpha} a_\iota$.

  On the other hand, since for each $\beta' < \delta$ there is a
  $\beta < \alpha$ (namely $\beta = \gamma_\beta$) with
  $\Glb_{\beta\le\iota<\alpha} a_\iota =
  \Glb_{\gamma_\beta'\le\iota<\alpha} a_\iota$, we also have
  $\Lub_{\beta<\alpha}\Glb_{\beta\le\iota<\alpha} a_\iota \ge
  \Lub_{\beta'<\delta}\Glb_{\gamma_\beta'\le\iota<\alpha} a_\iota$.
\end{proof}
\begin{theorem}
  \label{thr:pConvSound}
  Let $\calR$ be a left-linear GRS. If $g \pato[\calR] h$, then
  $\unrav{g} \pato[\unrav{\calR}] \unrav{h}$.
\end{theorem}
\begin{proof}
  Let $S = (g_\iota \to[c_\iota] g_{\iota+1})_{\iota<\alpha}$ be a
  reduction strongly $\prs$-converging to $g_\alpha$ in $\calR$, i.e.\
  $S\fcolon g_0 \pato[\calR] g_\alpha$. According to
  Theorem~\ref{thr:graphStepTermCxt}, there is for each $\gamma <
  \alpha$ a strongly $\prs$-converging reduction $T_\gamma \fcolon
  \unrav{g_\gamma} \pato[\unrav{\calR}] \unrav{g_{\gamma+1}}$ such
  that
  \[
  \text{$\Glb_{\iota<\len{T_\gamma}}c'_\iota = \unrav{c_\gamma}$ for
    $(c'_\iota)_{\iota<\len{T_\gamma}}$ the reduction contexts in
    $T_\gamma$.}\tag{$*$}
  \label{eq:cxt}
  \]
  Define for each $\delta \le \alpha$ the concatenation $U_\delta =
  \Concat_{\iota<\delta} T_\iota$. We will show that $U_\delta\fcolon
  \unrav{g_0} \pato[\unrav{\calR}] \unrav{g_\delta}$ for each $\delta
  \le \alpha$ by induction on $\delta$. The theorem is then obtained
  for the case $\delta = \alpha$.

  The case $\delta = 0$ is trivial. If $\delta = \delta' + 1$, then we
  have by induction hypothesis that $U_{\delta'}\fcolon \unrav{g_0}
  \pato[\unrav{\calR}] \unrav{g_{\delta'}}$. Since $T_{\delta'}\fcolon
  \unrav{g_{\delta'}} \pato[\unrav{\calR}] \unrav{g_{\delta}}$, and
  $U_\delta = U_{\delta'} \concat T_{\delta'}$, we have
  $U_\delta\fcolon \unrav{g_0} \pato[\unrav{\calR}] \unrav{g_\delta}$.

  For the case that $\delta$ is a limit ordinal, let $U_\delta =
  (t_\iota \to[c'_\iota] t_{\iota+1})_{\iota<\beta}$. For each $\gamma
  < \beta$ we find some $\delta' < \delta$ with
  $\prefix{U_\delta}{\gamma} < U_{\delta'}$. By induction hypothesis,
  we can assume that $U_{\delta'}$ is strongly
  $\prs$-continuous. According to Proposition~\ref{prop:contConv},
  this means that the proper prefix $\prefix{U_\delta}{\gamma}$
  strongly $\prs$-converges to $t_\gamma$. This shows that each proper
  prefix $\prefix{U_\delta}{\gamma}$ of $U_\delta$ strongly
  $\prs$-converges to $t_\gamma$. Hence, by
  Proposition~\ref{prop:contConv}, $U_\delta$ is strongly
  $\prs$-continuous.

  In order to show that $U_\delta\fcolon \unrav{g_0}
  \pato[\unrav{\calR}] \unrav{g_\delta}$, it remains to be shown that
  $\liminf_{\iota \limto \beta} c'_\iota = \unrav{g_\delta}$. Since
  $S$ is strongly $\prs$-converging, we know that
  $\liminf_{\iota\limto\delta} c_\iota = g_\delta$. By
  Theorem~\ref{thr:poUnravel}, we thus have
  $\liminf_{\iota\limto\delta} \unrav{c_\iota} = \unrav{g_\delta}$. By
  \eqref{eq:cxt} and the construction of $U_\delta$, the sequence of
  reduction contexts $(c'_\iota)_{\iota<\beta}$ consists of segments
  whose glb is the unravelling of a corresponding reduction context
  $c_\gamma$. More precisely, there is a strictly monotone sequence
  $(\gamma_\iota)_{\iota<\delta}$ with $\gamma_0=0$ and
  $\Lub_{\iota<\delta}\gamma_\iota=\beta$ such that $\unrav{c_\iota} =
  \Glb_{\gamma_\iota \le \gamma < \gamma_{\iota+1}} c'_\gamma$ for all
  $\iota<\delta$. Thus, we can complete the proof as follows:
  \[
  \unrav{g_\delta} = \liminf_{\iota\limto\delta} \unrav{c_\iota} =
  \liminf_{\iota\limto\delta} \Glb_{\gamma_\iota \le \gamma <
    \gamma_{\iota+1}} c'_\gamma
  \stackrel{\text{Lem.~\ref{lem:liminfSegm}}}= \liminf_{\iota \limto
    \beta} c'_\iota
  \]
\end{proof}

Note that the counterexample from Example~\ref{ex:complCounterEx} is
not applicable to strong $\prs$-convergence. Since in the considered
system every term graph resp.\ every term is a redex, we can reduce
every term graph resp.\ every term to $\bot$ by an infinite strongly
$\prs$-converging reduction. We therefore conjecture that the weak
completeness property does hold for strongly $\prs$-convergent term
graph reductions. That is, for every $\unrav{g} \pato[\unrav\calR] t$
there is some $h$ with $\unrav h = t$ such that $g \pato[\calR] h$.

However, we can use the confluence of strongly-$\prs$-converging term
reductions in order to obtain a weak form of completeness for
normalising reductions.
\begin{theorem}[\cite{bahr10rta2}]
  \label{thr:pConvCR}
  Every orthogonal term rewriting system is confluent w.r.t.\ strong
  $\prs$-convergence. That is, $t \pato t_1$ and $t \pato t_2$ implies
  $t_1 \pato t'$ and $t_2 \pato t'$.
\end{theorem}

\begin{corollary}
  For every orthogonal GRS $\calR$, we have that two normalising
  reductions $g \pato[\calR] g'$ and $t \pato[\unrav\calR] t'$ imply
  that $t' = \unrav{g'}$.
\end{corollary}
\begin{proof}
  According to Theorem~\ref{thr:pConvSound}, $g \pato[\calR] g'$
  implies $\unrav{g} \pato[\unrav\calR] \unrav{g'}$. Note that, by
  Proposition~\ref{prop:unravNF}, $\unrav{g'}$ is a normal form in
  $\unrav{\calR}$. Hence, the reduction $\unrav{g} \pato[\unrav\calR]
  \unrav{g'}$ together with $t \mato[\unrav\calR] t'$ implies $t' =
  \unrav{g'}$ according to Theorem~\ref{thr:pConvCR}.
\end{proof}

\section{Discussion}
\label{sec:discussion}

The main contribution of this paper is the establishment of an
appropriate calculus of infinitary term graph rewriting. We have shown
that strong $\mrs$-convergence as well as its conservative extension
in the form of strong $\prs$-convergence provide an adequate
theoretical underpinning of such a calculus. The simplicity of the
underlying metric resp.\ partial order structure of term graphs
contrasts the intricate structures that we have proposed in our
previous work~\cite{bahr11rta}. There, we have been focused
exclusively on weak convergence and the peculiar properties of weak
convergence made it necessary to carefully define the underlying
structures to be quite rigid. As a consequence, a number of
intuitively converging term graph reductions do not converge in that
calculus. An example is the reduction illustrated in
Figure~\ref{fig:fixedPointCombC}, which in the rigid calculus does not
$\mrs$-converge at all and $\prs$-converges only to the partial term
graph $\nn n \bot\,(n\,(n\,(\dots)))$. In the calculus that we have
presented in this paper, this term graph reduction strongly $\mrs$-
and thus $\prs$-converges to the term graph $\nn n
f\,(n\,(n\,(\dots)))$ depicted in Figure~\ref{fig:fixedPointCombC}.

The new approach that we have presented in this paper -- built upon
simple generalisations of the metric resp.\ the partial order on terms
to term graphs -- is less rigid and captures an intuitive notion of
convergence in the form of strong convergence. We have argued for its
appropriateness by independently developing two modes of convergence
-- $\mrs$- and $\prs$-convergence -- and showing that both yield the
same limits when restricted to total term graphs. Moreover, we have
shown the adequacy of our infinitary calculus by establishing its
soundness w.r.t.\ infinitary term rewriting.

We have also made the first steps towards a completeness result by
showing that normalising reductions of term graph rewriting systems
and their corresponding term rewriting systems are equivalent modulo
unravelling. We conjecture that this can be extended to a full
completeness property of normalising reductions.


\begin{thebibliography}{10}

\bibitem{ariola02apal}
Z.M. Ariola and S.~Blom.
\newblock {Skew confluence and the lambda calculus with letrec}.
\newblock {\em Annals of Pure and Applied Logic}, 117(1-3):95--168, 2002.

\bibitem{ariola97ic}
Z.M. Ariola and J.W. Klop.
\newblock Lambda calculus with explicit recursion,.
\newblock {\em Information and Computation}, 139(2):154 -- 233, 1997.

\bibitem{arnold80fi}
A.~Arnold and M.~Nivat.
\newblock The metric space of infinite trees. {A}lgebraic and topological
  properties.
\newblock {\em Fundamenta Informaticae}, 3(4):445--476, 1980.

\bibitem{bahr10rta}
P.~Bahr.
\newblock Abstract models of transfinite reductions.
\newblock In C.~Lynch, editor, {\em RTA 2010}, volume~6 of {\em LIPIcs}, pages
  49--66. Schloss Dagstuhl--Leibniz-Zentrum für Informatik, 2010.

\bibitem{bahr10rta2}
P.~Bahr.
\newblock Partial order infinitary term rewriting and b{\"o}hm trees.
\newblock In C.~Lynch, editor, {\em RTA 2010}, volume~6 of {\em LIPIcs}, pages
  67--84. Schloss Dagstuhl--Leibniz-Zentrum für Informatik, 2010.

\bibitem{bahr11rta}
Patrick Bahr.
\newblock {Modes of Convergence for Term Graph Rewriting}.
\newblock In Manfred Schmidt-Schau{\ss}, editor, {\em 22nd International
  Conference on Rewriting Techniques and Applications (RTA'11)}, volume~10 of
  {\em Leibniz International Proceedings in Informatics (LIPIcs)}, pages
  139--154, Dagstuhl, Germany, 2011. Schloss Dagstuhl--Leibniz-Zentrum fuer
  Informatik.

\bibitem{barendregt87parle}
H.P. Barendregt, M.C.J.D. van Eekelen, J.R.W. Glauert, R.~Kennaway, M.J.
  Plasmeijer, and M.R. Sleep.
\newblock Term graph rewriting.
\newblock In Philip C.~Treleaven Jaco~de Bakker, A. J.~Nijman, editor, {\em
  PARLE 1987}, volume 259 of {\em LNCS}, pages 141--158. Springer, 1987.

\bibitem{blom04rta}
S.~Blom.
\newblock An approximation based approach to infinitary lambda calculi.
\newblock In Vincent van Oostrom, editor, {\em RTA 2004}, volume 3091 of {\em
  LNCS}, pages 221--232. Springer, 2004.

\bibitem{courcelle83tcs}
B.~Courcelle.
\newblock Fundamental properties of infinite trees.
\newblock {\em Theoretical Computer Science}, 25(2):95--169, 1983.

\bibitem{goguen77jacm}
J.A. Goguen, J.W. Thatcher, E.G. Wagner, and J.B. Wright.
\newblock Initial algebra semantics and continuous algebras.
\newblock {\em Journal of the ACM}, 24(1):68--95, 1977.

\bibitem{kelley55book}
J.L. Kelley.
\newblock {\em General Topology}, volume~27 of {\em Graduate Texts in
  Mathematics}.
\newblock Springer-Verlag, 1955.

\bibitem{kennaway92rep}
R.~Kennaway.
\newblock On transfinite abstract reduction systems.
\newblock Technical report, CWI (Centre for Mathematics and Computer Science),
  Amsterdam, 1992.

\bibitem{kennaway95segragra}
R.~Kennaway.
\newblock Infinitary rewriting and cyclic graphs.
\newblock {\em Electronic Notes in Theoretical Computer Science}, 2:153--166,
  1995.
\newblock SEGRAGRA '95.

\bibitem{kennaway03book}
R.~Kennaway and F.-J. de~Vries.
\newblock Infinitary rewriting.
\newblock In Terese \cite{terese03book}, chapter~12, pages 668--711.

\bibitem{kennaway94toplas}
R.~Kennaway, J.W. Klop, M.R. Sleep, and F.-J. de~Vries.
\newblock On the adequacy of graph rewriting for simulating term rewriting.
\newblock {\em ACM Transactions on Programming Languages and Systems},
  16(3):493--523, 1994.

\bibitem{kennaway95ic}
Richard Kennaway, Jan~Willem Klop, M~Ronan Sleep, and Fer-Jan de~Vries.
\newblock {Transfinite Reductions in Orthogonal Term Rewriting Systems}.
\newblock {\em Information and Computation}, 119(1):18--38, 1995.

\bibitem{jones87book}
Simon Peyton-Jones.
\newblock {\em The Implementation of Functional Programming Languages}.
\newblock Prentice Hall, 1987.

\bibitem{plump99hggcbgt}
D.~Plump.
\newblock Term graph rewriting.
\newblock In Hartmut Ehrig, Gregor Engels, Hans-Jörg Kreowski, and Grzegorz
  Rozenberg, editors, {\em Handbook of Graph Grammars and Computing by Graph
  Transformation}, volume~2, pages 3--61. World Scientific Publishing Co.,
  Inc., 1999.

\bibitem{terese03book}
Terese.
\newblock {\em Term Rewriting Systems}.
\newblock Cambridge University Press, 1st edition, 2003.

\bibitem{turner79spe}
D.A. Turner.
\newblock A new implementation technique for applicative languages.
\newblock {\em Software: Practice and Experience}, 9(1):31--49, 1979.

\end{thebibliography}
\end{document}